\documentclass[11pt,letterpaper]{amsart}

\usepackage{soul}
\usepackage[utf8]{inputenc}
\usepackage[]{mdframed}
\usepackage{mathrsfs}
\usepackage[noerroretextools,backend=biber,style=alphabetic,natbib=false,giveninits=true,doi=true,isbn=false,url=false,date=year,maxbibnames=99,sorting=nty,defernumbers=true]{biblatex}
\DeclareNameAlias{default}{family-given}

\DeclareFieldFormat[article]{title}{#1}
\renewbibmacro{in:}{}
\DeclareFieldFormat[article]{pages}{#1}
\DeclareFieldFormat{journal}{#1}

\renewbibmacro*{volume+number+eid}{%
	\printfield{volume}%
	\setunit*{\addnbspace}
	\printfield{number}%
	\setunit{\addcomma\space}%
	\printfield{eid}}
\DeclareFieldFormat[article]{number}{\mkbibparens{#1}}
\DeclareFieldFormat[article]{volume}{\textbf{#1}}
\DeclareFieldFormat{year}{\mkbibparens{#1}}
\DeclareBibliographyDriver{article}{%
	\printnames{author}:%
	\newunit\newblock
	\printfield{title}%
	\newunit\newblock
	\printfield{journaltitle}
	\newunit
	\iffieldundef{number}{\printfield{volume}}{\printfield{volume}\addspace\printfield{number}}
	\addcomma\addspace\printfield{pages}\addspace
	\printfield{year}
}
\AtEveryBibitem{\clearfield{month}}
\AtEveryBibitem{\clearfield{day}}
\usepackage[english]{babel}
\usepackage[T1]{fontenc}
\usepackage{csquotes}
\usepackage{bbm}
\usepackage{amsfonts,amsthm,amsbsy,amssymb,dsfont,stmaryrd}
\usepackage[bbgreekl]{mathbbol}
\usepackage[dvipsnames]{xcolor}
\usepackage{braket}
\usepackage{subcaption}

\usepackage{enumitem}

\usepackage{eucal}

\makeatletter
\newcommand{\leqnomode}{\tagsleft@true\let\veqno\@@leqno}
\newcommand{\reqnomode}{\tagsleft@false\let\veqno\@@eqno}
\makeatother

\usepackage{mathtools}
\usepackage[makeroom]{cancel}

\numberwithin{equation}{section}

\newcommand\myshade{85}
\colorlet{mylinkcolor}{violet}
\colorlet{mycitecolor}{YellowOrange}
\colorlet{myurlcolor}{Aquamarine}

\usepackage[unicode=true,pdfusetitle,
bookmarks=true,bookmarksnumbered=false,bookmarksopen=false,
breaklinks=false,pdfborder={0 0 1},backref=false,
linkcolor  = mylinkcolor!\myshade!black,
citecolor  = mycitecolor!\myshade!black,
urlcolor   = myurlcolor!\myshade!black,
colorlinks = true,
]
{hyperref}
\usepackage[nameinlink]{cleveref}

\usepackage{tikz}
\usetikzlibrary{backgrounds}

\usepackage{ifthen}

\usepackage{graphicx}
\usepackage{slashed}
\definecolor{ct_black}{HTML}{000000}
\definecolor{ct_orange}{HTML}{ED872D}
\definecolor{ct_purple}{HTML}{7A68A6}
\definecolor{ct_blue}{HTML}{348ABD}
\definecolor{ct_turquoise}{HTML}{188487}
\definecolor{ct_red}{HTML}{E32636}
\definecolor{ct_pink}{HTML}{CF4457}
\definecolor{ct_green}{HTML}{467821}

\definecolor{ct2_green}{HTML}{9FF781}
\definecolor{ct2_green_dark}{HTML}{088A08}

\newif\ifshowchanges
\showchangesfalse
\ifshowchanges
  \newcommand{\revadd}[1]{{\color{ct_blue}#1}}
  \newcommand{\revdel}[1]{{\color{ct_red}\emph{[deleted:} #1\emph{]}}}
  \newcommand{\revnote}[1]{{\color{ct_purple}#1}}
  
\else
  \newcommand{\revadd}[1]{#1}
  \newcommand{\revdel}[1]{}
  \newcommand{\revnote}[1]{}
  
\fi

\theoremstyle{plain}
\newtheorem{thm}{\protect\theoremname}[section]
\theoremstyle{plain}
\newtheorem{lem}[thm]{\protect\lemmaname}
\theoremstyle{plain}

\theoremstyle{plain}
\newtheorem{prop}[thm]{\protect\propositionname}
\theoremstyle{plain}
\newtheorem{claim}[thm]{\protect\claimname}

\theoremstyle{remark}
\newtheorem{rem}[thm]{\protect\remarkname}

\theoremstyle{definition}
\newtheorem{defn}[thm]{\protect\definitionname}
\newtheorem{example}[thm]{\protect\examplename}
\providecommand{\assumptionname}{Assumption}
\providecommand{\claimname}{Claim}
\providecommand{\corollaryname}{Corollary}
\providecommand{\definitionname}{Definition}
\providecommand{\lemmaname}{Lemma}
\providecommand{\propositionname}{Proposition}
\providecommand{\remarkname}{Remark}
\providecommand{\theoremname}{Theorem}
\providecommand{\examplename}{Example}

\crefname{section}{Section}{Sections}
\crefname{example}{Example}{Examples}
\crefname{appendix}{Appendix}{Appendices}
\crefname{figure}{Figure}{Figures}
\crefname{assumption}{Assumption}{Assumptions}
\crefname{thm}{Theorem}{Theorems}
\crefname{lem}{Lemma}{Lemmas}
\crefformat{equation}{(#2#1#3)}
\crefname{table}{Table}{Tables}

\crefrangelabelformat{equation}{(#3#1#4--#5#2#6)}

\crefmultiformat{equation}{(#2#1#3}{, #2#1#3)}{#2#1#3}{#2#1#3}
\Crefmultiformat{equation}{(#2#1#3}{, #2#1#3)}{#2#1#3}{#2#1#3}

\newtheorem*{lem*}{\protect\lemmaname}

\newcommand{\ii}{i}

\newcommand{\ZZ}{\mathbb{Z}}
\newcommand{\TT}{\mathbb{T}}
\newcommand{\bS}{\mathbb{S}}

\newcommand{\NN}{\mathbb{N}}
\newcommand{\RR}{\mathbb{R}}

\newcommand{\CC}{\mathbb{C}}

\newcommand{\PP}{\mathbb{P}}
\newcommand{\EE}{\mathbb{E}}

\newcommand{\calA}{\mathcal{A}}
\newcommand{\calB}{\mathcal{B}}
\newcommand{\calC}{\mathcal{C}}
\newcommand{\calF}{\mathcal{F}}
\newcommand{\calN}{\mathcal{N}}
\newcommand{\calG}{\mathcal{G}}

\newcommand{\calD}{\mathcal{D}}
\newcommand{\calV}{\mathcal{V}}
\newcommand{\calU}{\mathcal{U}}
\newcommand{\calSU}{\mathcal{SU}}

\newcommand{\calH}{\mathcal{H}}

\newcommand{\calK}{\mathcal{K}}
\newcommand{\calL}{\mathcal{L}}

\newcommand{\calI}{\mathcal{I}}
\newcommand{\calJ}{\mathcal{J}}
\newcommand{\calP}{\mathcal{P}}

\newcommand{\calX}{\mathcal{X}}
\newcommand{\calY}{\mathcal{Y}}
\newcommand{\calZ}{\mathcal{Z}}

\newcommand\norm[1]{\left\lVert#1\right\rVert}
\newcommand{\ip}[2]{\langle #1, #2 \rangle}

\newcommand{\tr}{\operatorname{tr}}

\newcommand{\szpan}{\operatorname{span}}

\newcommand{\ve}{\varepsilon}
\newcommand{\vf}{\varphi}
\newcommand{\Id}{\mathds{1}}
\newcommand{\HH}{\mathbb{H}}

\newcommand{\Open}[1]{\mathrm{Open}(#1)}
\newcommand{\Closed}[1]{\mathrm{Closed}(#1)}
\newcommand{\dist}{\mathrm{dist}}

\newcommand{\sgn}{\operatorname{sgn}}
\newcommand{\findex}{\operatorname{ind}}

\newcommand{\supp}{\operatorname{supp}}

\DeclarePairedDelimiter\floor{\lfloor}{\rfloor}

\newcommand{\im}{\operatorname{im}}

\newcommand{\sigmaess}{\sigma_{\mathrm{ess}}}

\newcommand{\calAZ}{\mathcal{AZ}}

\usepackage{environ}

\NewEnviron{malign}{%
	\begin{align}\begin{split}
			\BODY
	\end{split}\end{align}
}

\newcommand{\prob}[1]{\PP\left[\Set{#1}\right]}
\newcommand{\ex}[1]{\EE\left[#1\right]}

\newcommand{\eq}[1]{\begin{align*}#1\end{align*}}

\newcommand{\eql}[1]{\begin{align}#1\end{align}}

\newcommand{\NT}{\mathrm{nt}}
\newcommand{\calQ}{\mathcal{Q}}
\newcommand{\br}[1]{\left(#1\right)}

\newcommand{\Ext}{\operatorname{Ext}}

\newcommand{\Cl}{\mathrm{Cl}}

\newcommand{\tensorid}[2]{{#1_{(#2)}}}

\newcommand{\Cstar}{$C^\ast$}

\newcommand{\Kone}{K_1}
\newcommand{\rmK}{K}
\newcommand{\rmC}{C}
\newcommand{\rmDK}{D\kern-0.1emK}
\newcommand{\rmKKO}{K\kern-0.1emK\kern-0.1emO}
\newcommand{\rmKO}{K\kern-0.1emO}

\newcommand{\Mod}[1]{\ (\mathrm{mod}\ #1)}

\newcommand{\Cli}{\mathrm{C\ell}}

\newcommand{\azAI}{\mathrm{A\kern-0.05emI}}
\newcommand{\azBDI}{\mathrm{B\kern-0.05emD\kern-0.05emI}}
\newcommand{\azD}{\mathrm{D}}
\newcommand{\azDIII}{\mathrm{D\kern-0.05emI\kern-0.05emI\kern-0.05emI}}
\newcommand{\azAII}{\mathrm{A\kern-0.05emI\kern-0.05emI}}
\newcommand{\azCII}{\mathrm{C\kern-0.05emI\kern-0.05emI}}
\newcommand{\azC}{\mathrm{C}}
\newcommand{\azCI}{\mathrm{C\kern-0.05emI}}

\title[Topological Classification of Insulators: III]%
{Topological Classification of Insulators:
III. Non-interacting Spectrally-Gapped Systems
in All Dimensions}

\author{Jui-Hui Chung}
\address{Program in Applied and Computational Mathematics, Princeton University}
\email{jc1220@math.princeton.edu}

\author{Jacob Shapiro}
\address{Department of Mathematics, Princeton University}
\email{jacobshapiro@princeton.edu}

\begin{document}
\maketitle
\begin{abstract}
We study non-interacting electrons in disordered materials which exhibit a spectral gap, in each of the ten Altland--Zirnbauer symmetry classes, in all space dimensions. We define an appropriate space of Hamiltonians and a topology on it so that the so-called strong topological invariants become \emph{complete} invariants yielding the Kitaev periodic table, but now derived as the set of path-connected components of the space of Hamiltonians, rather than as $K$-theory groups. We thus confirm the conjecture (phrased e.g. in \cite{KatsuraKoma2018}) regarding a one-to-one correspondence between topological phases of gapped non-interacting systems and the respective Abelian groups $\{0\},\ZZ,2\ZZ,\ZZ_2$ in the spectral gap regime. 

\revadd{The central conceptual point is that spherical locality and bulk non-triviality are the two structural hypotheses which make this non-stable statement true. Spherical locality provides the real-space asymptotic locality needed for the strong index pairings, while bulk non-triviality removes lower-dimensional or edge-type configurations which would otherwise create extra path-components. Once this phase space has been identified, the algebraic input is the standard $K$-theory of the associated Paschke-dual picture, and the remaining technical task is to lift that information to $\pi_0$ of symmetry-constrained projections and unitaries. These definitions of locality and bulk-non-triviality are expected to be the portable part of the argument in regimes, such as mobility gaps and interacting systems, where ordinary stabilized $K$-theory is not by itself the right formulation of the physical classification problem.}
\end{abstract}
\tableofcontents

\section{Introduction}

Topological insulators \cite{Hasan_Kane_2010} are phases of matter which are insulating in the bulk, yet may support robust conducting boundary modes. The paradigmatic example is the integer quantum Hall effect (IQHE) \cite{vKDP1980_PhysRevLett.45.494}, where quantized transport is explained by a topological mechanism; see, for example, \cite{TKNN_1982_PhysRevLett.49.405,ASS_1983_PhysRevLett.51.51,Graf07}. A decisive organizing principle was provided by Kitaev \cite{Kitaev2009}, who arranged free-Fermion topological phases into the periodic table indexed by the Altland--Zirnbauer symmetry classes \cite{AltlandZirnbauer1997} and Bott periodicity. By now there is a large mathematical literature on this subject, including vector-bundle, homotopy-theoretic, $K$-theoretic, KK-theoretic, and noncommutative-geometric approaches; see for example \cite{FreedMoore2013,DeNittisGomi2015,Kubota2016,Thiang2016,BourneCareyRennie2016,PSB_2016,GrossmannSchulz-Baldes2016,kennedy2016bott,KatsuraKoma2018,Kellendonk2019,AlldridgeMaxZirnbauer2020,BourneSchulz-Baldes2020,BourneOgata2021,AvronTurner2022,GontierMonacoSolal2022}.

Throughout this paper the phrase ``classification of insulators'' is used in the conventional tenfold-way sense: non-interacting single-particle bulk Hamiltonians with a spectral gap at the Fermi level, possibly disordered, and subject, when specified, to one of the ten internal Altland--Zirnbauer symmetry constraints. This is the standard Kitaev-table problem. Interacting phases, crystalline or spatial-symmetry-protected phases, and other extensions beyond the tenfold way are important directions, but they are outside the scope of the present theorem and do not alter the mathematical content of the conventional classification problem addressed here.

What is still missing, however, is a disorder-compatible classification of bulk phases by \emph{path-components} of the original Hamiltonian space. Existing approaches provide powerful invariants and robust index formulas, but they generally do not give an if-and-only-if criterion for when two fixed-fiber, disordered Hamiltonians lie in the same topological phase. Put differently: they do not usually compute the underlying topology of the physically relevant Hamiltonian space itself, namely its set of path-connected components.

The point of the present paper is to establish precisely such a completeness result in the spectral-gap regime. We show that, for the bulk Hamiltonian spaces considered here, the strong topological invariants are \emph{complete} invariants: two Hamiltonians lie in the same phase if and only if they are connected by a norm-continuous, symmetry-preserving path inside the same class. Equivalently, the relevant Hamiltonian spaces have exactly the path-components predicted by the strong entries of the Kitaev periodic table.

More concretely, we study the following fixed-fiber classification problem for disordered free-Fermion systems. For each space dimension $d$, each Altland--Zirnbauer symmetry class $\Sigma$, and each finite internal fiber $\CC^N$, we consider the space of bounded self-adjoint Hamiltonians on
\[
\ell^2(\ZZ^d)\otimes\CC^N
\]
which are spectrally gapped at the Fermi level, satisfy the symmetry constraints of $\Sigma$, and are allowed to deform only through norm-continuous symmetry-preserving paths inside the same class. Our main result identifies the path-connected components of the appropriate bulk Hamiltonian space with the strong entry of the Kitaev periodic table. Thus, in the setting treated here, the periodic table is realized as a classification by $\pi_0$ of Hamiltonian spaces, and not merely as a stabilized $K$-theoretic invariant.

The structural reason this problem is subtle is that one must first identify the natural Hamiltonian space and topology. In particular, one needs a real-space notion of locality strong enough to support the relevant index pairings and weak enough to accommodate disorder, and one must exclude configurations which are not genuinely bulk at infinity. The two key notions introduced below are therefore a locality condition, which we call \emph{spherical locality}, and an additional hypothesis, \emph{bulk non-triviality}. Together they single out the class of Hamiltonians for which the strong bulk classification by path-components is the right problem.

\revadd{This is the main conceptual point of the paper. Spherical locality by itself gives a natural ambient local algebra, but it contains non-bulk configurations with extra path-components. Bulk non-triviality by itself is not a locality principle. The theorem below shows that the two conditions imposed together provide the structure needed, in this model, for the strong entries of the Kitaev table to become a genuine non-stable homotopy classification rather than only a stabilized index calculation.}

We work throughout in the non-interacting, spectrally-gapped regime, restrict attention to strong phases, keep the internal fiber finite, and classify by genuine homotopies of Hamiltonians rather than stabilized equivalence classes. The mobility-gap regime and interacting systems provide important motivation, but they are not part of the main theorem proved here.

Stable $K$-theory and fixed-fiber homotopy answer complementary questions. Stable equivalence asks what remains after one is allowed to add trivial degrees of freedom. We ask the finer question of whether two Hamiltonians acting on the same microscopic Hilbert space can actually be connected by an admissible path. This distinction matters whenever one wants an operational criterion for whether a concrete device can be deformed, without changing the underlying degrees of freedom, from one protected regime to another.

The underlying mathematical problem is simple to formulate. One defines a space $\calI$ of Hamiltonians relevant to a given experiment (at fixed Fermi energy), equips it with a topology encoding the physically admissible deformations, and seeks to compute $\pi_0(\calI)$, the set of path-connected components. Stable quantization should then reflect the fact that $|\pi_0(\calI)|>1$: if two Hamiltonians lie in different path-components, one cannot continuously deform one into the other without leaving the class of admissible systems. In other words, topological phases should literally be the path-components of a physically meaningful space of Hamiltonians.

\revadd{This point is especially relevant to the quantum-information motivation for topological matter. In proposals for topological quantum computation and quantum memory, information is protected because it is encoded in global or topological degrees of freedom rather than in microscopic local details; see, for example, \cite{Kitaev2003Anyons,FreedmanKitaevLarsenWang2003,DennisKitaevLandahlPreskill2002,NayakSimonSternFreedmanDasSarma2008,LahtinenPachos2017}. For such a hardware interpretation, a particularly natural question is not only whether two models become equivalent after adding abstract trivial bands, but whether the actual Hamiltonian of the device can be deformed, through admissible Hamiltonians on the same Hilbert space, without closing the gap or breaking the relevant symmetry.}

It is also important that the classification problem be posed in a form compatible with disorder. Physically, the insulating condition may come either from a spectral gap or from a mobility gap. Under strong disorder one typically enters the Anderson-localized regime, where the spectral gap closes but the Fermi level lies in a region of localized states which do not contribute to transport \cite{Aizenman_Graf_1998,EGS_2005}. These considerations already suggest that a formulation tied too rigidly to momentum space cannot be the correct starting point for disordered bulk classification. They also motivate our insistence on real-space locality and on a direct homotopy analysis of Hamiltonians, which is the perspective most likely to extend beyond the spectral-gap regime; see also \cite{Graf_Shapiro_2018_1D_Chiral_BEC,Shapiro2019,Shapiro20,BSS23}.

Insisting on completeness, i.e.\ on computing $\pi_0$ and not merely weaker algebraic functors, is important for at least the following reasons. First, since topological phases are envisioned as resources for robust quantum information processing \cite{KempeKitaevRegev2006,Kitaev2003Anyons,FreedmanKitaevLarsenWang2003,DennisKitaevLandahlPreskill2002,NayakSimonSternFreedmanDasSarma2008}, it is important to know exactly when topological protection can and cannot break down; completeness gives genuine if-and-only-if criteria for topological phase transitions. Second, any eventual understanding of the strongly-disordered mobility-gap regime is unlikely to follow from ordinary $K$-theoretic calculations alone, so direct homotopy-theoretic and operator-theoretic tools should already be developed in the spectral-gap regime. Third, the interacting classification problem remains widely open, even at the many-body level, despite substantial recent progress on index constructions and related questions; see, for example, \cite{Ogata22,ORJ24,Bachmann_Shapiro_Tauber_2026}. \revadd{In those future settings, the lesson of the present paper is not that the same $K$-theory calculation should survive unchanged, but that one should first identify the correct analogues of locality and bulk non-triviality before asking for a homotopy classification.} A direct understanding of path-components is therefore of independent conceptual interest beyond the free-Fermion case.

\subsection{Existing literature.}

Most existing classifications of free-Fermion topological phases fall into three broad categories.

First, in the translation-invariant setting one classifies ground states over momentum-space base manifolds such as $\TT^d$ or $\bS^d$, often in terms of vector bundles or homotopy classes of maps into classifying spaces. This point of view is central to the modern subject and includes, for example, \cite{FreedMoore2013,kennedy2016bott,KennedyGuggenheim2015}. The papers \cite{kennedy2016bott,KennedyGuggenheim2015} are especially relevant conceptually because they are genuinely homotopy-theoretic. However, they concern momentum-space classifications of translation-invariant ground states, whereas we study disordered bulk Hamiltonians in real space. In particular, a classification formulated over momentum-space base manifolds is not compatible with Anderson localization, and the torus-based picture naturally includes weak invariants, whereas the problem addressed here is the disorder-compatible classification of strong bulk phases by path-components.

Second, in disordered settings one replaces momentum space by noncommutative analogues and studies index pairings in operator algebras; see for example \cite{Bellissard_1994JMP....35.5373B,Thiang2016,PSB_2016,KatsuraKoma2018,Kellendonk2019,BourneSchulz-Baldes2020,BourneOgata2021,GontierMonacoSolal2022}. This literature provides robust formulas for invariants and explains their stability under disorder, but typically at the level of $K$-theory rather than $\pi_0$. The distinction matters here: the issue is not only whether an invariant is well-defined, but whether equality of invariants forces an actual norm-continuous path inside the original fixed-fiber Hamiltonian space.

Third, there are coarse-geometric and Roe-algebra approaches \cite{EwertMeyer2019CoarseGeometryTopologicalPhases,kato2023uniformroe}. These correspond to different operator-algebraic models from the one considered here. In the uniform Roe case, the resulting $K$-theory does not reproduce the strong Kitaev table in the form relevant here \cite{kato2023uniformroe}. In the non-uniform Roe setting, one tensors each lattice site with infinitely many internal degrees of freedom \cite{EwertMeyer2019CoarseGeometryTopologicalPhases}. Our focus, by contrast, is the fixed-fiber bulk classification problem.

\revadd{It is useful to separate three issues which are sometimes conflated. The first is the $K$-theory calculation. For the algebraic model used below this calculation is standard, and can be understood through Paschke duality for the representation of $C(\bS^{d-1})$ determined by the normalized position operators \cite{paschke1981k,higson2000analytic,roe2004paschke}. The second is the passage from $K$-theory to actual path-components of unitaries or projections. This step is not automatic, but related mechanisms have already appeared in essentially-commuting models, for instance in \cite{chung2025essentially}. The third issue, which is the main point of the present paper, is conceptual: identify a real-space locality condition and a bulk non-triviality condition which together produce the non-stable strong classification of disordered fixed-fiber Hamiltonians.}

\revadd{The distinction is summarized in the following table.}

\begin{table}[h]
\begingroup

\scriptsize
\begin{center}
\begin{tabular}{p{0.22\textwidth}p{0.35\textwidth}p{0.36\textwidth}}
\noalign{\rule{\textwidth}{0.4pt}}
Framework & Output emphasized & Difference from the present paper \\
\noalign{\rule{\textwidth}{0.4pt}}
Momentum-space homotopy \cite{kennedy2016bott,KennedyGuggenheim2015} & Homotopy classes of translation-invariant ground states over $\TT^d$, $\bS^d$, or related base spaces; naturally includes torus-dependent weak data. & Conceptually homotopy-theoretic, but not a disorder-compatible, real-space classification of path-components of fixed-fiber Hamiltonian spaces. \\
\hline
Noncommutative and index-theoretic approaches \cite{Bellissard_1994JMP....35.5373B,PSB_2016,KatsuraKoma2018} & Robust formulas for strong invariants via operator algebras, Fredholm pairings, and $K$-theory. & Supplies the correct invariants, but typically does not prove an if-and-only-if $\pi_0$ statement for the original fixed-fiber Hamiltonian space. \\
\hline
Roe/coarse approaches \cite{EwertMeyer2019CoarseGeometryTopologicalPhases,kato2023uniformroe} & $K$-theory of coarse or Roe-type algebras, often in stabilized or infinite-fiber models. & Even if one regards the Roe model as natural for coarse geometry, the available results are $K$-theory calculations rather than non-stable homotopy classifications of fixed-fiber Hamiltonians. \\
\hline
This paper & Spherical locality plus bulk non-triviality; non-stable path-components of gapped Hamiltonian spaces in all dimensions and all AZ classes. & The strong Kitaev table is recovered as a genuine fixed-fiber Hamiltonian-space $\pi_0$ classification, not merely as a stabilized index table. \\
\noalign{\rule{\textwidth}{0.4pt}}
\end{tabular}
\end{center}
\caption{A schematic comparison of the problem solved here with nearby classification frameworks.}
\label{table:comparison-frameworks}
\endgroup
\end{table}
Let us now describe the result in more mathematical detail. We model a particle moving through the $d$-dimensional cubic lattice $\ZZ^d$ with a fixed number $N$ of internal degrees of freedom. Thus the relevant Hilbert space is
\[
\ell^2(\ZZ^d)\otimes\CC^N=:\calH_{d,N}.
\]
On it we consider gapped single-particle Hamiltonians, that is, bounded self-adjoint operators $H=H^\ast\in\calB(\calH_{d,N})$ which are invertible. Equivalently, we fix the Fermi energy at $E_F=0$, which is no loss of generality since any other Fermi level can be reduced to this case by shifting by a constant. Invertibility is the spectral-gap form of the insulating condition at $E_F=0$.

We now introduce the two structural notions discussed above: one captures locality, and the other singles out genuinely bulk configurations. We state the definitions here because they are needed to formulate the main theorem; their motivation, equivalent formulations, and structural consequences are developed in the subsequent sections.

\begin{defn}[spherical locality] For any $I\subseteq\mathbb{S}^{d-1}$, let \eql{\label{eq:lambda I on sphere}
\Lambda_I := \br{\sum_{x\in \ZZ^d\setminus\Set{0} : \frac{x}{\norm{x}}\in I} \delta_x\otimes\delta_x^\ast}\otimes\Id_N
}
where $\Set{\delta_x}_{x\in\ZZ^d}$ is the position orthonormal basis of $\ell^2(\ZZ^d)$.

We say that an operator $A\in\calB(\calH_{d,N})$ is \emph{spherically-local} iff for any $I,J\in\Closed{\mathbb{S}^{d-1}}$ such that $I\cap J=\varnothing$,
\eql{
\Lambda_I A \Lambda_J \in \calK(\calH_{d,N})\,,
}
i.e., it is a compact operator between the corresponding conical sectors.
\end{defn}

Next we want to distinguish systems which are genuinely bulk from those which are, in effect, trivial in an open set of directions at infinity.

\begin{defn}[bulk non-triviality] Let $P=P^\ast=P^2\in\calB(\calH_{d,N})$ be a spherically-local projection. We call $P$ \emph{bulk-non-trivial} iff for all $I\in\Open{\bS^{d-1}}\setminus\Set{\varnothing}$,
\eql{
\Lambda_I P\Lambda_I,\Lambda_I P^\perp \Lambda_I\notin\calK(\calH_{d,N})\,.
}
\end{defn}

Since the space of all spherically-local operators forms a \Cstar-algebra with respect to the operator norm, it is closed under continuous functional calculus. In particular, if $H=H^\ast\in\calB(\calH_{d,N})$ is gapped and spherically-local, then its Fermi projection $P\equiv\chi_{(-\infty,0)}(H)$ is also spherically-local. We therefore call a gapped spherically-local Hamiltonian $H$ \emph{bulk-non-trivial} iff its Fermi projection is.

Finally, recall that time-reversal and particle-hole symmetries are anti-unitary operators $\Theta,\Xi:\calH_{d,N}\to\calH_{d,N}$ which may square to $\pm\Id$ and which we assume act non-trivially only on the internal factor $\CC^N$ (equivalently, they commute with the position operators). We say that $H$ is \emph{time-reversal symmetric} iff $[H,\Theta]=0$, and \emph{particle-hole symmetric} iff $\{H,\Xi\}=0$. Setting $\Pi:=\Theta\Xi$, we say that $H$ is \emph{chiral symmetric} iff $\{H,\Pi\}=0$; this may occur even in the absence of time-reversal or particle-hole symmetry separately. In this way one obtains the ten Altland--Zirnbauer symmetry classes appearing in the \nameref{table:Kitaev}.

Our main result is the following.

\begin{thm}[Kitaev table agrees with path-connected components of non-trivial insulators]\label{thm:Kitaev table agrees with path-connected components of non-trivial insulators}
Fix an Altland--Zirnbauer symmetry class $\Sigma$, a dimension $d$ and an internal number of degrees of freedom $N$. Consider the space of all gapped, spherically-local, bulk-non-trivial Hamiltonians $H=H^\ast\in\calB(\calH_{d,N})$ respecting the symmetry class $\Sigma$, taken with the subspace topology with respect to the operator norm topology.

Then the set of path-connected components of this space agrees with the relevant entry within the \nameref{table:Kitaev}.
\end{thm}

Thus, for example, in class A and $d=2$ the path-components are indexed by $\ZZ$, while in class AII and $d=3$ they are indexed by $\ZZ_2$, exactly as predicted by the strong Kitaev table. What is new is that these groups arise here as actual path-components of Hamiltonian spaces. \revadd{The theorem should therefore be read less as a new computation of the abstract groups and more as a validation of the two structural concepts used to define the Hamiltonian space.}

A useful way to summarize the proof of \Cref{thm:Kitaev table agrees with path-connected components of non-trivial insulators} is as follows. First, because we work in the spectral-gap regime, continuous functional calculus allows us to deform any gapped spherically-local Hamiltonian to a canonical representative, such as its sign representative or, equivalently, its Fermi projection; crucially, this reduction respects the Altland--Zirnbauer symmetry constraints, so the classification problem for Hamiltonians becomes a classification problem for symmetry-constrained spherically-local projections or unitaries (see \cite[Lemma 5.13]{ChungShapiro2023}). Second, spherical locality provides the ambient local $C^\ast$-algebra and hence the topology in which the homotopy problem is posed, while bulk non-triviality removes lower-dimensional or edge-type configurations that would otherwise contaminate the bulk classification. Within this framework one computes the relevant index invariants by pairing with the appropriate Dirac-type cycle. Depending on the dimension and symmetry class, these invariants take values in $\ZZ$, $2\ZZ$, $\ZZ_2$, or vanish. \revadd{These index calculations are not the main novelty; they are the expected calculations for the associated Paschke-dual picture.} The main technical step is then to lift these $K$-theoretic calculations to a statement about path-components: after restricting to bulk-non-trivial objects, the same index invariants are shown to be complete for homotopy classes through spherically-local, symmetry-preserving deformations. \revadd{The mechanism for such a lift is related to the essentially-commuting homotopy methods of \cite{chung2025essentially}; what is new here is that the lift is carried out on the Hamiltonian spaces determined by spherical locality and bulk non-triviality, in all dimensions and all AZ symmetry classes.} Concretely, this is achieved by a localization-and-compression scheme: one deforms a local representative to one that is trivial on a large spherically-proper region, uses stabilization in matrix algebras to build the required homotopies, and then compresses back to obtain genuine homotopies in the original local algebra. Finally, the real symmetry classes are handled by passing to the appropriate fixed-point and graded settings, where the resulting periodicity recovers exactly the \nameref{table:Kitaev} at the level of path-connected components.

The present paper treats the all-dimensional, spectrally-gapped, non-interacting case within a broader program initiated in \cite{ChungShapiro2023}, whose goal is to classify topological phases directly as path-components of physically meaningful spaces of Hamiltonians. Parts of this program were carried out in one dimension in \cite{ChungShapiro2023} and, in a more operator-theoretic two-dimensional form, in \cite{chung2024topological,chung2025essentially}. The new feature here is that the classification is carried out in all space dimensions and for all ten Altland--Zirnbauer classes, while remaining in fixed finite fiber, without stabilization, and without assuming translation-invariance.

This paper is organized as follows. In \Cref{sec:Locality} we motivate and introduce spherical locality, relate it to Dirac locality, and establish the basic structural properties of the corresponding \Cstar-algebra $\calL_d$, including its $K$-groups and the relevant Fredholm index pairings. In \Cref{sec:Bulk non-triviality} we formulate bulk non-triviality, record its basic consequences, and explain why it is the correct condition for excluding non-bulk path-components. In \Cref{sec:The classification of bulk-non-trivial spherically-local projections and unitaries} we prove the main complex classification results by lifting the $K$-theoretic computations to $\pi_0$, showing that the strong index is complete for path-components of $\calU(\calL_d)$ in odd dimension and of $\calP^{\NT}(\calL_d)$ in even dimension. In \Cref{sec:Real symmetries} we treat the eight real symmetry classes using van Daele's $K$-theory and thereby complete the proof of \Cref{thm:Kitaev table agrees with path-connected components of non-trivial insulators}.

	\begin{table}
		\begin{center}
			\begin{tabular}{c|ccc|cccccccc}
				\multicolumn{4}{c|}{Symmetry } & \multicolumn{8}{c}{dimension} \\
				\multicolumn{1}{c}{AZ} &$\hspace{1.5mm}\Theta\hspace{1.5mm} $ &
				$\hspace{1.5mm} \Xi\hspace{1.5mm} $ &
				$\hspace{1.5mm} \Pi\hspace{1.5mm} $ &
				$1$   &  $2$ &  $3$ &  $4$ &  $5$ & $6$ & $7$& $8$ \\
				\hline
				A & $0$ & $0$ & $0$  &$0$& $\mathbb{Z}$ &$0$& $\mathbb{Z}$ &$0$& $\mathbb{Z}$ &$0$& $\mathbb{Z}$\\
				AIII & $0$ & $0$ & $1$ & $\mathbb{Z}$ &$0$& $\mathbb{Z}$ &$0$& $\mathbb{Z}$ &$0$& $\mathbb{Z}$& $0$\\
				\hline
				AI & $1$ & $0$ & $0$  &$0$&$0$&$0$&$\mathbb{Z}$&$0$&$\mathbb{Z}_2$&$\mathbb{Z}_2$& $\mathbb{Z}$ \\
				BDI & $1$ &$1$ &$1$ & $\mathbb{Z}$ &$0$&$0$&$0$&$\mathbb{Z}$&$0$&$\mathbb{Z}_2$& $\mathbb{Z}_2$\\
				D & $0$ &$1$ &$0$ & $\mathbb{Z}_2$& $\mathbb{Z}$ &$0$&$0$&$0$&$\mathbb{Z}$&$0$&$\mathbb{Z}_2$\\
				DIII&$-1$ &$1$ &$1$ &$\mathbb{Z}_2$& $\mathbb{Z}_2$& $\mathbb{Z}$ &$0$&$0$&$0$&$\mathbb{Z}$&$0$\\
				AII & $-1$ & $0$ & $0$ &$0$&$\mathbb{Z}_2$& $\mathbb{Z}_2$& $\mathbb{Z}$ &$0$&$0$& $0$&$\mathbb{Z}$\\
				CII & $-1$ &$-1$ & $1$&$\mathbb{Z}$ & $0$&$\mathbb{Z}_2$& $\mathbb{Z}_2$& $\mathbb{Z}$ &$0$&$0$&$0$ \\
				C & $0$ & $-1$& $0$ & $0$ &$\mathbb{Z}$ &$0$&$\mathbb{Z}_2$& $\mathbb{Z}_2$& $\mathbb{Z}$ &$0$& $0$\\
				CI & $1$ & $-1$ & $1$& $0$ & $0$&$\mathbb{Z}$&$0$&$\mathbb{Z}_2$& $\mathbb{Z}_2$& $\mathbb{Z}$& $0$ \\
			\end{tabular}
		\end{center}
        \captionsetup{labelformat=empty}
		\caption[Kitaev table]{The Kitaev periodic table. The entries stand for the respective $K$-theory groups in a given dimension and symmetry class. In this paper we prove that, for the bulk Hamiltonian spaces defined above, these entries are not only $K$-theoretically correct but also $\pi_0$-correct.}
		\label{table:Kitaev}
	\end{table}

\subsection*{Acknowledgements}
JS was supported in part by NSF grant DMS-2510207. The authors wish to thank Christopher Bourne, Charles L. Fefferman and Gian Michele Graf for stimulating discussions. JHC wishes to thank Shouda Wang for helpful discussions.

\section{Locality}\label{sec:Locality}

In this section we focus our attention on the question of locality, the basic physics principle by which far away particles should have negligible interaction. We want to give context so as to eventually define and develop a relatively weak notion of locality, which we will term \emph{spherically-local operators}.

Our single-particle Hilbert space is \eql{\label{eq:single-particle hilbert space}
\ell^2(\ZZ^d)\otimes\CC^N =:\calH_d\otimes \CC^N \equiv \calH_{d,N}
} where $d$ is the space dimension and $N$ is the fixed number of degrees of freedom per lattice site. Since $\CC^N$ is internal to a lattice site, we do not discern any locality within it, and so it is only necessary to keep track of the $\calH_d$ factor when discussing locality.

While in physics usually nearest neighbor, or finite hopping locality is the default, in mathematical physics one mainly encounters
\begin{defn}[exponential locality]\label{def:exp locality} An operator $A\in\calB(\calH_d\otimes\CC^N)$ is exponentially local iff there exist $C<\infty,\mu>0$ such that \eql{\label{eq:exp locality}
    \norm{\ip{\delta_x}{A\delta_y}}\leq C\exp\br{-\mu\norm{x-y}}\qquad(x,y\in\ZZ^d)\,.
} Here $\Set{\delta_x}_{x\in\ZZ^d}$ is the (orthonormal) position basis and $\norm{\cdot}$ on the LHS is \emph{any} matrix norm on the $N\times N$ matrix $\ip{\delta_x}{A\delta_y}$ of the partial matrix elements.
\end{defn}

\begin{prop}
    The space of all exponentially local operators taken with the operator norm is a $\star$-algebra but does \emph{not} form a $C^\ast$-algebra.
\end{prop}
\begin{proof}
    The main failure here is operator norm closure. To see it, consider for instance a translation-invariant operator whose integral kernel has sufficiently fast polynomial decay so as to make it a bounded operator. Its finite-hopping truncations, at hopping distance $R\in\NN$, converge to it in operator norm. But of course finite-hopping operators are exponentially local.
\end{proof}

The fact above prompts one to seek the smallest \Cstar-algebra generated by exponentially-local operators. This turns out to be the pre-existing notion of the 
\begin{defn}[Uniform Roe algebra]\label{def:Uniform Roe algebra} Let 
\eq{
    \rho_{\mathrm{u}}(\ZZ^d) &:= \overline{\Set{A \in \calB(\calH_d) | A\text{ is exponentially local}}} \\
    &= \overline{\Set{A \in \calB(\calH_d) | A_{xy}=A_{xy}\chi_{[0,R]}(\norm{x-y})\exists R\in\NN}}
} where the closure is taken with respect to  the norm topology. Then $\rho_{\mathrm{u}}(\ZZ^d)$ is called the \emph{uniform Roe algebra} over $\ZZ^d$.
\end{defn}
Defining $\rho_{\mathrm{u}}(\ZZ^d)$ as the norm-closure of an algebra, it is automatically a \Cstar-algebra. \revdel{The algebra $\rho_{\mathrm{u}}(\ZZ^d)$ is probably \emph{at odds} with the \nameref{table:Kitaev}.} \revadd{The uniform Roe algebra is a natural object for coarse-geometric questions, but it answers a different question from the fixed-fiber strong-bulk classification considered here.} Indeed, in \cite{kato2023uniformroe}, at least at the level of $K$-theory the two seem to disagree. \revdel{Rather, to obtain the same classification as the table yields, weaker notions of locality exist.} \revadd{In particular, the existing comparison with the Kitaev table is presently a $K$-theory comparison, not a true homotopy classification of Hamiltonian spaces.} One possibility is the 
\begin{defn}[Non-Uniform Roe algebra] Let 
    \eq{
    \rho_{\mathrm{nu}}(\ZZ^d) := \overline{\Set{A \in \calB(\calH_d\otimes \ell^2(\NN)) |\exists R\in\NN: A_{xy} = \chi_{[0,R]}(\norm{x-y})A_{xy} \in \calK(\ell^2(\NN)) }}\,.
} Here, for each $x,y\in\ZZ^d$, the partial matrix element $A_{xy}$ is actually an operator on $\ell^2(\NN)$. So we take the norm closure of finite-hopping operators where each matrix element is a compact operator.
\end{defn}

It turns out that the classification of $\rho_{\mathrm{nu}}(\ZZ^d)$ is much closer to the \nameref{table:Kitaev}, again, at least at the level of $K$-theory \cite{EwertMeyer2019CoarseGeometryTopologicalPhases}. \revdel{However, we find it somewhat unappealing due to tensoring with an infinite-dimensional internal fiber, and moreover, we do not know how to lift these $K$-theory results to $\pi_0$-results.} \revadd{This is an important and natural stabilized/infinite-fiber model. Our problem is different: we keep the internal fiber finite and ask for path-components of the corresponding Hamiltonian space. To our knowledge, the Roe-algebraic literature in this direction has so far produced $K$-theory calculations rather than non-stable homotopy classifications of fixed-fiber Hamiltonians.}

Here, instead, we will study yet another notion of locality. In \cite{ChungShapiro2023} we identified a one-dimensional mode of locality we termed $\Lambda$-locality. 
\begin{defn}[$\Lambda$-locality]\label{def:Lambda-locality} Let $\Lambda := \chi_\NN(X)$ be the projection onto the right half space of $\calH_1$. We term an operator $A\in\calB(\calH_1)$ to be $\Lambda$-local iff $[A,\Lambda]\in\calK$. 
\end{defn}
Clearly the set of all $\Lambda$-local operators forms a $C^\ast$-algebra, and moreover, in studying the topological properties of $\Lambda$-local operators, the spatial structure of $\calH_1$ gets washed away and all that matters is that we single out some fixed projection $\Lambda$ on some separable Hilbert space $\calH$, such that $\dim \im \Lambda,\dim \ker \Lambda$ are both infinite.

It is clear that \Cref{def:exp locality} implies \Cref{def:Lambda-locality} but not vice versa. Indeed, if $A\in\calB(\ell^2(\NN))$ is \emph{any} operator then extending it trivially to $\calB(\ell^2(\ZZ))$ via $\Lambda A \Lambda$ yields a $\Lambda$-local operator which may well fail to be exponentially local. \revdel{Why then, is it legitimate to relax exponential locality to $\Lambda$-locality, one may ask? The answer is that there is nothing particularly special about exponential locality either (why were we allowed to relax finite hopping to exponential locality to begin with?) and that in condensed matter physics the philosophy is we pick whichever mathematically-easiest model we have to still describe physically-non-trivial phenomena. In this sense, working with $\Lambda$-locality, we still have \emph{some} shadow of locality (the left and right hand sides of the sample interact in a ``compact way''), but the topological indices remain well-defined and the classification proofs (as was shown in \cite[Sections 3 and 4]{ChungShapiro2023}) greatly simplify.} \revadd{The point of such a relaxation is not to choose locality arbitrarily, but to isolate exactly the part of locality needed for the bulk indices and for homotopy completeness. In the one-dimensional case, $\Lambda$-locality records that the two directions at infinity interact only compactly, while still supporting the relevant Fredholm index and the associated path-component classification \cite[Sections 3 and 4]{ChungShapiro2023}.}

To generalize $\Lambda$-locality then to higher dimensions we follow the same principle: \revdel{relax exponential locality to some vaguer notion of locality which still keeps the topological indices well-defined and yields simplified proofs.} \revadd{replace metric off-diagonal decay by an asymptotic directional locality condition which is strong enough to define the strong Dirac-type pairings and weak enough to include the disordered examples one expects from localization theory.} The work of Prodan and Schulz-Baldes \cite{PSB_2016} paves the way on how to do this.

Let $k=\floor{d/2}$. The complex Clifford algebra $\Cl_d(\CC)$ of $d$ generators admits a self-adjoint $2^k$-dimensional representation $\Gamma_1,\dots,\Gamma_d$ that satisfy $\{\Gamma_i,\Gamma_j\}=2\delta_{ij}\Id_{2^k}$. On the augmented Hilbert space 
\eql{\label{eq:augmented Hilbert space}
\calH_d\otimes \CC^N\otimes \CC^{2^k}
}
we define the flat Dirac operator as 
\eql{\label{eq:flat Dirac operator}
W_d= \sum_{j=1}^d \widehat{X}_j \otimes \Id_N \otimes \Gamma_j.
}
Here $\widehat{X}_i$ is the unit position operator defined as
\eq{
    \widehat{X}_i := \frac{X_i}{\|(X_1,\dots,X_d)\|},\quad \forall i = 1,\dots, d
}
where $X_i$ is the position operator in the $i$-th coordinate; this is a bounded self-adjoint operator whose spectrum is $[-1,1]$. Note that strictly speaking $\widehat{X}\delta_0$ is not defined, so we use the convention $\widehat{X}\delta_0:=e_1\delta_0$.

\begin{defn}[Dirac locality]\label{defn:Dirac locality}
    We say that $A\in \calB(\calH_d\otimes \CC^N)$ is Dirac-local if
    \eq{
        [A\otimes \Id_{\CC^{2^k}}, W_d] \in \calK\,.
    }
\end{defn}
\begin{rem}\label{rem:auxiliary space is not important}
    We are interested in operators that act on the  Hilbert space $\calH_d\otimes \CC^N$ where $N$ is the number of internal degrees of freedom; see \cref{eq:single-particle hilbert space}. To define the notion of locality for operators on $\calH_d\otimes \CC^N$, following the work of \cite{PSB_2016,GrossmannSchulz-Baldes2016}, we augment the Hilbert space to \cref{eq:augmented Hilbert space} and also the operators $A$ to $A\otimes \Id_{\CC^{2^k}}$. However, we emphasize that we are interested only in the Hilbert space $\calH_d\otimes \CC^N$, and the homotopy of operators acting on it. The augmented Hilbert space is only used to define locality (and eventually topological indices).
\end{rem}

For $d\in 2\NN$ and a suitable representation of $\Cl_d(\CC)$, the flat Dirac operator $W_d$ \cref{eq:flat Dirac operator} takes an off-diagonal form
\eql{\label{eq:flat Dirac operator off-diagonal form}
    W_d = \begin{bmatrix}0 & L_d^\ast \\ L_d & 0\end{bmatrix} \,.
}
where $L_d$ is called the Dirac phase.
Indeed, one can arrange that the $\Gamma_j$ are written as
\eq{
\Gamma_j = \begin{bmatrix}
    0 & \Omega_j \\
    \Omega_j & 0
\end{bmatrix},\ \forall j=1,\dots,d-1, \quad \Gamma_d = \begin{bmatrix}
    0 & -\ii \Id_{2^{k-1}} \\
    \ii \Id_{2^{k-1}} & 0
\end{bmatrix}
}
where $\Omega_1,\dots,\Omega_{d-1}$ is a self-adjoint $2^{k-1}$-dimensional irreducible representation of the complex Clifford algebra $\Cl_{d-1}(\CC)$ which satisfy $\{\Omega_i,\Omega_j\}=2\delta_{ij}\Id_{2^{k-1}}$. A straightforward calculation shows that $\Gamma_1,\dots,\Gamma_{d}$ so defined is a complex representation of $\Cl_d(\CC)$. Then, inserting the expressions into the flat Dirac operator \cref{eq:flat Dirac operator} we get \cref{eq:flat Dirac operator off-diagonal form} with
\eq{
L_{d} = \sum_{j=1}^{d-1}\widehat{X}_j\otimes \Id_N\otimes \Omega_{j} + \ii \widehat{X}_d\otimes \Id_N\otimes \Id_{2^{k-1}}
}
acting on $\calH_d\otimes \CC^N\otimes \CC^{2^{k-1}}$. 
Instead of referring to the flat Dirac operator \cref{eq:flat Dirac operator}, it is conventional to talk about the Dirac projection for $d\in 2\NN+1$, and about the Dirac phase for $d\in 2\NN$, where we collect them in the following

\begin{defn}[Dirac phase and Dirac projection]\label{defn:Dirac phase and projection}
    Let $k=\floor{d/2}$. For $d\in 2\NN$, let $\Gamma_{1},\dots,\Gamma_{d-1}$ be a self-adjoint $2^{k-1}$-dimensional irreducible representation of the complex Clifford algebra $\Cl_{d-1}(\CC)$. Define the Dirac phase as
    \eql{\label{eq:dirac phase}
    L_{d} = \sum_{j=1}^{d-1}\widehat{X}_j\otimes \Id_N\otimes \Gamma_{j} + \ii \widehat{X}_d\otimes \Id_N\otimes \Id_{2^{k-1}}
    }
    which is a unitary operator on $\calH_d\otimes \CC^N\otimes \CC^{2^{k-1}}$.
    
    For $d\in 2\NN+1$, let $\Gamma_1,\dots,\Gamma_d$ be a self-adjoint $2^k$-dimensional irreducible representation of the complex Clifford algebra $\Cl_{d}(\CC)$. Define the Dirac projection as
    \eql{\label{eq:dirac projection}
    \Lambda_{d}= \frac{1}{2}\left(\sum_{j=1}^d \widehat{X}_j\otimes \Id_N\otimes \Gamma_j + \Id\otimes\Id_N\otimes \Id_{2^k}\right)
    }
    which is an orthogonal projection on $\calH_d\otimes \CC^N\otimes \CC^{2^{k}}$.
\end{defn}

\begin{rem}
    Since the Dirac phase and projection encode the same information as the flat Dirac operator, we can define Dirac locality using the Dirac phase and projection.
    For $d\in 2\NN$, an operator $A$ on $\calH_d\otimes \CC^N$ is Dirac local if $[A\otimes \Id_{\CC^{2^{k-1}}},L_d]\in \calK$; for $d\in 2\NN+1$, an operator $A$ on $\calH_d\otimes \CC^N$ is Dirac local if $[A\otimes \Id_{\CC^{2^{k}}},\Lambda_d]\in \calK$.
\end{rem}

\begin{rem}
    One may wonder whether the particular representation used in defining the Clifford algebra $\Cl_d(\CC)$ matters when defining Dirac locality in \cref{defn:Dirac locality}, as the flat Dirac operator (and also the Dirac phase and projection) depend on the explicit representation of the Clifford algebra. As we shall see later, it does not.
\end{rem}

\subsection{Spherically-local operators}

As noted in \cref{rem:auxiliary space is not important}, the auxiliary space $\CC^{2^k}$ in \cref{eq:augmented Hilbert space} is not relevant to what happens in the system $\calH_d\otimes\CC^N$.
Indeed, one can give an equivalent definition of Dirac locality without augmenting the Hilbert space. We begin with an algebraic definition and later on show it is equivalent to the geometric definition given in the introduction; both equivalent formulations shall be useful to us in the sequel.

\begin{defn}[spherical locality]\label{def:spherically-local}
For $A\in\calB(\calH_d\otimes \CC^N)$, we say that $A$ is spherically-local iff
\eq{
    [A,\widehat{X}_j\otimes \Id_N]\in\calK,\quad \forall j=1,\dots,d\,.
}
Denote the space of spherically-local operators as $\calL_{d,N}$, and denote $\calL_d:=\calL_{d,1}$. Sometimes we use the phrase \emph{hyper-spherically-local} if $[A,\widehat{X}_j\otimes \Id_N]=0$ for $j=1,\dots,d$.
\end{defn}

\begin{example}[low dimensions] If $d=1$ then we merely have $\widehat{X}_1=\sgn(X_1)$, and having a compact commutator with this operator is equivalent to the notion of $\Lambda$-locality (\cref{def:Lambda-locality}). If $d=2$ then 
\eql{\label{eq:Laughlin flux insertion operator}
\widehat{X}_1+\ii\widehat{X}_2 = L
}
the so-called Laughlin flux insertion operator studied in \cite{chung2024topological,chung2025essentially}. Clearly an operator essentially commutes with $L$ iff it essentially commutes with both $\widehat{X}_1$ and $\widehat{X}_2$. Indeed, if $[A,L]\in\calK$, since $L$ is unitary, $[A,L^\ast]\in\calK$ too, and so adding and subtracting we get the two components.
\end{example}

The first order of business is to establish that
\begin{lem}
    An operator $A\in\calB(\calH_d\otimes \CC^N)$ is Dirac-local iff it is spherically-local.
\end{lem}
\begin{proof}
If $A$ is spherically-local, i.e., $[A,\widehat{X}_j\otimes \Id_N]\in\calK$ for all $j$, then
\eq{
[A\otimes \Id_{\CC^{2^k}},W_d] = [A\otimes \Id_{\CC^{2^k}},\sum_{j=1}^d \widehat{X}_j\otimes \Id_N\otimes \Gamma_j] = \sum_{j=1}^d [A,\widehat{X}_j\otimes \Id_N] \otimes \Gamma_j \in \calK\,.
}
Conversely, suppose $A$ is Dirac-local.
Since the matrices $\Gamma_1,\dots,\Gamma_d$ satisfy the $\{\Gamma_j,\Gamma_k\}=2\delta_{jk}\Id_{2^k}$ for all $j,k=1,\dots, d$, it follows that
\eq{
\left\{\sum_{j=1}^d [A,\widehat{X}_j\otimes \Id_N]\otimes \Gamma_j,  \Id_{\calH_d}\otimes \Id_N\otimes \Gamma_k \right\} &= \sum_{j=1}^d [A,\widehat{X}_j\otimes \Id_N]\otimes \{\Gamma_j,\Gamma_k\}\\ &= 2[A,\widehat{X}_k\otimes \Id_N]\otimes  \Id_{2^k} \\&\in \calK\,.
}
Thus $[A,\widehat{X}_k\otimes \Id_N]\in\calK$ for all $k=1,\dots,d$.
\end{proof}

\begin{lem}[Exponential locality implies spherical-locality] If $A$ is exponentially local as in \cref{eq:exp locality} then $A$ is spherically-local. The converse is false.
\end{lem}
\begin{proof}
    Two identities will be useful. The first: for any $A\in\calB(\calH_d\otimes\CC^N)$ and the position basis $\Set{\delta_x}_{x\in\ZZ^d}$, we have 
    \eql{
        \norm{A}_p \leq \sum_{k\in\ZZ^d}\br{\sum_{x\in\ZZ^d}\norm{A_{x+k,x}}^p}^{1/p}\qquad(p\geq1)
    } where $\norm{A}_p\equiv\br{\tr\br{|A|^p}}^{1/p}$ is the Schatten-$p$ norm. We shall use the fact that if $\norm{A}_p<\infty$ for some $p<\infty$ then $A$ is compact. 

    The second identity is
    \eql{
        \norm{\hat{x}-\hat{y}}^2\leq 2\frac{\norm{x-y}^2}{1+\norm{x}\norm{y}}\qquad(x,y\in\ZZ^d)\,.
    } Together with the exponential locality of $A$, \cref{eq:exp locality}, these identities imply that the commutator $[\hat{X}_i,A]$ is Schatten-$(d+1)$ and as such compact indeed.

    That the converse is false is clear from the $d=1$ case, where we could take an operator that fails to be exponentially-local $A$, then truncate it to 
    \eq{
    A_{\mathrm{trun.}} := \Lambda A \Lambda + \Lambda^\perp A \Lambda^\perp
    } and it would still fail to be exponentially local, but is obviously trivially $\Lambda$-local.
\end{proof}

We now digress momentarily to discuss Anderson localization and mobility-gapped Hamiltonians. The basic idea is that the absence of states around $E_F$, i.e., a spectral gap, is not the only way a system can be insulating. Another possibility is that the states with energy around $E_F$ are \emph{Anderson localized} due to disorder. This usually emerges as dense pure point spectrum surrounding $E_F$, together with exponential spatial decay of the eigenstates associated to that pure point spectrum \cite{Aizenman_Graf_1998}. The natural setting to derive such estimates is with an ensemble of random operators. However, it cannot be that it is disorder-averaging or the randomness itself which is responsible for topological properties of the system. Indeed, the random structure is merely a mathematical tool to study typical behaviors of non-periodic systems. Since we are interested in topological properties which are universal, we should not rely on the randomness and disorder-averaging, and thus, following the philosophy of \cite{EGS_2005}, we model Anderson localized Hamiltonians via almost-sure consequences on random operators. This is the approach taken in \cite{Shapiro2019,Shapiro20,BSS23}. Roughly speaking, one consequence of a Hamiltonian being mobility-gapped at $E_F$ is that its Fermi projection $P$ exhibits non-uniform off-diagonal decay and the eigenvalues around $E_F$ are simple. This non-uniform off-diagonal decay was termed ``weakly-local'':
\begin{defn}[weakly-local operator] An operator $A\in\calB(\calH_d\otimes\CC^N)$ is weakly-local iff there exists some $\nu\in\NN$ such that for all $\mu\in\NN$ sufficiently large, there exists some $C_\mu<\infty$ with which
\eql{
\norm{A_{xy}}\leq C_\mu \br{1+\norm{x-y}}^{-\mu}\br{1+\norm{x}}^{+\nu}\qquad(x,y\in\ZZ^d)\,.
}   
\end{defn}

Thus, if a Hamiltonian $H$ is mobility-gapped, its associated Fermi-projection $P$ is weakly-local almost surely \cite{Aizenman_Graf_1998,EGS_2005}.

The typical scenario is that the Hamiltonian itself is exponentially-local, or even nearest-neighbor, but the associated Fermi-projection is merely weakly-local. It is thus natural to ask whether weakly-local Fermi-projections are also spherically-local. 
\begin{claim}\label{claim:weakly-local implies spherically local}
    If an operator $A\in\calB(\calH_d\otimes\CC^N)$ is weakly-local then it is spherically-local. The converse is false.
\end{claim}
\begin{proof}
    Let $I,J$ be two closed disjoint subsets of $\mathbb{S}^{d-1}$ and $\Lambda_I,\Lambda_J$ the projections to the associated cones in $\ell^2(\ZZ^d)$, which are disjoint by hypothesis. We will show $\Lambda_I A \Lambda_J$ is Hilbert-Schmidt, and hence compact.

    Since $I,J$ are disjoint, there exists some constant $c>0$ such that \eq{
    \norm{x-y}^2 \geq c \br{\norm{x}+\norm{y}}^2\,.
    } Moreover, $\br{1+\norm{x}}^\nu\leq\br{1+\norm{x}+\norm{y}}^\nu$. Then pick $\mu> \nu+d$ on the estimate
    \eq{
        \norm{\Lambda_I A \Lambda_J}_{\mathrm{HS}}^2 &\leq \sum_{x\in C_I, y\in C_J}\norm{A_{xy}}^2\\
        &\leq \sum_{x,y\in\ZZ^d}C_\mu^2\br{1+\norm{x}+\norm{y}}^{2\br{\nu-\mu}} \\
        &< \infty\,.
    }

    As above, it is clear we could come up with operators which are spherically-local but not weakly-local.
\end{proof}

We next phrase what might replace the fractional-moment condition of Anderson localization in this setting, to yield the necessary spherical-locality of the Fermi projection to define indices:
\begin{prop}[Spherical localization criterion for Fermi projections]\label{conj:spherical fractional moment fermi projection}
Let \((\Omega,\PP)\) be a probability space, and let
\[
H_\omega=H_\omega^\ast\in\calB(\calH_{d,N})
\]
be a measurable family of bounded self-adjoint operators. Assume that
\eql{\label{eq:random Hamiltonian spherical locality}
\prob{H_\omega\in\calL_{d,N}}=1\,.
}
Let \(E_F\in\RR\), and define the Fermi projection
\eq{
P_\omega:=\chi_{(-\infty,E_F)}(H_\omega).
}
Assume that there exists an open interval \(\Delta\subset\RR\), with
\(E_F\in\Delta\), such that, for every \(j=1,\dots,d\),
\eql{\label{eq:spherical fractional moment condition}
\sup_{\substack{\norm{g}_\infty\leq 1\\ \supp(g)\subset\Delta}}
\ex{
\left\|
[\widehat X_j\otimes\Id_N,g(H_\omega)]
\right\|_{\mathfrak S_{d+1}}^{d+1}
}
<\infty,
}
where the supremum is taken over bounded Borel functions
\(g:\RR\to\CC\) supported in \(\Delta\).

Then
\eql{\label{eq:fermi projection spherically local almost surely}
\prob{P_\omega\in\calL_{d,N}}=1\,.
}
\end{prop}
\begin{proof}
Choose \(a<E_F\) such that
\eq{
[a,E_F]\subset\Delta .
}
Let \(f\in C_b(\RR)\) satisfy
\eq{
0\leq f\leq 1,
\qquad
f(t)=1\quad(t\leq a),
\qquad
f(t)=0\quad(t\geq E_F).
}
Define
\eq{
g_0:=\chi_{(-\infty,E_F)}-f .
}
Then \(g_0\) is a bounded Borel function satisfying
\eq{
\norm{g_0}_\infty\leq 1,
\qquad
\supp(g_0)\subset [a,E_F]\subset\Delta .
}
Therefore \(g_0\) is one of the functions allowed in
\cref{eq:spherical fractional moment condition}. Hence, for each
\(j=1,\dots,d\),
\eq{
\ex{
\left\|
[\widehat X_j\otimes\Id_N,g_0(H_\omega)]
\right\|_{\mathfrak S_{d+1}}^{d+1}
}
<\infty .
}
Since the random variable inside the expectation is non-negative, it follows that
\eq{
\left\|
[\widehat X_j\otimes\Id_N,g_0(H_\omega)]
\right\|_{\mathfrak S_{d+1}}
<\infty
}
for almost every \(\omega\). Equivalently,
\eq{
[\widehat X_j\otimes\Id_N,g_0(H_\omega)]
\in\mathfrak S_{d+1}
\subset\calK(\calH_{d,N})
}
for almost every \(\omega\). Taking the intersection of these finitely many
full-measure events, we obtain
\eql{\label{eq:g0 compact commutators almost surely}
\prob{
[\widehat X_j\otimes\Id_N,g_0(H_\omega)]
\in\calK(\calH_{d,N})
\text{ for all }j=1,\dots,d
}=1 .
}

On the other hand, by \cref{eq:random Hamiltonian spherical locality},
\eq{
\prob{H_\omega\in\calL_{d,N}}=1 .
}
Since \(\calL_{d,N}\) is a \Cstar-algebra, it is closed under continuous
functional calculus. Thus
\eql{\label{eq:fH spherically local almost surely}
\prob{f(H_\omega)\in\calL_{d,N}}=1 .
}
Equivalently,
\eq{
[\widehat X_j\otimes\Id_N,f(H_\omega)]
\in\calK(\calH_{d,N}),
\qquad j=1,\dots,d,
}
almost surely.

Now
\eq{
P_\omega
=
\chi_{(-\infty,E_F)}(H_\omega)
=
f(H_\omega)+g_0(H_\omega).
}
Therefore, on the full-measure event on which both
\cref{eq:g0 compact commutators almost surely} and
\cref{eq:fH spherically local almost surely} hold, we have, for every
\(j=1,\dots,d\),
\eq{
[\widehat X_j\otimes\Id_N,P_\omega]
&=
[\widehat X_j\otimes\Id_N,f(H_\omega)]
+
[\widehat X_j\otimes\Id_N,g_0(H_\omega)]
\\
&\in
\calK(\calH_{d,N})+\calK(\calH_{d,N})
\\
&\subset
\calK(\calH_{d,N}) .
}
Hence
\eq{
P_\omega\in\calL_{d,N}
}
almost surely.
\end{proof}

Since Fermi projections of exponentially local mobility-gapped Hamiltonians are weakly local \cite{EGS_2005}, \cref{claim:weakly-local implies spherically local} allows to conclude that such Fermi projections are  spherically-local, and hence a classification of these is tantamount to the classification of ground states of mobility-gapped Hamiltonians. However, unlike in the spectral gap situation, we do not have a deformation retraction argument that sets up a topological isomorphism between the space of spherically-local mobility-gapped Hamiltonians and spherically-local projections.

We postpone to future studies the investigation on the topological connection between ground states of mobility-gapped Hamiltonians and the Hamiltonians themselves.

\revadd{This is one reason we emphasize the concepts rather than only the spectral-gap theorem. In the mobility-gap setting there may be no useful deformation retraction from Hamiltonians to projections inside the same class of Hamiltonians, and in interacting systems the single-particle $K$-theoretic object disappears altogether. What should survive, if a direct classification is possible, is the need for an appropriate locality condition and an appropriate bulk non-triviality condition. The present spectral-gap theorem is evidence that spherical locality and bulk non-triviality are the correct such conditions in the non-interacting model.}

\begin{rem}[Other modes of locality driven by indices] One may consider different, inequivalent modes of locality, driven by the index formulas. For instance, consider the integer quantum Hall effect, $d=2$. Then, the formula for the Hall conductivity is given by 
\eq{
    \sigma_{\mathrm{Hall}}(P) = 2\pi \tr\br{P [[\Lambda_1,P],[\Lambda_2,P]]}
} where $\Lambda_i \equiv \chi_{\NN}(X_i)$ for $i=1,2$. As such, in $d=2$, one could imagine to define $P$ as being local iff $[\Lambda_1,P][\Lambda_2,P]$ is trace-class, or if $[\Lambda_1,P] K_2$ and $K_1 [\Lambda_2,P]$ are, for any two operators $K_1,K_2$ which have finite support window in the $i=1,2$ axis. This leads to the non-commutative Sobolev spaces of Bellissard et al \cite{Bellissard_1994JMP....35.5373B}.

Yet another possibility comes from Kitaev's formula for the index associated with the Chern number \cite{FSSWY22}. For a Fermi projection $P$ and $\Lambda_1,\Lambda_2$ as above, we define $P$ to be local iff 
\eq{
[\Lambda_1,\exp\br{-2\pi\ii \Lambda_2 P \Lambda_2}]
} is compact.

It is \emph{not true} that these two alternatives are equivalent to spherical-locality, indeed, they appear as weaker notions. They are also very difficult to work with and seem less natural.
\end{rem}

\paragraph{Further abstraction.} It would be important later to study the commutative
$C^\ast$-algebra generated by $\widehat{X}_1,\dots,\widehat{X}_d$ which we denote as
\eql{\label{eq:C star algebra generated by unit position operators}
\calX_d:=C^*(\widehat{X}_1,\dots,\widehat{X}_d)\,.
}

\begin{lem}\label{lem:maximal ideal space algebra generated by Xis}
There is a $*$-isometric isomorphism 
\eql{\label{eq:isomorphism continuous functions on sphere to algebra generated by Xis}
\rho:C(\bS^{d-1})\to \calX_d
}
such that
\eql{\label{eq:some properties of isomorphism from sphere}
\rho(x\mapsto x_k) = \widehat{X}_k\qquad(k=1,\cdots,d),\quad \rho(x\mapsto 1) = \Id\,.
}
Furthermore, the representation $\rho:C(\bS^{d-1})\to \calB(\calH_d)$ (obtained by the inclusion $\calX_d\hookrightarrow \calB(\calH_d)$) is faithful and $\im \rho$ contains no non-zero compact operators.
\end{lem}
\begin{proof}
Since $\calX_d$ is an Abelian $C^\ast$-algebra, by the Gelfand-Naimark theorem \cite[Theorem 4.29]{Douglas1998}, it is $\ast$-isomorphic to the space of continuous functions $C(\Delta_d)$ where $\Delta_d$ is the maximal ideal space of $\calX_d$, defined as the space all multiplicative linear functionals on $\calX_d$ equipped with weak-$\ast$ topology. Consider the map
\eql{\label{eq:homeomorphism from maximal ideal space to sphere}
    \Delta_d \ni \varphi\mapsto (\varphi(\widehat{X}_1),\dots,\varphi(\widehat{X}_d)) \in \RR^d\,.
} Note that it indeed takes values in $\RR^d$ thanks to the self-adjointness of $\widehat{X}_j$ and the $\ast$-homomorphism property of $\vf$. The map is a homeomorphism onto its image. The map is continuous by definition of weak-$\ast$ topology. If $\vf(\hat{X}_j)=\psi(\hat{X}_j)$ for all $j$, then $\vf(Y)=\psi(Y)$ whenever $Y$ is a polynomial in $X_1,\dots,X_d$, and since these polynomials are dense in $\calX_d$, it follows that $\vf=\psi$ and hence the map \cref{eq:homeomorphism from maximal ideal space to sphere} is injective. 
We show that the image is $\bS^{d-1}$, and hence $\Delta_d\cong \bS^{d-1}$. To that end, let $\varphi\in \Delta_d$. Since $\widehat{X}^2_1+\dots+\widehat{X}_d^2=\Id$, it follows that
\eq{
    \varphi(\widehat{X}_1)^2 + \dots+\varphi(\widehat{X}_d)^2 = 1\,. 
}
Therefore, the image must lie in $\bS^{d-1}$. To show that each point on $\bS^{d-1}$ corresponds to some functional in $\Delta_d$, we consider the evaluation functional: for each $x\in \ZZ^d$, define $\varphi_x(\widehat{X}_i)=x_i/\|(x_1,\dots,x_d)\|$. Thus, we have shown that
\eq{
    \calX_d \cong C(\bS^{d-1})\,.
}

Now we verify \cref{eq:some properties of isomorphism from sphere}. The Gelfand-Naimark theorem gives an isomorphism, the Gelfand transform \cite[Definition 2.24]{Douglas1998}, that maps $T\in\calX_d$ onto $\varphi\mapsto \varphi(T)$ for $\varphi\in \Delta_d$. Fix $k\in \{1,\dots,d\}$. The map $x\mapsto x_k$ in $C(\bS^{d-1})$ corresponds via \cref{eq:homeomorphism from maximal ideal space to sphere} to the map $\varphi\mapsto \varphi(\widehat{X}_k)$ in $C(\Delta_d)$. On the other hand, the Gelfand transform maps $\widehat{X}_k$ to $\varphi\mapsto \varphi(\widehat{X}_k)$. Thus, we have established $\rho(x\mapsto x_k)=\widehat{X}_k$. Finally, the identity $\rho(x\mapsto 1)=\Id$ follows from the fact that all the isomorphisms constructed are unital.

Suppose $\rho(f)$ is compact for $f\in C(\bS^{d-1})$, then $\im f$ only accumulates at zero, and hence by continuity, $f=0$. Thus the only compact element in $C(\bS^{d-1})$, and hence in $\calX_d$, is the zero element.
\end{proof}

Using \cref{lem:maximal ideal space algebra generated by Xis} we can describe $\calL_d$ as
\eql{\label{eq:dual characterization of local algebra}
    \calL_d = \Set{A\in\calB(\calH_d)| [A,\rho(f)]\in\calK(\calH_d),\, \forall f\in C(\bS^{d-1})}=:\calD_\rho(C(\bS^{d-1}))
}
where $\rho$ is the isomorphism \cref{eq:isomorphism continuous functions on sphere to algebra generated by Xis}. The form \cref{eq:dual characterization of local algebra} will have important consequences in calculating the $K$-groups of this algebra.

\revadd{In this language $\calL_d$ is the concrete Paschke dual algebra associated with the representation $\rho:C(\bS^{d-1})\to\calB(\calH_d)$, i.e. the algebra of operators which commute with $\rho(C(\bS^{d-1}))$ modulo the compact operators; compare \cite{paschke1981k,higson2000analytic,roe2004paschke}. This observation is useful context and makes the $K$-theory of the ambient algebra unsurprising. What is not contained in Paschke duality alone is the physical phase space singled out below: the bulk-non-trivial spherically-local projections and Hamiltonians, with fixed finite fiber and Altland--Zirnbauer symmetry constraints, considered up to genuine norm-continuous homotopy.}


\subsection{Re-dimerization}


It suffices to consider $\calL_d$ without the internal degrees of freedom since the spaces $\calL_d$ and $\calL_{d,N}$ are unitarily equivalent. To that end, we build a re-dimerization operator. See also \cref{subsec:abstract spherical locality} for an alternative abstract argument.

In the special case when $N=m^d$ for some $m\in \NN$, we consider the ``re-dimerization'' of perfect hypercubes: for all $x\in\ZZ^d$, decompose it as
\eq{
x_i = m q_i + r_i
}
for $q_i\in\ZZ$ and $0\leq r_i<m$, and define the unitary operator $U$ as
\eq{
U:\calH_d\ni\delta_x\mapsto \delta_q\otimes e_r \in \calH_d\otimes \CC^{m^d}
}
where $e_r$ is the standard basis for $\CC^{m^d}$. 
The operator $U$ ``re-dimerizes`` a perfect hypercube into a single lattice point but with $N=m^d$ degrees of freedom. 
We have
\eq{
K_i:=\widehat{X}_i - U^\ast (\widehat{X}_i\otimes \Id_{m^d}) U \in \calK(\calH_d),\quad \forall i=1,\dots,d \,.
}
Indeed, we have $U^\ast (\widehat{X}_i\otimes \Id_{m^d}) U \delta_x = \hat{q}_i \delta_x$, and $\widehat{X}_i\delta_x = \hat{x}_i\delta_x$ which are both diagonal operators, and it is straightforward to show that $\|\hat{x}-\hat{q}\|\to 0$ as $\|x\|\to \infty$. Therefore, the operator $K_i$ is compact.
The map 
\eq{
\calL_d \ni A \mapsto UAU^\ast \in\calL_{d,N}
}
provides the isomorphism from spherically-local operators on $\calH_d$ to $\calH_d\otimes \CC^{m^d}$. To that end, suppose $A\in\calB(\calH_d)$ satisfies $[A,\widehat{X}_j]\in\calK$ for all $j=1,\dots,d$. Then
\eq{
\left[UAU^\ast, \widehat{X}_j\otimes\Id_{m^d}\right] = U \left[A, U^\ast(\widehat{X}_j\otimes \Id_{m^d})U \right] U^\ast \,.
}
Now, replace $U^\ast(\widehat{X}_i\otimes \Id_{m^d})U$ with $\widehat{X}_i-K_i$, we see that the expression is compact.

In general, we have the following whose proof is similar to the one in the previous paragraph.
\begin{prop}[Re-dimerization]\label{prop:re-dimerization}
    Let $N\in\NN$ be arbitrary and let $M\in M_d(\ZZ)$ be an integer-valued matrix with $|\det M|=N$ and let $R\subset \ZZ^d$ be the set of coset representatives of $\ZZ^d/M\ZZ^d$ (with $|R|=N$). For each lattice point $x\in \ZZ^d$, there is a unique decomposition
    \eq{
    x = Mq+r
    }
    where $q\in\ZZ^d$ and $r\in R$. Define the unitary operator on standard basis as
    \eq{
    \calH_d\ni \delta_x\mapsto \delta_{q}\otimes e_r \in \calH_d \otimes \CC^N .
    }
    Then
    \eq{
    \calL_d \ni A \mapsto UAU^\ast \in \calL_{d,N}
    }
    is a unitary equivalence.
\end{prop}
\begin{proof}
    Define the operators
    \eq{
    \widehat{Y}_j := \frac{(M\widehat{X})_j}{\|M\widehat{X}\|}
    }
    acting on $\calH_d$. Define the \Cstar-algebra
    \eq{
    \calY_{d,N}:= \Set{A\in\calB(\calH_d\otimes \CC^N) | [A,\widehat{Y}_j\otimes \Id_N]\in\calK,\ \forall j=1,\dots,d}.
    }
    We first show that $\Set{UAU^\ast | A\in\calL_d} = \calY_{d,N}$.  
    The operator $U^* (\widehat{Y}_j\otimes \Id_N) U$ is diagonal 
    \eq{
    U^* (\widehat{Y}_j\otimes \Id_N) U \delta_x = U^* (\widehat{Y}_j\otimes \Id_N) \delta_q\otimes e_r = \frac{(Mq)_j}{\|Mq\|}\delta_x
    }
    with eigenvalues $\lambda_{x}=(Mq)_j/\|Mq\|$. Let us analyze the difference $D_j = \widehat{X}_j - U^\ast \widehat{Y}_j\otimes\Id_N U\in\calB(\calH_d)$.
    We verify that the operator $D_j$ is compact. Since $D_j$ is diagonal, it is compact if and only if its eigenvalues decay to zero as the index goes to infinity ($\|x\| \to \infty$). Using the vector inequality $\|u/\|u\| - v/\|v\|\|\leq 2\|u-v\|/\|v\|$, we have
    \eq{
    \left\|x - \frac{Mq}{\|Mq\|}\right\| = \left\|\frac{Mq+r}{\|Mq+r\|} - \frac{Mq}{\|Mq\|}\right\| \leq \frac{2\|r\|}{\|Mq\|}. 
    }
    Since $R$ is a finite set, we have $2\|r\|\leq 2\max_{r\in R}\|r\|$, and since $M$ is non-singular, it has a smallest singular value $\sigma_{\min} > 0$, and hence $\|Mq\| \ge \sigma_{\min} \|q\|$. Thus $2\|r\|/\|Mq\|\to 0$ as $\|q\|\to \infty$ and hence as $\|x\|\to\infty$. This shows that $D_j\in\calK(\calH_d)$. 
    Suppose $A\in\calL_d$. Then 
    \eq{
    [UAU^\ast,\widehat{Y}_j\otimes \Id_N] = U[A,U^\ast (\widehat{Y}_j\otimes \Id_N)U ]U^\ast = U[A,\widehat{X}_j-D_j]U^\ast \in\calK(\calH_d\otimes \CC^N).
    }
    This shows that $\Set{UAU^\ast | A\in\calL_d} \subset \calY_{d,N}$. The other direction is argued in the same way.
    
    Next we show that
    \eq{
    \calY_{d,N} = \calL_{d,N}.
    }
    Intuitively, the algebra $\calY_{d,N}$ are those ``elliptically''-local operators, i.e., those operators that essentially commute with $\widehat{Y}_j\otimes \Id_N$. 
    Consider the smooth map $\varphi:\bS^{d-1}\to\bS^{d-1}$ defined as $\varphi(u) = Mu/\|Mu\|$ whose inverse is smooth as well. Using continuous functional calculus, we have $\widehat{Y}_j\otimes\Id_N = \vf(\widehat{X}_j\otimes\Id_N) $ and $\widehat{X}_j\otimes\Id_N = \vf^{-1}(\widehat{Y}_j\otimes\Id_N)$. If $A\in\calB(\calH_d\otimes \CC^N)$ satisfies $[A,\widehat{X}_j\otimes\Id_N]\in\calK$, then $[A,\widehat{Y}_j\otimes\Id_N]=[A,\vf(\widehat{X}_j\otimes\Id_N)]\in\calK$ which shows that $\calL_{d,N}\subset \calY_{d,N}$. This other direction is analogous.
\end{proof}


The matrix $M \in M_d(\mathbb{Z})$ in \cref{prop:re-dimerization}  defines the geometry of the ``re-dimerization hypercube`` or ``coarse block`` on the lattice $\ZZ^d$. The columns of $M$ are the basis vectors for the re-dimerized lattice. Instead of moving by the standard unit vectors $e_1, \dots, e_d$, the re-dimerized hypercube moves by the columns of $M$.

\begin{example}
    Consider re-dimerizing $\ZZ^2$ lattice. We can tile the plane with  dominoes using 
    \eq{
    M=\begin{bmatrix}
    2 & 0 \\ 0 & 1
    \end{bmatrix},\quad R=\Set{(0,0),(1,0)}.
    }
    We can also tile the plane with L-shape using 
    \eq{
    M=\begin{bmatrix}
    2 & -1 \\ 1 & 2
    \end{bmatrix},\quad R=\Set{(0,0),(0,1),(1,0),(1,1),(0,2)}.
    }
    They provide a re-dimerization of $\calH_2\to \calH_2\otimes \CC^2$ and $\calH_2\to \calH_2\otimes \CC^5$ respectively.
\end{example}

\subsection{Abstract spherical-locality}\label{subsec:abstract spherical locality}


Let $\calY_d$ be a commutative $C^\ast$-algebra with identity generated by elements $Y_1,\dots,Y_d\in\calY_d$. We define the joint spectrum of $Y_1,\dots,Y_d$ in $\calY_d$ as
\eql{\label{eq:joint spectrum}
\sigma(Y_1,\dots,Y_d) := \Set{(\vf(Y_1),\dots,\vf(Y_d)) | \vf\in \Delta_d} \subset \CC^d
}
where $\Delta_d$ is the maximal ideal space of $\calY_d$. It is known that $\Delta_d$ is homeomorphic to $\sigma(Y_1,\dots,Y_d)$. See \cite[Definition 2.3.2, Lemma 2.3.3]{kaniuth2009course} for more detail. 

\begin{defn}[Abstract spherical locality]\label{defn:abstract spherical locality}
    Let $\calH$ be a separable Hilbert space and $\calY_d\subset \calB(\calH)$ be commutative $C^\ast$-algebra with identity generated by operators $Y_1,\dots,Y_d\in \calB(\calH)$. 
    Suppose $\sigma(Y_1,\dots,Y_d)=\bS^{d-1}$. Then an operator $A\in\calB(\calH)$ is called $\calY_d$-spherically-local iff $[A,Y_j]\in\calK$ for all $j$. We denote the space of  $\calY_d$-spherically-local operators as $\calL (\calY_d)$. We call $\calY_d$ the spherical algebra (generated by $Y_1,\dots,Y_d$ acting on $\calH$.)
\end{defn}

\begin{example}
    Consider the Hilbert space $\calH_d=\ell^2(\ZZ^d)$. In $d=1$, take $Y_1=\widehat{X}_1=\Lambda-\Lambda^\perp$. Then $\sigma(Y_1)=\bS^{0}$ and $\rmC^*(Y_1)=\rmC^*(\Lambda)$ and $\calL(\rmC^*(Y_1)) = \calL_1$ (defined in \cref{def:spherically-local}). In $d=2$, take $Y_1=\widehat{X}_1$ and $Y_2=\widehat{X}_2$. Then $\sigma(Y_1,Y_2)=\bS^1$ and $\rmC^\ast(Y_1,Y_2)=\rmC^\ast(L)$ where $L$ is the Laughlin flux insertion operator \cref{eq:Laughlin flux insertion operator}, and we have $\calL(\rmC^\ast(L)) = \calL_2$.
\end{example}


\begin{lem}\label{lem:abstract spherical locality isomorpihc}
    Let $\calY_d$ and $\calY'_d$ be spherical algebras defined in \cref{defn:abstract spherical locality}.
    Then $\calL(\calY_d)\cong \calL(\calY'_d)$. 
\end{lem}
\begin{proof}
    It follows from \cref{defn:abstract spherical locality} that we can describe $\calL(\calY_d)$ as
    \eq{
    \calL(\calY_d) = \Set{A\in\calB(\calH) | [A,Y]\in \calK,\, \forall Y\in \calY_d}.
    }
    We can analogously describe the algebra $\calL(\calY'_d)$. 
    Since $\calY_d$ is generated by some operators $Y_1,\dots,Y_d$ whose joint spectrum (as defined in \cref{eq:joint spectrum}) is $\bS^{d-1}$, it follows by Gelfand-Naimark theorem that $\calY_d\cong C(\bS^{d-1})$. Analogously, we have $\calY'_d\cong C(\bS^{d-1})$. Let $\rho:C(\bS^{d-1})\to\calY_d$ and $\rho':C(\bS^{d-1})\to\calY'_d$ be isomorphic representations given by the Gelfand-Naimark theorem. Then we have
    \eq{
    \calL(\calY_d) = \Set{A\in\calB(\calH) | [A,\rho(f)]\in \calK,\, \forall f\in C(\bS^{d-1})}\,.        
    }
    We can analogously describe the algebra $\calL(\calY'_d)$. 
    Let $\calY_d\subset \calB(\calH)$ and $\calY'_d\subset\calB(\calH')$ be algebras acting on the Hilbert spaces $\calH,\calH'$. 
    Crucially, the algebras $\calY_d,\calY'_d$ of operators do not contain compact ones except zero (similarly shown in \cref{lem:maximal ideal space algebra generated by Xis}). In this setting, we can apply \cite[Theorem 3.4.6]{higson2000analytic} (which is itself a consequence of the Voiculescu's Theorem) to get a unitary isomorphism $U:\calH\to\calH'$ such that
    \eql{\label{eq:tmp unitary isomorphism in showing isomorphic spherically-local}
    U\rho(f)U^\ast - \rho'(f)\in\calK(\calH),\quad \forall f\in C(\bS^{d-1})\,.
    }
    We claim that $\calL(\calY_d)\ni A\mapsto UAU^\ast \in\calL(\calY'_d)$ is well-defined and gives the sought-after isomorphism. Indeed, let $A\in \calL(\calY_d)$ and $f\in C(\bS^{d-1})$ be arbitrary, then
    \eq{
    [UAU^\ast,\rho'(f)] = U[A,U^\ast \rho'(f)U]U^\ast \in \calK
    }
    where we apply \cref{eq:tmp unitary isomorphism in showing isomorphic spherically-local} and the assumption that $A\in\calL(\calY_d)$ in the last step. Thus $UAU^\ast\in\calL(\calY'_d)$.
\end{proof}

Let us connect the spherical locality in \cref{def:spherically-local} to the abstract one in \cref{defn:abstract spherical locality}. Let $\calX_{d,N}\subset \calB(\calH_d\otimes \CC^N)$ be the commutative $C^\ast$-algebra with identity generated by $\widehat{X}_1\otimes \Id_N,\dots,\widehat{X}_d\otimes \Id_N$.
Then the space of spherically-local operators on $\calH_d\otimes \CC^N$ is
\eq{
    \calL_{d,N}= \Set{A\in\calB(\calH_d\otimes \CC^N) | [A,X]\in\calK,\, \forall X\in\calX_{d,N}}\,.
}
Following \cref{lem:maximal ideal space algebra generated by Xis}, it is clear that $\calX_{d,N}$ is a spherical algebra, i.e., the joint spectrum of $\widehat{X}_1\otimes \Id_N,\dots,\widehat{X}_d\otimes \Id_N$ is $\bS^{d-1}$, and hence the spherically-local operators $\calL_{d,N}$ coincide with the $\calX_{d,N}$-spherically-local operators $\calL(\calX_{d,N})$.

On the other hand, according to \cref{lem:maximal ideal space algebra generated by Xis}, the algebra $\calX_d\subset \calB(\calH_d)$ defined in \cref{eq:C star algebra generated by unit position operators} is also a spherical algebra. Thus, using \cref{lem:abstract spherical locality isomorpihc}, we have 
\eq{
\calL_{d,N}\cong \calL(\calX_{d,N})\cong \calL(\calX_d) \cong \calL_d
}
which provides, alternative to \cref{prop:re-dimerization}, an abstract identification of $\calL_{d,N}$ and $\calL_d$ for all $N\in\NN$. 

\begin{rem}
    It may be tempting to characterize abstractly the spherically-local algebras using a set of commuting self-adjoint operators $Z_1,\dots,Z_d$ on some separable Hilbert space $\calH$ such that the spectrum is $\sigma(Z_j)=[-1,1]$ for all $j$, and $\sum_{j=1}^d Z_j^2=\Id$.
    However, the commutative $C^\ast$-algebra $\calZ_d$ generated by $Z_1,\dots,Z_d$ may not necessarily be isomorphic to $C(\bS^{d-1})$.
    For starters, it fails for $d=1$. Indeed, the assumption $\sigma(Z_1)=[-1,1]$ contradicts another assumption $Z_1^2=\Id$.
    For $d=2$, consider the configuration $\Omega = \ZZ^2\setminus \Set{x\in\ZZ^2| x_1>0,x_2>0}$. Let $\widehat{Z}_1,\widehat{Z}_2$ be the unit position operators on $\ell^2(\Omega)$. Clearly $\sigma(Z_1)=\sigma(Z_2)=[-1,1]$ and $Z_1^2+Z_2^2=\Id$.
    The $C^\ast$-algebra $\calZ$ is generated by $Z_1+\ii Z_2$ and $Z_1-\ii Z_2$, and hence $\calZ$ is isomorphic to $C(\sigma(Z_1+\ii Z_2))$. However, due to the geometry of $\Omega$, we have $\sigma(Z_1+\ii Z_2)=\bS^{1}\setminus\Set{0< \vf< \pi/2}$.
\end{rem}

\subsection{Geometric characterization of spherically-local operators}

Let $F\subset \bS^{d-1}$, we denote $\Lambda_F\in\calB(\calH_d)$ to be the projection
\eql{\label{eq:Lambda_F projection definition}
    \Lambda_F:=\sum_{x\in \ZZ^d\setminus\Set{0}\,:\, \hat{x}\in F} \delta_x\otimes \delta_x^\ast
}
where $\hat{x}_i=x_i/\|(x_1,\dots,x_d)\|$.

\begin{thm}\label{thm:geometric characterization of locality}
    An operator $A\in\calB(\calH_d)$ is spherically-local in the sense of \Cref{def:spherically-local} iff $\Lambda_F A\Lambda_G\in\calK(\calH_d)$ for every disjoint pair of closed subsets $F,G$ of $\bS^{d-1}$.
\end{thm}

    \cref{thm:geometric characterization of locality} generalizes \cite[Theorem 2.5]{chung2024topological} to higher dimensions. The proof below is adapted from \cite[Theorem 2.1]{gilfeather1983commutants}.

\begin{proof}[Proof of \Cref{thm:geometric characterization of locality}]
    Let $\rho:C(\bS^{d-1})\to \calB(\calH_d)$ be the representation induced by the composition of the isomorphism $C(\bS^{d-1})\to \calX_d$ in \cref{eq:isomorphism continuous functions on sphere to algebra generated by Xis} and the inclusion $\calX_d\hookrightarrow \calB(\calH_d)$. By \cite[Theorem IX.1.14]{conway2007course}, there exists a spectral measure $E$ on the Borel sets of $\bS^{d-1}$ such that
    \eql{\label{eq:spectral measure on Xd}
        \rho(f) = \int f dE,\quad \forall f\in C(\bS^{d-1})
    }
    and $\rho$ extends to a representation on $B(\bS^{d-1})$, the space of bounded Borel-measurable functions on $\bS^{d-1}$.
    
    Let $F\subset \bS^{d-1}$ be any Borel set. We argue that
    \eql{\label{eq:lambda is spectral measure}
        \Lambda_F = E(F)\,.
    }
    Fix a point $z\in \bS^{d-1}$. Let $f\in C(\bS^{d-1})$ be the continuous function defined by $f_z(x) = \sum_{i=1}^d (x_i-z_i)^2$.
    Using \cref{eq:some properties of isomorphism from sphere}, $f_z$ corresponds to the operator $\rho(f_z)=\sum_{i=1}^d (\widehat{X}_i-z_i\Id)^2$ in $\calX_d$ whose kernel is $\im \Lambda_{\{z\}}$. Using \cite[Theorem 12.28]{rudin1991functional}, it follows that $\im \Lambda_{\{z\}} = \ker \rho(f_z) = \im E(\{z\})$.    
    Since there are countably many points in $\ZZ^d$, there are only countably many points $z$ on $\bS^{d-1}$ for which $\im \Lambda_{\{z\}}$ is non-vanishing. Thus, using the countable additivity of the spectral measure $E$, it follows that \cref{eq:lambda is spectral measure} holds for all Borel sets in $\bS^{d-1}$.

    Using the representation $\rho:B(\bS^{d-1})\to \calB(\calH_d)$ and the characterization \cref{eq:lambda is spectral measure} of the spectral measures, we are ready to prove the result. Suppose $A\in\calL_d$. Let $F,G\subset\bS^{d-1}$ be a disjoint pair of closed subsets. Then, there exists a continuous function $f\in C(\bS^{d-1})$ such that $F\subset f^{-1}(1)$ and $G\subset f^{-1}(0)$. By \cref{eq:dual characterization of local algebra}, we have $\rho(f)A-A\rho(f)\in\calK(\calH_d)$, and hence
    \eq{
        \Lambda_F (\rho(f)A-A\rho(f))\Lambda_G \in \calK(\calH_d)\,.
    }
    Using \cref{eq:lambda is spectral measure}, it follows that
    \eq{
        \rho(\chi_F f) A \Lambda_G - \Lambda_F A \rho(f\chi_G) \in \calK(\calH_d)\,.
    }
    By construction of $f$, we have $\chi_F f=\chi_F$ and $f\chi_G = 0$. Thus, we arrived at $\Lambda_FA\Lambda_G\in\calK(\calH_d)$.

    We now prove the necessity of the statement. Let $f\in C(\bS^{d-1})$ be arbitrary. Fix $\ve>0$. There exists a partition of $\RR^d$ into cubes fine enough such that if we let $\{F_k\}_{k=1}^{n}$ be the collection of the sets of intersection of cube and $\bS^{d-1}$, then we have 
    \eql{\label{eq:function evaluate at cubes difference small}
    \sup_{x,y\in F_k} |f(x)-f(y)| < \ve,\quad \forall k\in\{1,\dots,n\}\,.
    }
    Pick any $x_k\in F_k$. It follows from \cite[Proposition IX.1.10]{conway2007course} that
    \eql{\label{eq:approximate spectral integral with cubes}
        \|\rho(f)-\sum_{k=1}^n f(x_k) E(F_k)\| < \ve\,.
    }
    Since $\bS^{d-1} = \cup_{k=1}^n F_k$ is a disjoint union, it follows that $\Id = \sum_{k=1}^n E(F_k)$ and, using \cref{eq:lambda is spectral measure}, we have
    \eq{
        A\rho(f)-\rho(f) A &= \sum_{j=1}^n\sum_{k=1}^n \Lambda_{F_j} (A\rho(f)-\rho(f)A)  \Lambda_{F_k} \\
        &= \sum_{j=1}^n \left(\sum_{k\sim j}+\sum_{k\not\sim j}\right) (\Lambda_{F_j} A\Lambda_{F_k}\rho(f) - \rho(f)\Lambda_{F_j} A\Lambda_{F_k})
    }
    where we have used the notation $k\sim j$ to denote those $k$ such that the cube $F_k$ is adjacent to $F_j$, and $k\not\sim j$ to those $k$ where $F_k$ is not adjacent to $F_j$. In particular, if $F_k$ and $F_j$ are not adjacent, then there are disjoint closed sets that separate them, and hence $\Lambda_{F_j} A\Lambda_{F_k}\in\calK(\calH_d)$ by assumption. Therefore, the operator relating to $\sum_{j=1}^n\sum_{k\not\sim j}$ is compact, and we denote it by $K_\ve$. Now we break up the terms
    \eq{
    A\rho(f)-\rho(f)A - K_\ve  = & \sum_{j=1}^n\sum_{k\sim j} \Lambda_{F_j}A\Lambda_{F_k}\left(\rho(f)-\sum_{l=1}^n f(x_l)\Lambda_{F_l}\right) \\
    + &  \sum_{j=1}^n\sum_{k\sim j} \Lambda_{F_j}A\Lambda_{F_k}\sum_{l=1}^n f(x_l)\Lambda_{F_l} \\
    - & \sum_{j=1}^n\sum_{k\sim j} \left(\rho(f)-\sum_{l=1}^n f(x_l)\Lambda_{F_l}\right) \Lambda_{F_j}A\Lambda_{F_k} \\
    - & \sum_{j=1}^n\sum_{k\sim j} \left(\sum_{l=1}^n f(x_l)\Lambda_{F_l}\right) \Lambda_{F_j}A\Lambda_{F_k}
    }
    and proceed to bound the operator norms of the first, second, third and the fourth terms individually. For the first term, let $T:= A(\rho(f)-\sum_{l=1}^n f(x_l)\Lambda_{F_l})$ for convenience, and let $\xi\in\calH_d$ be arbitrary, and we have
    \eq{
    \left\|\sum_{j=1}^n\sum_{k\sim j} \Lambda_{F_j}T\Lambda_{F_k}\xi \right\|^2 = \sum_{j=1}^n\left\| \Lambda_{F_j}T\sum_{k\sim j}\Lambda_{F_k}\xi \right\|^2 \leq \max_j \|\Lambda_{F_j} T\|^2 \sum_{j=1}^n \sum_{k\sim j} \|\Lambda_{F_k}\xi\|^2\,.
    }
    Observe that for each cube in $\RR^d$, there are $3^{d}-1$ cubes that surrounds the center cube. It follows that
    \eql{\label{eq:adjacent cubes estimate}
        \sum_{j=1}^n \sum_{k\sim j} \|\Lambda_{F_k}\xi\|^2 \leq 3^d \|\xi\|^2\,.
    }
    Thus, using the above \cref{eq:adjacent cubes estimate} and \cref{eq:approximate spectral integral with cubes}, the operator norm of the first term is bounded above by $\ve 3^{d/2}\|A\|$, and similarly, so is the operator norm of the third term. Now, consider the second and the fourth terms, let $\xi\in\calH_d$ be arbitrary, we have
    \eq{
    \left\|\sum_{j=1}^n \sum_{k\sim j} (f(x_k)-f(x_j)) \Lambda_{F_j}A\Lambda_{F_k} \xi\right\|^2 = & \sum_{j=1}^n \left\|\Lambda_{F_j} A \sum_{k\sim j}(f(x_k)-f(x_j))\Lambda_{F_k}\xi \right\|^2 \\
    \leq & \max_{j}\|\Lambda_{F_j}A\|^2 \sum_{j=1}^n |f(x_k)-f(x_j)|^2\|\Lambda_{F_k}\xi\|^2\,.
    }
    Using \cref{eq:function evaluate at cubes difference small} and \cref{eq:adjacent cubes estimate}, the operator norm of the second and fourth terms combined is bounded by $\ve 3^{d/2}\|A\|$.
    Thus
    \eq{
    \|A\rho(f)-\rho(f)A - K_\ve\| \leq \ve 3^{d/2-1} \|A\|\,.
    }
    We can construct a sequence of compact operators $K_\epsilon$ that converges to $A\rho(f)-\rho(f)A$ with $\epsilon\to 0$. Thus $A\rho(f)-\rho(f)A\in\calK(\calH_d)$.
\end{proof}

\begin{rem}\label{rem:dyadic cubes}
In \cref{thm:geometric characterization of locality}, we work with all closed subsets of $\bS^{d-1}$. In fact, a smaller collection of subsets suffices. To be concrete, let us work with the closed dyadic cubes, namely with those that have the form
\eql{\label{eq:closed dyadic cubes}
    \Set{(x_1,\dots,x_d)\in \RR^{d}:j_i2^{-k}\leq x_i \leq (j_i+1)2^{-k} \text{ for } i=1,\dots,d} \cap \bS^{d-1} 
}
for some integers $j_1,\dots,j_d$ and some positive integer $k$. We will refer to them as the closed dyadic cubes (in $\bS^{d-1}$). The collection of all closed dyadic cubes in $\bS^{d-1}$ is countable. 
For each positive integer $k$, we consider the collection $D_k$ of closed dyadic cubes of the $k$-th generation $D_k$ to be of the form \cref{eq:closed dyadic cubes} for some integers $j_1,\dots,j_d$. Then $D_k$ is a finite covering of $\bS^{d-1}$ (consisting of cubes whose interiors are disjoint), and if $k_1<k_2$, then each cube in $D_{k_2}$ is included in some cube in $D_{k_1}$.
\end{rem}

\begin{thm}\label{thm:dyadic cubes characterization of locality}
    An operator $A\in\calB(\calH_d)$ is spherically-local iff $\Lambda_F A\Lambda_G\in\calK(\calH_d)$ for every disjoint pair of closed dyadic cubes $F,G$ in $\bS^{d-1}$.    
\end{thm}
\begin{proof}
    The proof is the same as \cref{thm:geometric characterization of locality}.
\end{proof}

\subsection{Index formulae}

Let us compute the $K$-groups of the spherically-local algebra $\calL_d$. For those with nontrivial $K$-groups, we discuss their index formulae.

\begin{prop}\label{prop:K-groups calculation of local algebra}
The $K$-groups of the algebra $\calL_d$ are
\eq{
    K_0(\calL_d)
    =\begin{cases}
        0 & \text{$d$ is odd} \\
        \ZZ & \text{d is even}
    \end{cases},\quad
    K_1(\calL_d) =\begin{cases}
        \ZZ & \text{$d$ is odd} \\
        0 & \text{d is even}
    \end{cases}\,.
}
Every element in $K_0(\calL_d)$ is the class $[P]_0$ of a projection $P$ in $\calP(\calL_d)$ and every element in $K_1(\calL_d)$ is the class $[U]_1$ of a unitary $U$  in $\calU(\calL_d)$.
\end{prop}
\begin{proof}
From \cref{eq:dual characterization of local algebra}, we have $\calL_d=\calD_\rho(C(\bS^{d-1}))$.
Then
\eq{
    K_1(\calD_\rho(C(\bS^{d-1}))) \cong K^0(C_0(\RR^{d-1})) \cong K^0(S^{d-1}\CC) =\begin{cases}
        \ZZ & \text{$d$ is odd} \\
        0 & \text{d is even}
    \end{cases}
}
where in the first isomorphism, we use the definition of $K$-homology groups (\cite[Definition 5.2.1]{higson2000analytic}) and the fact that the $C(\bS^{d-1})$ is the unitization of $C_0(\RR^{d-1})$; in the second isomorphism, we use that $C_0(\RR^{d-1})\cong S^{d-1}\CC$; and in the third isomorphism, we repeatedly use the Bott periodicity (with $S$ being the suspension) to reduce the calculations to the cases of $K^1(\CC)$ or $K^0(\CC)$, which are given by
\eq{
K^1(\CC)=0,\quad K^0(\CC)=\ZZ\,.
}
Similarly, we also have
\eq{
    K_0(\calD(C(\bS^{d-1}))) \cong K^1(C_0(\RR^{d-1})) = K^1(S^{d-1}\CC)=\begin{cases}
        0 & \text{$d$ is odd} \\
        \ZZ & \text{d is even}
    \end{cases}\,.
}

We proceed to the second part of the claim. The argument follows from \cite[Proposition 5.1.4 and Remark 5.1.5]{higson2000analytic}. Let $[P]_0-[Q]_0$ be an element in $K_0(\calL_d)$, where $P\in\calP_n(\calL_d)$ and $Q\in\calP_m(\calL_d)$. Using $[\Id_m]_0=0$ from \cref{lem:projection MvN amplification}, it follows that $-[Q]_0 = [Q^\perp]_0$. Thus $[P]_0-[Q]_0 = [P\oplus Q^\perp]_0$. 
Using $\Id_{n+m}\sim_0 \Id$ from \cref{lem:projection MvN amplification}, there exists $V\in M_{1,n+m}(\calL_d)$ such that $\Id_{n+m}=V^\ast V$ and $\Id = VV^\ast$. Let $T:= V(P\oplus Q^\perp)V^\ast$. Then $T\in \calP(\calL_d)$ and $[T]_0=[P\oplus Q^\perp]_0=[P]_0-[Q]_0$.

 Suppose an element $[U]_1\in K_1(\calL_d)$ is represented by $U\in\calU_n(\calL_d)$ for some $n\in\NN$. We need to find $T\in\calU(\calL_d)$ such that $[T]_1=[U]_1$. The proof technique follows \cite[Remark 5.1.5]{higson2000analytic} and \cite[Lemma 2.11]{lee2007note}. Using \cref{lem:projection MvN amplification}, we have $\Id_n\sim_0\Id$ and hence there exists an element $V\in M_{1,n}(\calL_d)$ such that $\Id_n = V^\ast V$ and $\Id = VV^\ast$. Consider the elements
    \eq{
    W= \begin{bmatrix}
        1 \\
        0 \\
        \vdots\\
        0
    \end{bmatrix} V\in M_n(\calL_d),\quad S = \begin{bmatrix}
        W & \Id_n-WW^\ast \\
        0  & W^\ast
    \end{bmatrix}\in M_{2n}(\calL_d)\,.
    }
    It is straightforward to check that $S\in \calU_{2n}(\calL_d)$ and
    \eq{
    WUW^\ast + \Id_n-WW^\ast &= \begin{bmatrix}
        VUV^\ast & 0 \\
        0 & \Id_{n-1}
    \end{bmatrix}\in \calU_n(\calL_d) \\ S\begin{bmatrix}
        U & 0 \\
        0 & \Id_n
    \end{bmatrix} S^\ast &= \begin{bmatrix}
        WUW^\ast + \Id_n-WW^\ast & 0 \\
        0 & \Id_n
    \end{bmatrix}\in \calU_{2n}(\calL_d)\,.
    }
    Thus we have
    \eq{
    [VUV^\ast]_1 = [WUW^\ast + \Id_n-WW^\ast]_1 = [S (U\oplus \Id_n)S^\ast]_1 = [S]_1 + [U]_1 + [S^\ast]_1 = [U]_1
    }
    and $VUV^\ast \in \calU(\calL_d)$ is the sought-after operator.
\end{proof}

For an operator $A\in\calB(\calH_d)$ and a positive integer $m\in\NN$, let us for convenience write 
\eql{\label{eq:notation tensor identity}
A_{(m)}:=A\otimes \Id_m
}
which is an operator on $\calB(\calH_d\otimes \CC^m)\cong \calB(\calH_d)\otimes \CC^m \cong M_m(\calB(\calH_d))$.

There are natural index pairings of $K_0(\calX_d)\times K_1(\calL_d)\to \ZZ$ and of $K_1(\calX_d)\times K_0(\calL_d)\to \ZZ$. The pairings are bilinear maps given by
\eql{\label{eq:0 bilinear index map}
K_0(\calX_d)\times K_1(\calL_d)\ni ([P]_0,[U]_1)\mapsto \findex(PU_{(m)}P+P^\perp) \in \ZZ
}
for all $P\in\calP_m(\calX_d)$, $m\in\NN$ and $ U\in\calU(\calL_d)$, and
\eql{\label{eq:1 bilinear index map}
K_1(\calX_d)\times K_0(\calL_d)\ni ([U]_1,[P]_0)\mapsto \findex(P_{(m)}U P_{(m)} +P_{(m)}^\perp) \in \ZZ
}
for all $U\in\calU_m(\calX_d)$, $m\in\NN$ and $ P\in\calP(\calL_d)$.
The pairings come from the treatments of $K$-groups of $\calL_d$ as the $K$-homology groups $C(\bS^{d-1})\cong\calX_d$. See \cite[Chapter 7.2]{higson2000analytic} for more detail. 

\begin{rem}
    In the index pairing of $K_0(\calX_d)$ with $K_1(\calL_d)$, we only consider element in $K_1(\calL_d)$ of the form $[U]_1$ where $U\in\calU(\calL_d)$. This suffices by \cref{prop:K-groups calculation of local algebra}. Similarly, in the index pairing of $K_1(\calX_d)$ with $K_0(\calL_d)$, we only consider element in $K_0(\calL_d)$ of the form $[P]_0$ where $P\in\calP(\calL_d)$.
    The Fredholm index formulae make sense by construction.
    Indeed, for the first pairing above, we have $[\tensorid{U}{m},P]\in\calK$, and $P\tensorid{U^\ast}{m}P+P^\perp$ is the parametrix for $PU_{(m)}P+P^\perp$, which is therefore Fredholm.
    Analogously, for the second pairing, the operator $P_{(m)}U P_{(m)} +P_{(m)}^\perp$ is Fredholm.
    \end{rem}

Let $k=\floor{d/2}$. Let us first discuss the case when $d\in 2\NN$. Consider the Dirac phase $L_d$ \cref{eq:dirac phase} for which the class $[L_d]_1$ generates $K_1(\calX_d)\cong \ZZ$, see \cite[Proposition 1]{SCHULZBALDES2023125519}. The index pairing induces the index homomorphism
\eql{\label{eq:K0 Ld index homomorphism}
K_0(\calL_d)\ni [P]_0 \mapsto \findex( P_{(2^{k-1})} L_d P_{(2^{k-1})}+ P_{(2^{k-1})}^\perp) \in \ZZ \,.
}
For $d\in 2\NN+1$, we consider the Dirac projection $\Lambda_d$ \cref{eq:dirac projection} for which the class $[P]_0$ generates the non-trivial summand in $K_0(\calX_d)\cong \ZZ\oplus \ZZ$. Then the index pairing induces the index homomorphism
\eql{\label{eq:K1 Ld index homomorphism}
K_1(\calL_d)\ni [U]_1 \mapsto \findex(\Lambda_d U_{(2^k)}\Lambda_d+\Lambda_d^\perp) \in \ZZ\,.
}
In fact, more is true.

\begin{prop}\label{prop:index and K-groups}
    The homomorphisms \cref{eq:K0 Ld index homomorphism} and \cref{eq:K1 Ld index homomorphism} are isomorphisms.
\end{prop}
\begin{proof}
    Denote $\calA_d:= C_0(\RR^{d-1})$ for convenience.
    The index homomorphisms \cref{eq:K0 Ld index homomorphism} and \cref{eq:K1 Ld index homomorphism} originate from the pairing between $K$-theory group $K_j(\calA_d)$ and the $K$-homology $K^j(\calA_d)$ that lead to the bilinear maps \cref{eq:1 bilinear index map} and \cref{eq:0 bilinear index map} respectively.
    Indeed, these follow from noting $K_1(\calL_d)=K_1(\calD_\rho(\widetilde{\calA}_d))=K^0(\calA_d)$ and $K_0(\calL_d)=K_0(\calD_\rho(\widetilde{\calA}_d))=K^1(\calA_d)$; and that $K_1(C(\bS^{d-1}))\cong K_1(\calA_d)$, and $K_0(\calA_d)\cong \widetilde{K}_0(C(\bS^{d-1}))$. 
    To show that the index homomorphisms are in fact isomorphisms, we invoke the universal coefficient theorem \cite[Theorem 7.6.1]{higson2000analytic}: for each $j\in\{0,1\}$ there is a natural short exact sequence
\[
0\longrightarrow \Ext^{1}_{\ZZ}\bigl(K_{j-1}(\calA_d),\ZZ\bigr)
\longrightarrow
K^{j}(\calA_d)
\longrightarrow
\mathrm{Hom}_{\ZZ}\bigl(K_{j}(\calA_d),\ZZ\bigr)
\longrightarrow 0,
\]
where the right-hand map is the index homomorphism. The $K$-groups of $\calA_d$ are
\[
K_{0}(\calA_d)=
\begin{cases}
\ZZ & d-1\text{ odd}\\
0 & d-1\text{ even}
\end{cases},
\qquad
K_{1}(\calA_d)=
\begin{cases}
0 & d-1\text{ odd}\\
\ZZ & d-1\text{ even}
\end{cases}.
\]
In particular, the non-zero group $K_{j}(\calA_d)$ is free abelian of rank $1$, hence
$\Ext^{1}_{\ZZ}(K_{j-1}(A),\ZZ)=0$ in the relevant parity.
Therefore, we have the index isomorphisms 
\eq{
K^j(\calA_d)\cong \mathrm{Hom}_{\ZZ}(K_{j}(A),\ZZ).
}
Evaluating at the generators for $K_j(\calA_d)$, which are exactly the Dirac phase and projection when $d$ is even and odd respectively (see \cite[Proposition 1]{SCHULZBALDES2023125519}), give the desired result.
\end{proof}

\begin{lem}\label{lem:projection MvN amplification}
    Let $P\in\calP_n(\calL_d)$ be a spherically-local projection such that $P\rho^n(a)P$ is never compact unless $a=0$. Then
    \eq{
    P\oplus \Id_m \sim_0 P,\quad \forall m\in\NN
    }
    i.e., there exists an element $V\in M_{n,n+m}(\calL_d)$ such that $P\oplus \Id_m = V^\ast V$ and $P=VV^\ast$. In particular, we have $\Id_n\sim_0 \Id_m$ for all $n,m\in\NN$, and $[\Id]_0=0$ in $K_0(\calL_d)$.
\end{lem}
\begin{proof}
    Let $P\in\calP_n(\calL_d)$ be spherically-local and $P\rho^n(a)P$ be never compact unless $a=0$. Since $[P,\rho^n(a)]\in\calK(\calH_d)$ for all $a\in C(\bS^{d-1})$ by construction, we may consider the $\ast$-homomorphism 
    \eq{
    \varphi_P:C(\bS^{d-1})\to \calQ(\im P)
    } 
    defined by $a\mapsto \pi(P\rho^n(a)P)$, see \cite[Definition 2.7.7]{higson2000analytic}. Since $P\rho^n(a)P$ is never compact unless $a=0$, it follows that $\varphi_P$ is injective, and hence is an extension of the compact operators by $C(\bS^{d-1})$, see \cite{arveson1977notes}. Since $\varphi_{\Id_m}$ is a trivial extension, it follows from Voiculescu's Theorem that $\varphi_P\oplus\varphi_{\Id_m}$ is unitarily equivalent to $\varphi_P$. By \cite[Lemma 5.1.2]{higson2000analytic}, the projections $P\oplus \Id_m$ and $P$ are Murray-von Neumann equivalent.
\end{proof}

\section{Bulk non-triviality}\label{sec:Bulk non-triviality}
In this section we describe a novel constraint which, roughly speaking, corresponds to the system being a non-trivial insulator ``in all space dimensions'', and hence, the system is a ``bulk'' non-trivial system. This term is to be contrasted with edge systems, where part of space is trivial. Usually one considers edge systems corresponding to a division into two half-spaces (say the upper and lower half planes in $d=2$). However, it turns out, that for the purposes of topological classification, a more fine-grained notion of infinitely-extending to all directions of space is necessary. We have first introduced this concept in the context of one-dimensional classification in \cite{ChungShapiro2023}. Below we extend it to all higher dimensions in a compatible way.

In zero dimensions, the idea is as follows. We are interested in insulators, and WLOG we assume Hamiltonians are gapped at $E_F=0$, i.e., invertible operators. Then clearly we cannot deform Hamiltonians which only have spectrum above $0$ to those that only have spectrum below $0$ without closing the gap or breaking self-adjointness. Moreover, this remains true if there is only \emph{essential} spectrum above $0$ or only essential spectrum below $0$. We call such insulators \emph{non-trivial}. Thus, in higher space dimensions (as we shall describe momentarily) we employ this idea in all space directions.


In \cite{ChungShapiro2023} we identified that such a condition should go hand in hand with locality, in the following sense. If our locality condition in $d=1$, using the operator $\Lambda = \chi_\NN(X_1)$ (which is equal to $\Lambda_{\{1\}}$ in the notation of \cref{eq:Lambda_F projection definition} since $\bS^0=\{-1,1\}$), meant that the connection pieces between left and right halves of the system should be compact, then it turned out that to have honestly bulk systems, we should ask the system to be a non-trivial insulator, at least \emph{essentially}, in each half of the system separately. This led us to
\begin{defn}[bulk non-triviality in $d=1$]\label{def:Lambda non triviality} A $\Lambda$-local self-adjoint projection $P$ is called bulk-non-trivial with respect to $\Lambda$ iff $\sigma_{\mathrm{ess}}(\Lambda P \Lambda:\im\Lambda\to\im\Lambda) = \Set{0,1}$ and $\sigma_{\mathrm{ess}}(\Lambda^\perp P \Lambda^\perp:\im\Lambda^\perp\to\im\Lambda^\perp) = \Set{0,1}$.
\end{defn}
It is important in the above definition that the operator $\Lambda P\Lambda$ is viewed in the subspace $\im\Lambda$ (and similarly $\Lambda^\perp P\Lambda^\perp$ be viewed in $\im\Lambda^\perp$).

We also established the chiral systems, those whose Fermi projection is of the form \eql{\label{eq:chiral projections}
P = \frac12\br{\begin{bmatrix}
    0 & U^\ast \\ U & 0
\end{bmatrix}+\Id}
} for some unitary $U$, are automatically bulk-non-trivial with respect to $\Lambda\oplus \Lambda$ iff $[U,\Lambda]\in\calK$. 

What, then, is the appropriate notion of bulk non-triviality in higher dimensions? As we discussed above, using the Dirac projection for $\Lambda$ leads to unconvincing results.

Instead, we follow spherical locality to formulate 
\begin{defn}[bulk non-triviality]\label{defn:bulk non triviality}
    Let $P\in\calP(\calL_d)$ be a spherically-local projection. We call $P$ bulk-non-trivial iff for all non-empty open subsets $I$ of $\bS^{d-1}$, both the operators $\Lambda_I P\Lambda_I$ and $\Lambda_I P^\perp \Lambda_I$ are \emph{not} compact. We denote the set of bulk-non-trivial projections to be $\calP^\NT(\calL_d)$.
\end{defn}

\revadd{The condition is deliberately formulated direction-by-direction at infinity. It rules out projections which are non-trivial only along a lower-dimensional subset of directions, or which become essentially empty or essentially full in some open solid angle. The examples below show that these non-bulk configurations create path-components not predicted by the strong Kitaev table; hence bulk non-triviality is not a cosmetic assumption but part of the definition of the correct bulk phase space.}

\begin{prop}
    If $d=1$, then \Cref{defn:bulk non triviality} and \Cref{def:Lambda non triviality} are equivalent.
\end{prop}
\begin{proof}
Let $P\in\calP(\calL_1)$. Suppose $P$ is bulk-non-trivial in the sense of \cref{defn:bulk non triviality}. Recall the notation $\Lambda=\chi_\NN(X_1)=\Lambda_{\{1\}}$, and $\Lambda^\perp = \Lambda_{\{-1\}}$. Then the operators 
\eql{\label{eq:1d projections to half line}
\Lambda P\Lambda,\Lambda^\perp P\Lambda^\perp,\Lambda P^\perp\Lambda,\Lambda^\perp P^\perp \Lambda^\perp
}
are not compact. 
In one-dimension, since $[\Lambda,P]\in\calK(\calH_1)$, it follows that the operators \cref{eq:1d projections to half line} are essential projections, i.e., they are orthogonal projections in the Calkin algebra. Indeed, take the operator $\Lambda P\Lambda$ for instance, we have
\eq{
\pi(\Lambda P\Lambda)^2 = \pi(\Lambda P\Lambda P\Lambda) = \pi(\Lambda [P,\Lambda] P\Lambda + \Lambda P\Lambda) = \pi(\Lambda P\Lambda),\quad \pi(\Lambda P\Lambda)^\ast = \pi(\Lambda P\Lambda)
}
where $\pi$ is the quotient map to the Calkin algebra.
In particular, $\sigmaess(\Lambda P\Lambda:\im\Lambda\to\im\Lambda)$ must be one of $\{0\},\{1\}$ or $\{0,1\}$. 
Since $\Lambda P\Lambda$ is not compact, we rule out the case $\{0\}$. Suppose it is $\{1\}$. Then $\sigmaess(\Lambda-\Lambda P\Lambda:\im\Lambda\to\im\Lambda)=\{0\}$. This leads to the contradiction that $\Lambda P^\perp\Lambda$ is compact.
Thus it must be the case that $\sigmaess(\Lambda P\Lambda:\im\Lambda\to\im\Lambda)=\{0,1\}$.
Analogously, one can verify $\sigmaess(\Lambda^\perp P\Lambda^\perp:\im\Lambda^\perp\to\im\Lambda^\perp)=\{0,1\}$.

Conversely, suppose $P$ is bulk-non-trivial in the sense of \cref{def:Lambda non triviality}. It follows from definition that $\Lambda P\Lambda$ and $\Lambda^\perp P\Lambda^\perp$ are not compact. Suppose $\Lambda P^\perp \Lambda$ is compact. Then $\sigmaess(\Lambda P^\perp \Lambda:\im\Lambda\to\im\Lambda)=\{0\}$ and hence 
\eq{
\sigmaess(\Lambda P\Lambda:\im\Lambda\to\im\Lambda)=
\sigmaess(\Lambda-\Lambda P^\perp \Lambda:\im\Lambda\to\im\Lambda)=\{1\}\,.
}
This contradicts that we assume $\sigmaess(\Lambda P\Lambda:\im\Lambda\to\im\Lambda)=\{0,1\}$. Analogously, one verify that $\Lambda^\perp P^\perp \Lambda^\perp$ is not compact, or else $\sigmaess(\Lambda^\perp P\Lambda^\perp:\im\Lambda^\perp\to\im\Lambda^\perp)=\{1\}$, leading to a contradiction.
\end{proof}

\begin{prop}[Chiral systems are automatically bulk-non-trivial] If $P$ is a projection of the form \cref{eq:chiral projections} for some unitary $U$ which is spherically-local, then $P$ obeys \Cref{defn:bulk non triviality}.
\end{prop}
\begin{proof}
Let $\calH:=\calH_d\otimes \CC^N$ and view the chiral projection $P$ in
\cref{eq:chiral projections} as acting on the graded Hilbert space
$\calH\oplus \calH\cong \calH\otimes \CC^2$.
We will write $\Lambda_I$ also for the projection $\Lambda_I\oplus \Lambda_I$
(on $\calH\oplus\calH$); this is consistent with the convention that $\Lambda_I$
acts only on the $\ell^2(\ZZ^d)$ factor and as the identity on internal degrees.

First we establish the locality of $P$. Fix $j\in\{1,\dots,d\}$ and write $\widehat{X}_j$ for $\widehat{X}_j\otimes \Id_N$
(on $\calH$) and also for $\widehat{X}_j\oplus \widehat{X}_j$ (on $\calH\oplus\calH$).
Since $U$ is spherically-local, $[U,\widehat{X}_j]\in\calK(\calH)$ and hence also
$[U^\ast,\widehat{X}_j]\in\calK(\calH)$. Using \cref{eq:chiral projections},
\[
\Bigl[P,\widehat{X}_j\Bigr]
=\frac12
\begin{bmatrix}
0 & [U^\ast,\widehat{X}_j]\\
[U,\widehat{X}_j] & 0
\end{bmatrix}\in \calK(\calH\oplus\calH).
\]
Thus $P\in\calP(\calL_d)$ (and similarly $P^\perp\in\calP(\calL_d)$).

Next, we show that $\Lambda_I$ has infinite rank for every non-empty open $I\subset\bS^{d-1}$. Indeed, there are infinitely many $x\in \ZZ^d$ for which $\hat{x}:=x/\|x\|\in I$ if $I$ is non-empty and open subset of $\bS^{d-1}$.

Finally, we show that $\Lambda_I P\Lambda_I$ and $\Lambda_I P^\perp\Lambda_I$ are not compact.
Using \cref{eq:chiral projections} and the fact that $\Lambda_I$ acts diagonally on
$\calH\oplus\calH$,
\[
\Lambda_I P\Lambda_I
=\frac12
\begin{bmatrix}
\Lambda_I & \Lambda_I U^\ast \Lambda_I\\
\Lambda_I U \Lambda_I & \Lambda_I
\end{bmatrix},
\qquad
\Lambda_I P^\perp\Lambda_I
=\frac12
\begin{bmatrix}
\Lambda_I & -\Lambda_I U^\ast \Lambda_I\\
-\Lambda_I U \Lambda_I & \Lambda_I
\end{bmatrix}.
\]
Let $E_{11}:=\Id_\calH\oplus 0$ be the projection onto the
first chiral block. We have
\[
E_{11}(\Lambda_I P\Lambda_I)E_{11}=\frac12\,\Lambda_I,\qquad
E_{11}(\Lambda_I P^\perp\Lambda_I)E_{11}=\frac12\,\Lambda_I.
\]
Since $\Lambda_I$ is not compact by the above, neither $\Lambda_I P\Lambda_I$ nor
$\Lambda_I P^\perp\Lambda_I$ can be compact.

\end{proof}

\begin{lem}\label{lem:alternative bulk non triviality}
    Let $P$ be a spherically-local projection. Then $P$ is bulk-non-trivial iff $P\rho(a)P$ and $P^\perp \rho(a)P^\perp$ are never compact unless $a=0$, where $\rho$ is the isomorphism from \Cref{lem:maximal ideal space algebra generated by Xis}.
\end{lem}
\begin{proof}
    Suppose $P$ is not bulk-non-trivial. There exists an open subset $I$ of $\bS^{d-1}$ such that $\Lambda_I P\Lambda_I\in\calK(\calH_d)$. There exists a non-zero continuous function $f\in C(\bS^{d-1})$ such that $f(z)=0$ for $z\in I^c$; e.g., we can apply Urysohn's lemma to the sets $\{z_0\}$ and $I^c$ for some $z_0\in I$. Consider the decomposition
    \eq{
    P\rho(f)P=(\Lambda_I+\Lambda_{I^c})P\rho(f)P(\Lambda_I+\Lambda_{I^c})\,.
    }
    Recall that $[P,\rho(f)]\in\calK(\calH_d)$ since $P$ is spherically-local, and hence we may interchange $P$ and $\rho(f)$ modulo a compact operator. Also $[\Lambda_I,\rho(f)]=0$ using \cref{eq:lambda is spectral measure}. 
    Thus $\Lambda_IP\rho(f)P\Lambda_I = \Lambda_IP\Lambda_I \rho(f)+K$ holds for some $K\in\calK(\calH_d)$ and is compact by assumption. And the terms involving $\Lambda_{I^c}$ vanish since $\Lambda_{I^c}\rho(f)=0$.

    Suppose $P\rho(f)P\in\calK(\calH_d)$ for some non-zero $f\in C(\bS^{d-1})$. Then, there exists a point $z\in \bS^{d-1}$ such that $f(z)\neq 0$; furthermore, there exists an open neighborhood $I\subset \bS^{d-1}$ of $z$ such that $|f(z)|\geq c>0$ for some positive number $c$.
    We can construct a continuous function $g$ such that $g(z)=1/f(z)$ for $z\in I$. Then
    \eq{
    \Lambda_IP\rho(f)\rho(g)\Lambda_I =  \Lambda_IP\rho(fg)\Lambda_I=\Lambda_IP\Lambda_I
    }
    is compact since $P\rho(f)$ is compact by assumption. This leads to a contradiction.
\end{proof}

\begin{example}[Systems that fail bulk non-triviality in $d=1$ violate the \nameref{table:Kitaev}] First we present a one-dimensional example on $\ell^2(\ZZ)$:
Consider the two Hamiltonians \eq{
    H_1 &:= \Lambda - \Lambda^\perp \\
    H_2 &:= -\Lambda + \Lambda^\perp\,.
} The corresponding Fermi projections are given by
\eq{
    P_1 &= \Lambda^\perp \\
    P_2 &= \Lambda\,.
} Writing them in the grading $\ell^2(\ZZ)\cong\im(\Lambda^\perp)\oplus\im(\Lambda)$ yields 
\eq{
    P_1 &= \Id\oplus 0 \\
    P_2 &= 0\oplus\Id
} whence it is clear that neither is bulk-non-trivial according to \cref{def:Lambda non triviality}. Moreover, it is impossible to interpolate between $P_1$ and $P_2$ in a path that essentially commutes with $\Lambda$ (i.e., a local path), as we showed in \cite[Example 5.6]{ChungShapiro2023}. This stands in contradiction to the top left cell of the \nameref{table:Kitaev}, which stipulates that class A $d=1$ systems are path-connected.
\end{example}

\begin{example}[Systems that fail bulk non-triviality in $d=2$ violate the \nameref{table:Kitaev}] 
Let
\eq{
I_{+}&:=\{(\cos\theta,\sin\theta)\in\bS^{1}:\ -\tfrac{\pi}{2}<\theta<\tfrac{\pi}{2}\}
=\{ \omega\in\bS^1:\ \omega_1>0\},\\
I_-&:=\bS^1\setminus \overline{I_+}=\{\omega\in\bS^1:\ \omega_1<0\}.
}
Define two projections on $\ell^2(\ZZ^2)$ by
\[
P_1:=\Lambda_{I_+},\qquad P_2:=\Lambda_{I_-}.
\]
These are self-adjoint projections and are spherically-local (in fact $[P_i,\widehat X_j]=0$ for
$j=1,2$) since they are multiplication operators in the position basis depending only on the
direction $\hat x=x/\|x\|$.


Let $L := \exp(\ii \arg(X_1+\ii X_2))$ be the Laughlin flux insertion unitary operator. Since $[P_i,L]=0$ for $i=1,2$, it follows that $\findex (\PP_1 L)=0=\findex(\PP_2 L)$, i.e., the two systems have the same Chern numbers. Thus the index invariant alone would not distinguish them in the space of spherically-local projections. Suppose such homotopy exists, then, by \cite[Proposition 2.2.6]{rordam2000introduction}, there exists $U\in\calU(\calL_d)$ such that $P_1=U^\ast P_2 U$. Take any non-empty closed arc $J\subset I_+$ for which $\Lambda_J$ has infinite rank, and let $K:=\overline{I_-}$. Since $U$ is spherically-local and that $J$ and $K$ are disjoint closed subsets of $\bS^{1}$, using \cref{thm:geometric characterization of locality}, we have $\Lambda_{K}U\Lambda_J\in\calK$. Then 
\eq{
\Lambda_J=P_1\Lambda_J = U^\ast P_2 U \Lambda_J = U^\ast P_2 \Lambda_K U\Lambda_J\in\calK
}
which is a contradiction since $\Lambda_J$ is not compact.

Both $P_1$ and $P_2$ fail bulk non-triviality in the sense of
\Cref{defn:bulk non triviality}. Indeed, take any non-empty open arc $J\subset I_+$.
Then $\Lambda_J\le \Lambda_{I_+}=P_1$ and $\Lambda_J\perp \Lambda_{I_-}=P_2$, so
\[
\Lambda_J P_1^\perp \Lambda_J = 0 \in \calK,
\qquad
\Lambda_J P_2 \Lambda_J = 0 \in \calK.
\]
Thus $P_1,P_2\notin\calP^{\NT}(\calL_2)$ even though their index invariant agrees. Consequently,
if one attempted to classify \emph{all} local class~A $d=2$ projections using only the Chern/Fredholm
index, these “half-space-at-infinity’’ projections would contribute extra path-components not
accounted for by the \nameref{table:Kitaev}; this is precisely why the bulk non-triviality
constraint is necessary. 
\end{example}



\section{The classification of bulk-non-trivial spherically-local projections and unitaries}\label{sec:The classification of bulk-non-trivial spherically-local projections and unitaries}


In this section, we shall prove the following two theorems about classification of unitaries and projections. 

For readability, we separate the proof into two layers.
In \cref{subsec:proofs of complex classification theorems} we show how the classification theorems follow formally from three key
functional-analytic propositions.
The proofs of these propositions, and the decoupling lemmas they depend on, are deferred to \cref{subsec:technical lemmas}.
This allows the reader to first see the logical skeleton of the proof and then consult the
geometric/analytic input as needed.

Recall the notations $\tensorid{A}{m}=A\otimes \Id_m$ defined in \cref{eq:notation tensor identity} and let $k=\floor{d/2}$.

\begin{thm}\label{thm:path-connected components of local unitaries}
    The set of path-connected components of $\calU(\calL_d)$ is given, via bijection, as
    \eq{
    \pi_0(\calU(\calL_d)) \cong \begin{cases}
        \ZZ & d\in2\NN+1 \\
        \Set{0} & d\in2\NN
    \end{cases}\,. 
    }
    If $d\in2\NN+1$, the index is given by 
    \eq{
    \calN:\calU(\calL_d)\ni U \mapsto \findex(\Lambda_d U_{(2^k)}\Lambda_d+\Lambda_d^\perp) \in \ZZ
    }
    where $\Lambda_{d}$ is the Dirac projection defined in \cref{eq:dirac projection}.
\end{thm}

\begin{thm}\label{thm:path-connected components of local bulk-non-trivial projections}
    The set of path-connected components of $\calP^\NT(\calL_d)$ is given, via bijection, as
    \eq{
    \pi_0\br{\calP^\NT(\calL_d)} \cong \begin{cases}
        \Set{0} & d\in2\NN+1 \\
        \ZZ & d\in2\NN
    \end{cases}\,. 
    }
    If $d\in2\NN$, the index is given by
    \eq{
    \calN:\calP^\NT(\calL_d)\ni P \mapsto \findex( P_{(2^{k-1})} L_d P_{(2^{k-1})}+ P_{(2^{k-1})}^\perp) \in \ZZ 
    }
    where $L_d$ is the Dirac phase defined in \cref{eq:dirac phase}.
\end{thm}

We define a countable set $\bS^{d-1}_\ZZ\subset \bS^{d-1}$ to be all points in $\bS^{d-1}$ that is equal to some $\hat{x}$ for $x\in \ZZ^{d}$. 

Here, we define some subsets of $\ZZ^d$. For any subset $J$ in $\bS^{d-1}$, we define $C_J\subset \ZZ^d$ to be all points $x\in \ZZ^d$ such that $\hat{x}\in J$. For $r>0$ a positive number, we define $B_r\subset \ZZ^d$ as those points $x\in \ZZ^d$ such that $\|x\|< r$. For $\varphi\in\calH_d$, we define $\supp(\varphi)\subset \ZZ^d$ to be those points $x\in \ZZ^d$ such that $\varphi_x\neq 0$. 

If $S\subset \ZZ^d$, we denote by $\Lambda_S$ the projection operator $\Lambda_S = \sum_{x\in S}\delta_{x}\otimes \delta_x^\ast$.

Let $P$ be a projection and $A$ a bounded operator. We say that $\im P$ \emph{reduces} $A$ if both $\im P$ and $\im P^\perp$ are invariant under $A$.

The general strategy to prove these theorems is a generalization of \cite{chung2025essentially}. For $U\in\calU(\calL_d)$, we want to show that if $\calN (U)=0$ has index 0, or equivalently, its $\Kone$-group $[U]_1=0$ is trivial, then $U\sim_h\Id$ in $\calU(\calL_d)$. The proof consists of two steps. The first one is purely functional analytic involving only the structure of spherical locality. For any $U\in\calU(\calL_d)$, we show that there exists a ``spherically-proper'' projection $\Lambda_E$ where $E\subset \ZZ^d$ such that
\eq{
U\sim_h W=\Lambda_E W \Lambda_E+\Lambda_E^\perp \quad \mathrm{in}\ \calU(\calL_d).
}
In other words, we may deform $U$ to some spherically-local unitary operator $W$ that acts as identity on the vector subspace $\im\Lambda_E^\perp = \szpan \Set{\delta_x | x\in E^{c}}$ with respect to a large, infinite set of lattice points $\Set{x\in\ZZ^d | x\in E^c}$. The infinite subspace $\im\Lambda_E^\perp$ provides us with a leeway to compress the homotopy occurring inside the matrix algebra over $\calL_d$. This brings us to the second step of the proof. If $[U]_1=0$ and hence $[W]_1=0$, then, by construction of $\Kone$-group, there exists $n\in\NN$ such that
\eq{
W\oplus \Id_n \sim_h \Id_{n+1} \quad \mathrm{in}\ \calU_{n+1}(\calL_d).
}
The homotopy $\gamma_t$ occurs in the matrix algebra $\calU_{n+1}(\calL_d)$ and we would like it to be completely inside the original space $\calU(\calL_d)$. To that end, we construct a spherically-local ``re-dimerization'' operator $V\in M_{1,n+1}(\calL_d)$ such that
\eq{
V (W\oplus \Id_n)V^\ast =W,\quad VV^\ast = \Id,\quad V\gamma_t V^\ast\in \calU(\calL_d)
}
which completes the proof. 

The first step is detailed in \cref{prop:local unitary reduced by geo proper proj} while the second step in \cref{prop:compress unitary homotopy from matrix algebra back to original algebra}.

Let us introduce formally the spherically-proper projection.



\begin{defn}[spherically-proper]\label{defn:spherically-proper}
    We call a subset $F\subset \ZZ^d$ spherically-proper iff for every non-empty open subset $I\subset \bS^{d-1}$, the cone $C_I:=\Set{x\in \ZZ^d | x/\|x\|\in I}$ satisfies:
    \eq{
    |F\cap C_I|=\infty,\quad |F^c\cap C_I|=\infty.
    }
    We call a projection $P$ spherically-proper if $P=\Lambda_F$ for some spherically-proper set $F\subset \ZZ^d$.
\end{defn}

\begin{rem}[bulk-non-trivial and spherically-proper]\label{rem:bulk-non-trivial and spherically-proper}
    Let $P=\Lambda_F$ for some $F\subset\ZZ^d$. Then $P$ is spherically-proper iff it is  bulk-non-trivial as in \cref{defn:bulk non triviality}. 
    Indeed, if $\Lambda_F$ is spherically-proper, for $I\subset\bS^{d-1}$ open, then we have $\Lambda_I\Lambda_F\Lambda_I = \Lambda_{C_I\cap F}$ and $\Lambda_I\Lambda_F^\perp\Lambda_I = \Lambda_{C_I\cap F^c}$, both of which have infinite-dimensional range and hence are non-compact. The converse is similar.
\end{rem}

\begin{example}[spherically-proper sets in $d=1$ and $d=2$]\label{ex:shperically-proper}
    For $d=1$, let $F\subset \ZZ$ be the set of all even integers. Then $F$ is spherically-proper. For $d=2$, we pick a single point in each rational ray $C_{\{z\}}$ for $z\in \bS^1_{\ZZ}$, and let $F$ be the set of all these points. Then $F$ is spherically-proper. 
\end{example}

\subsection{The proofs of \cref{thm:path-connected components of local unitaries} and \cref{thm:path-connected components of local bulk-non-trivial projections}}\label{subsec:proofs of complex classification theorems}

We first show how \cref{thm:path-connected components of local unitaries} and \cref{thm:path-connected components of local bulk-non-trivial projections} follow largely from \cref{prop:spherically-proper projections equivalent to identity}, \cref{prop:local unitary reduced by geo proper proj} and \cref{prop:compress unitary homotopy from matrix algebra back to original algebra}.
At this stage these propositions should be viewed as the functional-analytic tools
that implement pinning to a spherically-proper region and compression of stabilized homotopies.
Their proofs are postponed to \cref{subsec:technical lemmas}, where we develop the necessary decoupling lemmas.

\begin{prop}\label{prop:spherically-proper projections equivalent to identity}
    Let $P$ be a spherically-proper projection. Then $P\sim P^\perp \sim \Id$ in $\calL_d$.
\end{prop}

\begin{prop}\label{prop:local unitary reduced by geo proper proj}
    Let $U\in\calU(\calL_d)$ be a spherically-local unitary operator. Then there exists a spherically-proper projection $P$ such that 
    \eq{
        U\sim_h W= PWP + P^\perp
    }
    in $\calU(\calL_d)$ for some $W\in\calU(\calL_d)$.
\end{prop}

\begin{prop}\label{prop:compress unitary homotopy from matrix algebra back to original algebra}
    Let $U\in\calU(\calL_d)$. Suppose $U$ takes the form $PUP+P^\perp$ for some spherically-proper projection $P$, and there exists $n\in\NN$ such that $U\oplus \Id_n\sim_h \Id_{n+1}$, then $U\sim_h \Id$ in $\calU(\calL_d)$.
\end{prop}

\begin{proof}[Proof of \Cref{thm:path-connected components of local unitaries}]
    We establish the set bijection
    \eql{\label{eq:pi0 unitary bijective to K1}
    \pi_0(\calU(\calL_d)) \to \rmK_1(\calL_d)
    }
    that maps unitary in path-connected component to its $K$-class.
    First we show it is injective.
    It suffices to show that if $U\in\calU(\calL_d)$ satisfies $[U]_1=0$, then $U\sim_h \Id$ in $\calU(\calL_d)$. Indeed, for $U,V\in\calU(\calL_d)$ with $[U]_1=[V]_1$, we have $[UV^\ast]_1=0$, and the homotopy $UV^\ast\sim_h\Id$ induces the homotopy $U\sim_h V$ in $\calU(\calL_d)$. Suppose $U\in\calU(\calL_d)$ and $[U]_1=0$. Using \cref{prop:local unitary reduced by geo proper proj}, we may assume $U$ takes the form of $P_0UP_0+P_0^\perp$ for some spherically-proper projection $P_0$. 
    Since $[U]_1=0$, it follows that $U\oplus \Id_n\sim_h \Id_{n+1}$ in $\calU_{n+1}(\calL_d)$ for some $n\in\NN$. Then, with the assumptions in \cref{prop:compress unitary homotopy from matrix algebra back to original algebra} satisfied, we have $U\sim_h \Id$. This completes the proof for \cref{eq:pi0 unitary bijective to K1}. 
    
    Let us show that \cref{eq:pi0 unitary bijective to K1} is surjective. Let $[U]_1\in K_1(\calL_d)$ be a $K$-class with $U\in\calU_n(\calL_d)$ for some $n\in\NN$. We'd like to find $W\in\calU(\calL_d)$ such that $[W]_1=[U]_1$. 
    Let $\ZZ^d=F_1\cup \dots\cup F_n$ be a partition of $\ZZ^d$ into spherically-proper subsets. Let $P_j=\Lambda_{F_j}$. Using \cref{prop:spherically-proper projections equivalent to identity}, there exists partial isometry $V_j\in\calL_d$ such that $V_j^\ast V_j = \Id$ and $V_j V_j^\ast = P_j$. Let $V=\begin{bmatrix}V_1 &\dots & V_n\end{bmatrix}\in M_{1,n}(\calL_d)$, and consider $W:=VUV^\ast\in\calU(\calL_d)$. It is straightforward to show that
    \eq{
    \begin{bmatrix}
        \cos t \Id_n & -\sin t V^\ast \\ \sin t V & \cos t \Id
    \end{bmatrix}
    \begin{bmatrix}
        U & 0 \\ 0 & \Id
    \end{bmatrix}
    \begin{bmatrix}
        \cos t \Id_n & \sin t V^\ast \\ -\sin t V & \cos t \Id        
    \end{bmatrix}
    }
    for $t\in[0,\pi/2]$ provides the homotopy $U\oplus \Id\sim_h \Id_n \oplus W$. Therefore $[U]_1=[W]_1$.
    
    Using \cref{eq:pi0 unitary bijective to K1} and \cref{prop:K-groups calculation of local algebra}, we establish that $\pi_0(\calU(\calL_d))\cong \ZZ$ if $d\in 2\NN+1$ and $\pi_0(\calU(\calL_d))\cong \Set{0}$ if $d\in 2\NN$.

    Next, we show that for $d\in 2\NN+1$, the index is given by $\calN:\calU(\calL_d)\to \ZZ$, which amounts to showing that $\calN$ is continuous and bijective. It is clear that $\calN$ is continuous.
    If $U,V\in\calU(\calL_d)$ have the same index $\calN(U)=\calN(V)$, using \cref{prop:index and K-groups}, we have $[U]_1=[V]_1$, and hence $U,V$ belong to the same path-connected component by \cref{eq:pi0 unitary bijective to K1}. 
    For any $m\in\ZZ$, using \cref{prop:index and K-groups}, there exists $[U]_1\in \rmK_1(\calL_d)$ for some $U\in\calU(\calL_d)$ such that $\calN(U)=m$. This completes the proof.
\end{proof}

\begin{rem}
    The injectivity proof in \cref{thm:path-connected components of local unitaries} bears the flavor of contractibility proof in \cite{Kuiper1965} based on the Eilenberg–Mazur swindle argument, which is elementary and purely functional analytic. We explore this possibility in the future.
\end{rem}

\begin{proof}[Proof of \Cref{thm:path-connected components of local bulk-non-trivial projections}]
    We establish the set bijection
    \eql{\label{eq:pi0 projection same as K0}
    \pi_0(\calP^\NT(\calL_d))\to \rmK_0(\calL_d)
    }
    that maps projection in each path-connected component to its $K$-class.
    First we show that it is injective.
    Consider the case when $d$ is even. Suppose $P,Q\in\calP^\NT(\calL_d)$ belong to the same class $[P]_0=[Q]_0$. Then $P\oplus \Id_n \sim Q\oplus \Id_n$ for some $n\in\NN$. Using \cref{lem:projection MvN amplification}, it follows that $P\sim_0 P\oplus \Id_n\sim Q\oplus \Id_n\sim_0 Q$ and hence $P\sim Q$. 
    Since $[\Id]_0=0$ by \cref{lem:projection MvN amplification}, we have $[P^\perp]_0=[\Id]_0-[P]_0=-[P]_0=-[Q]_0=[\Id]_0-[Q]_0= [Q^\perp]_0$. Thus $P^\perp \sim Q^\perp$ based on the same argument (in showing $P\sim Q$). Using \cite[Proposition 2.2.2]{rordam2000introduction}, we have $P\sim_u Q$, i.e., there exists $U\in\calU(\calL_d)$ such that $Q=UPU^\ast$. Since by \cref{thm:path-connected components of local unitaries} we have that $\pi_0(\calU(\calL_d))=\{0\}$ when $d$ is even, it follows that $U\sim_h \Id$ in $\calU(\calL_d)$. Thus $P\sim_h Q$ in $\calP(\calL_d)$, which is in fact in $\calP^\NT(\calL_d)$ by \cref{lem:non bulk projection path lifts to bulk path}. 

    Now consider the case when $d$ is odd. Let $P$ be a spherically-proper projection, e.g., that from \cref{ex:shperically-proper}. Let $Q\in\calP^\NT(\calL_d)$ be arbitrary. We show that $Q\sim_h P$, establishing that $\pi_0(\calP^\NT(\calL_d))=\Set{0}$.
    Since $\rmK_0(\calL_d)=\Set{0}$ from \cref{prop:K-groups calculation of local algebra} and hence $[Q]_0=[P]_0$, analogous to the injectivity proof in the case when $d$ is even, it follows that $Q=UPU^\ast$ for some $U\in\calU(\calL_d)$.
    Since $P$ is spherically-proper, using \cref{prop:spherically-proper projections equivalent to identity}, we have $P\sim \Id$, and hence there exists spherically-local partial isometry $V$ such that $P=VV^\ast$ and $\Id=V^\ast V$. Consider the spherically-local unitary $W=VUV^\ast + P^\perp\in\calU(\calL_d)$. We claim that $U\sim_h W$ in $\calU(\calL_d)$. To that end, we consider 
    \eq{
    R:=\begin{bmatrix}
        V & P^\perp \\ 0 & V^\ast
    \end{bmatrix}\in \calU_2(\calL_d)
    }
    Then we have $R (U\oplus \Id) R^\ast = W\oplus \Id$, and hence $[U]_1 = [W]_1$. Using \cref{eq:pi0 unitary bijective to K1}, we have $U\sim_h W$ in $\calU(\calL_d)$. In particular, this provides the desired homotopy
    \eq{
    Q = UPU^\ast \sim_h WPW^\ast = VUV^\ast P VU^\ast V^\ast = P
    }
    in $\calP(\calL_d)$, which is in fact in $\calP^\NT(\calL_d)$ by \cref{lem:non bulk projection path lifts to bulk path}.

    We now show that \cref{eq:pi0 projection same as K0} is surjective. Using \cref{prop:K-groups calculation of local algebra}, an element in $K_0(\calL_d)$ is of the form $[Q]_0$ for $Q\in\calP(\calL_d)$. However, $Q$ may not be bulk-non-trivial; therefore, we seek $P\in\calP^\NT(\calL_d)$ such that $[P]_0=[Q]_0$. To that end, let $S$ be a spherically-proper projection, which is also bulk-non-trivial, see \cref{rem:bulk-non-trivial and spherically-proper}. Consider $R:=Q\oplus S\in\calP_2(\calL_d)$. In fact $R$ is bulk-non-trivial since $S$ is. (Although \cref{defn:bulk non triviality} is defined only in $\calP(\calL_d)$, it naturally extends to $\calP(M_2(\calL_d))\equiv \calP_2(\calL_d)$.) 
    Moreover, we have $[R]_0=[Q]_0+[S]_0 = [Q]_0$ where we use $[S]_0=[\Id]_0=0$ with the help of \cref{prop:spherically-proper projections equivalent to identity} and \cref{lem:projection MvN amplification}.
    Since $S$ and $S^\perp$ are spherically-proper, using \cref{prop:spherically-proper projections equivalent to identity}, there exist partial isometry $V_1,V_2\in\calL_d$ such that $V_1V_1^\ast = S$, $V_2V_2^\ast = S^\perp$ and $V_1^\ast V_1=V_2^\ast V_2=\Id$. Let $V:=\begin{bmatrix}
        V_1 & V_2
    \end{bmatrix}\in M_{1,2}(\calL_d)$. Then $VV^\ast=\Id$ and $V^\ast V=\Id_2$. Consider $P:=VRV^\ast\in\calP(\calL_d)$. We have $[P]_0=[R]_0=[Q]_0$. Moreover, $P\in\calP^\NT(\calL_d)$ is bulk-non-trivial. Indeed, we use \cref{lem:alternative bulk non triviality} and suppose that $P\rho(a)P$ is compact. Consider $V^\ast P\rho(a)PV = RV^\ast \rho(a) V R$. Now $\rho(a)V- V\rho_2(a)$ is compact since $V$ is spherically-local and $\rho_2(a):=\rho(a)\oplus \rho(a)$. Thus $R\rho_2(a)R$ is compact. Since $R$ is bulk-non-trivial, it follows that $a=0$ and we conclude that $P$ is bulk-non-trivial.

    Next, we show that for $d\in 2\NN$, the index is given by $\calN:\calP^\NT(\calL_d)\to \ZZ$, which amounts to showing that $\calN$ is continuous and bijective. It is clear that $\calN$ is continuous. That $\calN$ is bijective follows from \cref{prop:index and K-groups} and \cref{eq:pi0 projection same as K0}.
    
\end{proof}

\begin{rem}
    The set of bulk-non-trivial spherically-local projections $\calP^\NT(\calL_d)$ is a maximal component within $\calP(\calL_d)$, in the sense that if $P\in \calP^\NT(\calL_d)$, $Q\in \calP(\calL_d)$ and $P,Q$ are in the same path-connected component of $\calP(\calL_d)$, then actually $Q\in \calP^\NT(\calL_d)$. This follows from \cref{lem:non bulk projection path lifts to bulk path}.
\end{rem}

\subsection{Technical lemmas}\label{subsec:technical lemmas}

One of the main techniques is perturbing spherically-local operators that contribute negligible interaction. Geometrically, we seek regions in $\ZZ^d$ where interactions are weakly coupled.
\begin{defn}[$\ve$-weak coupling]
Let $E, F \subset \mathbb{Z}^d$ be two subsets of the lattice and let $A\in\calB(\calH_d)$. We say that $E$ is $\ve$-weakly coupled to $F$ via $A$ if $\|\Lambda_E A\Lambda_F\| \leq \ve$. We say that $E$ is decoupled from $F$ via $A$ if $\Lambda_E A \Lambda_F=0$.
\end{defn}
Equivalently, $E$ is $\ve$-weakly coupled to $F$ via $A$ if and only if $|\braket{\vf,A\psi}|\leq \ve$ for all normalized states $\vf\in \im\Lambda_E$ and $\psi\in\im\Lambda_F$. In other words, the interaction energy, or transition amplitude, between any specific configuration in $F$ and any specific configuration in $E$ is uniformly small.

\begin{lem}[compactness implies small coupling outside finite set]\label{lem:approx compact by finite projection}
Let $K\in \calK(\calH_d)$ be compact. For any $\ve>0$, there exists a finite subset $F\subset \ZZ^d$ such that $\ZZ^d\setminus  F$ is $\ve$-weakly coupled to $\ZZ^d$ via $K$.
\end{lem}
\begin{proof}
Consider an increasing sequence of finite subsets $F_1\subset F_2\cdots$ such that $\cup_{k\in\NN}F_k=\ZZ^d$. Indeed, we can construct $F_k$ to be $B_r \cap \ZZ^d$ for larger and larger $r>0$. Then $\Lambda_{F_k}$ converges to $\Id$ in the strong operator topology. Since $K$ is compact, it follows that $\Lambda_{F_k}K$ converges in norm to $K$. Therefore, for $k$ large enough, we have $\|\Lambda_{F_k}K-K\|\leq \ve$.
\end{proof}

\begin{figure}[t]
    \centering

    \begin{minipage}{\textwidth}
        \centering
        \begin{tikzpicture}[scale=0.7, >=latex]
            \def\R{3.6}          
            \def\Jcenter{60}     
            \def\Jhalfwidth{25}  
            
            \pgfmathsetmacro{\Jstart}{\Jcenter - \Jhalfwidth}
            \pgfmathsetmacro{\Jend}{\Jcenter + \Jhalfwidth}

            \clip (-\R-0.5,-\R-0.5) rectangle (\R+0.5,\R+0.5);

            \fill[gray!30] (0,0) -- (\Jstart:\R) arc (\Jstart:\Jend:\R) -- cycle;
            
            \def\logfactor{0.25} 
            \def\logk{12}        

            \fill[orange!30] (0,0) -- (\Jend:\R) 
                -- plot[domain=\R:0, samples=100, smooth] 
                   ({ \x*cos(\Jend) - \logfactor*ln(1+\logk*\x)*sin(\Jend) }, 
                    { \x*sin(\Jend) + \logfactor*ln(1+\logk*\x)*cos(\Jend) })
                -- cycle;

            \fill[orange!30] (0,0) -- (\Jstart:\R) 
                -- plot[domain=\R:0, samples=100, smooth] 
                   ({ \x*cos(\Jstart) + \logfactor*ln(1+\logk*\x)*sin(\Jstart) }, 
                    { \x*sin(\Jstart) - \logfactor*ln(1+\logk*\x)*cos(\Jstart) })
                -- cycle;

            \begin{scope}[on background layer]
                \fill[cyan!10] (0,0) circle (\R);
            \end{scope}

            \draw[->, thick] (-\R-0.3,0) -- (\R+0.3,0);
            \draw[->, thick] (0,-\R-0.3) -- (0,\R+0.3);

            \draw[thick, gray!90] (0,0) -- (\Jstart:\R);
            \draw[thick, gray!90] (0,0) -- (\Jend:\R);
            
            \draw[thick, orange!80] plot[domain=0:\R, samples=100, smooth] 
                ({ \x*cos(\Jend) - \logfactor*ln(1+\logk*\x)*sin(\Jend) }, 
                 { \x*sin(\Jend) + \logfactor*ln(1+\logk*\x)*cos(\Jend) });
                 
            \draw[thick, orange!80] plot[domain=0:\R, samples=100, smooth] 
                ({ \x*cos(\Jstart) + \logfactor*ln(1+\logk*\x)*sin(\Jstart) }, 
                 { \x*sin(\Jstart) - \logfactor*ln(1+\logk*\x)*cos(\Jstart) });

            \node[font=\bfseries] at (\Jcenter:2.0) {$C_J$};

            \pgfmathsetmacro{\FlabelAngLower}{\Jstart - 10}
            \node[font=\bfseries, orange!90!black] at (\FlabelAngLower:2.2) {$F$};

            \node[font=\bfseries, blue!60!black] at (\Jcenter+180:2.0) {$E$};
        \end{tikzpicture}

        \caption{Illustration for \cref{lem:contained in cone} (cone decoupling). Given a closed subset \(J\subset S^{d-1}\), the complement cone \(C_{J^{c}}\) is partitioned into \(E\sqcup F\) so that the ``bulk'' part \(E\) has uniformly small coupling to \(C_J\) (i.e.\ \(\|\Lambda_E A \Lambda_J\|\) is small), while the remainder \(F\) is directionally thin: for any closed cone \(C_I\) disjoint from \(C_J\), the intersection \(F\cap C_I\) is finite.}
    \end{minipage}

    \vspace{1.0em} 

    \begin{minipage}{\textwidth}
        \centering
        \begin{tikzpicture}[scale=0.7, >=latex]
            \def\rzero{1.2}
            \def\rone{2.7}  
            \def\rtwo{3.6}  

            \fill[gray!20] (0,0) circle (\rzero);
            \draw[ultra thick, gray!80] (0,0) circle (\rzero);
            
            \fill[blue!20, even odd rule] (0,0) circle (\rone) (0,0) circle (\rzero);
            \draw[ultra thick, blue!50] (0,0) circle (\rone);

            \fill[gray!20, even odd rule] (0,0) circle (\rtwo) (0,0) circle (\rone);

            \draw[->, thick] (-4.2,0) -- (4.2,0);
            \draw[->, thick] (0,-4.2) -- (0,4.2);

            \fill[red] (0.4, 0.3) circle (2pt) node[right, black] {$x_1$};
            \fill[red] (-1.5, 0.8) circle (2pt) node[left, black] {$x_2$};
            \fill[red] (2.2, -1.9) circle (2pt) node[right, black] {$x_3$};

            \node[font=\bfseries,blue!90!black] at (55:1.85) {$B_{r_2} \!\setminus\! B_{r_1}$};
        \end{tikzpicture}

        \caption{Schematic for \cref{lem:contained in annulus}. Points \(x_k\) are chosen in disjoint annuli \(B_{r_k}\setminus B_{r_{k-1}}\) so that the operator \(A\) couples \(\delta_{x_k}\) essentially only inside its own annulus: the complement \(\mathbb{Z}^d\setminus(B_{r_k}\setminus B_{r_{k-1}})\) is \(\varepsilon_k\)-weakly coupled to \(\{x_k\}\).}
    \end{minipage}
\end{figure}

\begin{lem}[cone-decoupling partition]\label{lem:contained in cone}
Let $A$ be a spherically-local operator and $\ve>0$ be arbitrary. Let $J$ be any proper, non-empty closed subset of $\bS^{d-1}$. Then there exist two disjoint subsets $E$ and $F$ in $\ZZ^d$ partitioning $C_{J^c}$ 
such that $E$ is $\ve$-weakly coupled to $C_J$ via $A$, the operator $\Lambda_E A\Lambda_J$ is compact, and for any closed set $I\subset \bS^{d-1}$ disjoint from $J$, we have $|F\cap C_I|<\infty$.
\end{lem}
\begin{proof}
    Let $J$ be a closed subset of $\bS^{d-1}$. Define $N_k$ by
    \eq{
    N_k = \Set{x\in \bS^{d-1} | \|x-y\|<1/k \text{ for some $y$ in $J$}} \supset J.
    }
    It is clear that $N_k$ is open in $\bS^{d-1}$, decreasing ($N_{k}\supset N_{k+1}$ holds for each $k$), and $J = \cap_k N_k$. In particular, $N_k^c$ is closed and $J^c = \cup_k N_k^c$. Since $N_k^c$ and $J$ are disjoint and closed, by \cref{thm:geometric characterization of locality}, it follows that $\Lambda_{N^c_k}A\Lambda_J \in \calK$. 
    Using \cref{lem:approx compact by finite projection}, we can decompose $C_{N^c_k}$ into two disjoint subsets
    \eq{
    	C_{N^c_k} = E_k \cup F_k
    }
    such that $|F_k|<\infty$ and $\| \Lambda_{E_k} A \Lambda_J\|\leq \ve/2^k$. 
    Let
    \eq{
    	E = \cup_{k\in\NN} E_k \,.
    }
    We can rewrite $E$ as $\cup_{k\in\NN} G_k$ where $G_k = E_k\setminus (\cup_{i=1}^{k-1}E_i)$ are disjoint. Then $\Lambda_E A\Lambda_J = (\sum_k \Lambda_{G_k}) A \Lambda_J$. In fact the series $\sum_k \Lambda_{G_k}A\Lambda_J$ converge in operator norm. Indeed, we have
    \eq{
    	 \sum_{k=1}^\infty \| \Lambda_{G_k} A \Lambda_J\| \leq \sum_{k=1}^\infty \| \Lambda_{E_k} A \Lambda_J\| \leq \ve\,.
    }
    Since each $\Lambda_{G_k}A\Lambda_J$ is compact, it follows that $\Lambda_EA\Lambda_J$ is compact. Moreover, $\|\Lambda_EA\Lambda_J\|\leq \ve$.
    Define $F$ by
    \eq{
    	F = C_{J^c}\setminus E\,.
    }
    We argue that $|F\cap C_{N_k^c}|<\infty$ for all $k\in\NN$. Indeed, we have
    \eq{
        F\cap C_{N_k^c}=F\cap(E_k\cup F_k) = F\cap F_k\subset F_k
    }
    and we have $|F\cap C_{N_k^c}|\leq |F_k|<\infty$.
    Let $I$ be any closed subset of $\bS^{d-1}$ disjoint from $J$. Let $\delta = \dist(I,J)>0$. If we pick $k\in\NN$ large enough so that $1/k < \delta$, then we have $I\subset N_{k}^c$. Thus, $|F\cap C_I|\leq |F\cap C_{N_k^c}|<\infty$.
\end{proof}

\begin{lem}[points in annuli with weak coupling]\label{lem:contained in annulus}
Let $A\in\calB(\calH_d)$ and $\{\ve_i\}_{i\in\NN}$ be sequence of positive numbers. There exists a spherically-proper sequence $\{x_i\}_{i\in\NN}$ of points in $\ZZ^d$, and a sequence $\{r_i\}_{i\in\NN}$ of strictly increasing radii $0=:r_0<r_1<\dots$, such that $x_i$ lies in the annulus $B_{r_i}\setminus B_{r_{i-1}}$, and $\ZZ^d\setminus (B_{r_i}\setminus B_{r_{i-1}})$ is $\ve_i$-weakly coupled to $\{x_i\}$ via $A$ for each $i\in\NN$.
\end{lem}
\begin{proof}
Let $\{I_i\}_{i\in\NN}$ be a countable basis for $\bS^{d-1}$. Note in $d=1$ we simply alternate between two points of $\bS^0$.
We do the construction iteratively.
Let us pick an element $x_1\in C_{I_1}\subset \ZZ^d$. There exists $r_1>\|x_1\|$ such that 
\eq{
\|\Lambda_{B_{r_1}^c} A \Lambda_{\{x_1\}}\|\leq \ve_1 \,.
}
Indeed, this follows from \cref{lem:approx compact by finite projection} with the fact that $A \Lambda_{\{x_1\}}$ is compact and $\Lambda_{B_{r}}$ converges to $\Id$ strongly as $r\to\infty$. 
Since $|B_{r_1}|$ is finite, which implies $A^\ast \Lambda_{B_{r_1}}$ is compact, it follows from \cref{lem:approx compact by finite projection} again that there exists $t_1>r_1$ such that $\|\Lambda_{B_{t_1}^c}A^\ast\Lambda_{B_{r_1}}\|\leq \ve_2/2$, or
\eq{
    \| \Lambda_{B_{r_1}} A \Lambda_{B_{t_1}^c}\|\leq \ve_2/2 \,.
}
Pick $x_2\in C_{I_2}\subset \ZZ^d$ with $\|x_2\|> t_1$. There exists $r_2>\|x_2\|$ such that
\eq{
    \|\Lambda_{B_{r_2}^c} A  \Lambda_{\{x_2\}}\| \leq \ve_2/2
}
which follows again from \cref{lem:approx compact by finite projection}. Therefore, we have
\eq{
    \|\Lambda_{B_{r_2}^c\cup B_{r_1}} A \Lambda_{\{x_2\}}\| \leq \|\Lambda_{B_{r_2}^c} A \Lambda_{\{x_2\}}\| + \|\Lambda_{B_{r_1}} A \Lambda_{B_{t_1}^c}\| \leq \ve_2
}
where we used $\Lambda_{\{x_2\}}\leq \Lambda_{B_{t_1}^c}$ in the first inequality. 

We iterate the procedure: we pick $t_2>r_2$ such that $\|\Lambda_{B_{t_2}^c} A^\ast \Lambda_{B_{r_2}}\|\leq \ve_3/2$; pick $x_3\in C_{I_3}\subset \ZZ^d$ with $\|x_3\|>t_2$; and pick $r_3>\|x_3\|$ such that $\|\Lambda_{B_{r_3}^c} A (\delta_{x_3}\otimes \delta_{x_3}^\ast)\| \leq \ve_3/2$; and we get $\|\Lambda_{B_{r_3}^c\cup B_{r_2}} A \Lambda_{\{x_3\}}\|\leq \ve_3$; and so on.

\end{proof}

\begin{figure}[t]
    \centering
        \begin{tikzpicture}[scale=0.7, >=latex]

            \pgfmathsetseed{42}

            \def\maxShells{22} 
            
            \newcommand{\getRadius}[1]{0.164*#1}

            \pgfmathsetmacro{\maxRadius}{\getRadius{\maxShells}}
            \pgfmathsetmacro{\bgRadius}{\maxRadius + 0.2}

            \newcommand{\drawSector}[4]{
                \pgfmathsetmacro{\innerR}{#3 + 0.01} 
                \pgfmathsetmacro{\outerR}{#4 - 0.01}
                
                \fill[orange!30] (#1:\innerR) arc (#1:#2:\innerR) -- (#2:\outerR) arc (#2:#1:\outerR) -- cycle;
                \draw[thick, orange!80] (#1:\innerR) arc (#1:#2:\innerR) -- (#2:\outerR) arc (#2:#1:\outerR) -- cycle;
            }

            
            \fill[green!20, opacity=0.3] (0,0) -- (10:\bgRadius) arc (10:45:\bgRadius) -- cycle;
            \draw[green!60!black, thick] (0,0) -- (10:\bgRadius);
            \draw[thick,green!60!black] (0,0) -- (45:\bgRadius);
            \node[font=\bfseries, green!40!black] at (20:{\bgRadius-0.4}) {$C_I$};

            \fill[red!20, opacity=0.3] (0,0) -- (52:\bgRadius) arc (52:80:\bgRadius) -- cycle;
            \draw[red!80!black, thick] (0,0) -- (52:\bgRadius);
            \draw[red!80!black, thick] (0,0) -- (80:\bgRadius);
            \node[font=\bfseries, red!40!black] at (67.5:{\bgRadius-0.4}) {$C_J$};

            \foreach \k in {3,...,\maxShells} {
                \pgfmathsetmacro{\prevK}{\k-1}
                \pgfmathsetmacro{\innerR}{\getRadius{\prevK}}
                \pgfmathsetmacro{\outerR}{\getRadius{\k}}
                
                \pgfmathsetmacro{\startAngle}{mod(\k * 83, 360)}
                \pgfmathsetmacro{\arcLen}{20 + 60*rnd}
                \pgfmathsetmacro{\endAngle}{\startAngle + \arcLen}
                
                \drawSector{\startAngle}{\endAngle}{\innerR}{\outerR}

                \ifnum\k=10
                    \node[orange!80!black] at ({(\startAngle+\endAngle)/2}:{(\innerR+\outerR)/2+0.6}) {$R_k$};
                \fi
            }

            \draw[->, thick] (-\bgRadius-0.2,0) -- (\bgRadius+0.2,0);
            \draw[->, thick] (0,-\bgRadius-0.2) -- (0,\bgRadius+0.2);

        \end{tikzpicture}
        \caption{Illustration for \cref{lem:localized centers}. The finite sets \(R_k\) (support islands for \(B\delta_{x_k}\)) lie in disjoint shells and become asymptotically localized in a single direction on \(S^{d-1}\). Consequently, for disjoint dyadic cones \(C_I\) and \(C_J\), only finitely many islands can intersect both cones.}
\end{figure}
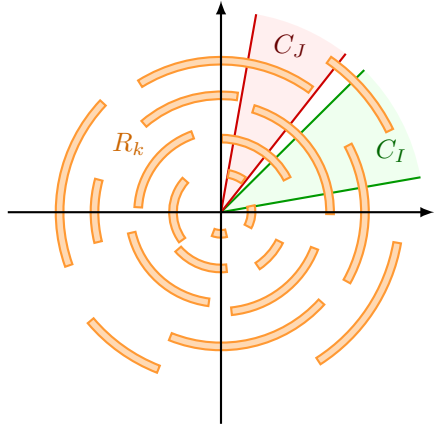

\begin{lem}\label{lem:localized centers}
Let $A$ be a spherically-local operator and $\ve>0$. 
Then there exists a spherically-proper sequence $\{x_k\}_{k\in\NN}$ of points on $\ZZ^d$, and a spherically-local operator $B$ with $\|A-B\|\leq \ve$ and $A-B\in\calK(\calH_d)$ such that if we define
\eql{\label{eq:interaction range of B}
    R_k := \supp(B \delta_{x_k}) \cup \supp(\delta_{x_k}) \subset \ZZ^d
 }
then $|R_k|<\infty$ for all $k$; the sets $R_k$ are pairwise disjoint; and satisfies the cone-separation property: for any closed disjoint dyadic cubes $I$ and $J$, there are finitely many $k$ such that $R_k\cap C_I\neq \varnothing$ and $R_k\cap C_J\neq \varnothing$.
\end{lem}
\begin{proof}
    Let $A$ be a spherically-local operator.
    Consider the collection $\{G_k\}$ of all closed dyadic cubes in $\bS^{d-1}$. For each cube $G_k$, we can use \cref{lem:contained in cone} to partition $C_{G_k^c}$ into two disjoint subsets $E_k$ and $F_k$ such that $\Lambda_{E_k}A\Lambda_{G_k}$ has small norm and is compact, and $|F_k\cap C_I|<\infty$ for all closed set $I\subset \bS^{d-1}$ disjoint from $F_k$. 
    On the other hand, using \cref{lem:contained in annulus}, there exists a spherically-proper sequence of points $\{x_k\}$ and a sequence of increasing balls $\{B_{r_k}\}$ (with $0=:r_0<r_1<r_2<\dots$) such that $x_k\in B_{r_k}\setminus B_{r_{k-1}}$ and $\Lambda_{\ZZ^d\setminus(B_{r_k}\setminus B_{r_{k-1}})} A \Lambda_{\{x_k\}}$ has small norm and is compact. 
    Now, apply \cref{lem:turn off elements} to get a spherically-local operator $B$ such that $\|A-B\|\leq \ve$ and $A-B\in\calK(\calH_d)$ and
    \eql{
    \Lambda_{\ZZ^d\setminus(B_{r_k}\setminus B_{r_{k-1}})} B \Lambda_{\{x_k\}} &= 0 \label{eq:B annulus zero}\\ 
    \Lambda_{E_k} B\Lambda_{G_k} & =  0 \label{eq:B cone finite}
    }
    for all $k\in\NN$. In other words, $\ZZ^d\setminus(B_{r_k}\setminus B_{r_{k-1}})$ is decoupled from $\{x_k\}$ via $B$; and the complement of $C_{G_k}$ is almost decoupled from $C_{G_k}$ except for a set $F_k$.
    Since each $x_k$ lies in disjoint annulus $B_{r_k}\setminus B_{r_{k-1}}$ and the support of $B\delta_{x_k}$ is also contained in the same annulus via \cref{eq:B annulus zero}, then
    \eq{
    R_k\subset B_{r_k}\setminus B_{r_{k-1}}
    }
    and it follows that $R_k$ as defined in \cref{eq:interaction range of B} have the properties that $|R_k|<\infty$ for all $k$, and the sets $R_k$ are pairwise disjoint.

    Let $I$ and $J$ be arbitrary disjoint closed dyadic cubes. By way of contradiction, suppose there are infinitely many $k$ such that $R_{k}\cap C_I\neq \varnothing$ and $R_{k}\cap C_J\neq \varnothing$ for all $i\in\NN$; we denote the subsequence as $\{k_1\}$. 
    Let $\mathcal{D}_l$ be the collection of all closed dyadic cubes of the $l$-th generation as defined in \cref{rem:dyadic cubes} where $l$ is large enough such that no cubes in $\mathcal{D}_l$ intersects both $I$ and $J$. Indeed, this is possible using the fact that $I$ and $J$ are disjoint closed dyadic cubes, and the construction of closed dyadic cubes in \cref{rem:dyadic cubes}.
    Since $\mathcal{D}_l$ is a finite covering of $\bS^{d-1}$ but $\{x_{k_1}\}$ contains infinitely many points of $\ZZ^d$, there exists a cube $Q\in \mathcal{D}_l$ such that $\hat{x}_{k_1}\in Q$ for infinitely many $k_1$; we denote the sub-subsequence as $\{k_2\}$. Without loss of generality, suppose $Q$ is disjoint from $I$.
    
    Since $R_{k_2}\cap C_I\neq \varnothing$ by assumption, we pick a point $y_{k_2}$ in $R_k\cap C_I$ for each $k_2$. Since $\{x_{k_2}\}\subset C_Q$ and $C_Q$ is disjoint from $C_I\supset R_k\cap C_I$, we cannot have $y_{k_2}$ being the same as $x_{k_2}$. 
    Therefore 
    \eq{
    y_{k_2}\in \supp(B\delta_{x_{k_2}})\cap C_I \subset C_{Q^c}
    }
    for each $k_2$. Denote the particular partition, as performed in the beginning, of $C_{Q^c}$ as $E$ and $F$, where $E$ is decoupled from $C_Q$ via $B$ (with \cref{eq:B cone finite}), and for all closed subsets $K\subset \bS^{d_1}$ disjoint from $Q$, we have $|F\cap C_K|<\infty$ (as promised by \cref{lem:contained in cone}).  
    Thus $y_{k_2}\subset E\cup F$. However, $y_{k_2}$ cannot be in $E$; otherwise, we have
    \eq{
        \braket{\delta_{y_{k_2}},B\delta_{x_{k_2}}} = \braket{\Lambda_{E}\delta_{y_{k_2}},B\Lambda_{Q}\delta_{x_{k_2}}} = \braket{\delta_{y_{k_2}},\Lambda_{E}B\Lambda_{Q}\delta_{x_{k_2}}}=0
    }
    where the last equality follows from \cref{eq:B cone finite}, and we would have $y_{k_2}\notin \supp(B\delta_{x_{k_2}})$, a contradiction. Thus $y_{k_2}\in F\cap C_I$ for all $k_2$. This contradicts the fact that $|F\cap C_I|<\infty$.
\end{proof}




\begin{proof}[Proof of \cref{prop:spherically-proper projections equivalent to identity}]
    We show that $P\sim \Id$. The proof for $P^\perp\sim \Id$ is similar since $P$ is spherically-proper iff $P^\perp$ is. Let $\Set{y_k}_{k=1}^\infty=\ZZ^d$ be an enumeration of the lattice. Let $F$ be the spherically-proper set that correspond to $P=\Lambda_F$. Define subsets $N_k$ by
    \eq{
    N_k=\Set{x\in\bS^{d-1} | \|x-\hat{y}_k\|<1/k}\,.
    }
    Iteratively, for each $k\in \NN$, pick $x_k\in F$ such that $x_k$ minimizes
    \eq{
        \Set{ \norm{x} | x\in C_{N_k}\cap F\setminus \{x_1,\dots,x_{k-1}\}}\,.
    }
    This is possible since $|C_{N_k}\cap F|=\infty$ for all $k$ by definition of spherical properness.
    Consider the mapping $\ZZ^d \ni y_k\mapsto x_k \in F$. It is clear that the map is bijective.
    The map is injective since at each step we exclude previously chosen points from $F$. Assume by way of contradiction that the map is not surjective. Choose any $z^\ast \in F\setminus \Set{x_k}_{k=1}^\infty$ that is not picked by the algorithm. Let $K\subset\NN$ be all indices where $\hat{y}_k=\hat{z}^\ast$ for $k\in K$. The index set $K$ is infinite as there are infinitely points in $\ZZ^d$ in the same direction of $z^\ast$. Moreover we have $z^\ast \in C_{N_k} \cap F\setminus \Set{x_1,\dots,x_{k-1}}$ for all $k\in K$, i.e., $z^\ast$ is a valid candidate for every step $k\in K$ but never get picked. The algorithm picks $x_k$ that minimizes the norm among candidates. Since $z^\ast$ is a candidate, the chosen $x_k$ must satisfy $\|x_k\|\leq \|z^\ast\|$. Thus for every step $k\in K$, the algorithm selects distinct $x_k\in F$ with norm $\|x_k\|\leq \|z^\ast\|$. Since $K$ is an infinite set, it implies that there are infinitely many distinct elements in $F\subset \ZZ^d$ with norm less than or equal to $\|z^\ast\|$, which is impossible.
    
    Let $V$ maps $\delta_{y_k}$ to $\delta_{x_k}$ and extend linearly to an operator from $\calH_d$ to $\im P$. Then $V^\ast V=\Id$ and $VV^\ast = P$.
    We now show that $V$ is spherically-local. Let $I$ and $J$ be a pair of disjoint closed subsets of $\bS^{d-1}$, and we consider the operator $\Lambda_J V\Lambda_I$. The distance between the set $\dist(I,J)=\delta > 0$ is strictly positive. Let $k_0$ be large enough so that $1/k_0 < \delta$. Then, for all $k\geq k_0$ and $y_k\in C_I$, we have $C_J\cap C_{N_k}=\varnothing$. Thus $\Lambda_J V\Lambda_I = \Lambda_J V \Lambda_{\{x_k\}_{k<k_0}}$ is finite-rank.
\end{proof}

\begin{rem}
    The partial isometries $V$ constructed in \cref{prop:spherically-proper projections equivalent to identity} are real, i.e. $\calC V\calC=V$ where $\calC$ is the complex conjugation. Indeed, they merely reshuffles the position bases. This fact will be used in the real symmetry cases. 
\end{rem}

\begin{proof}[Proof of \cref{prop:local unitary reduced by geo proper proj}]
    Let $U\in\calU(\calL_d)$ be a spherically-local unitary operator. Let $\ve>0$ be some small number to be determined. Using \cref{lem:localized centers}, there exists a spherically-local operator $G\in\calL_d$ such that $\|U-G\|\leq \ve$, and a spherically-proper sequence of points $\{x_k\}_{k\in \NN}$ such that the subsets $R_k\subset\ZZ^d$ as defined in \cref{eq:interaction range of B} has those properties specified in the lemma. In particular, we can choose $\ve$ small enough so that $G\in\calG(\calL_d)$ is invertible. Denote $P$ the spherically-proper projection $\Lambda_{\{x_k\}_{k\in\NN}}$. 
    
    Let us define a spherically-local invertible operator $V$ such that $VG$ acts as identity on $\im P$. To that end, we define $V$ on $\im \Lambda_{R_k}$ as any invertible operator that maps $G\delta_{x_k}\mapsto \delta_{x_k}$. Define $V$ to be identity on $(\oplus_{k}\im \Lambda_{R_k})^\perp$. The operator $V$ is well-defined since $\{R_k\}_{k\in\NN}$ consists of mutually disjoint subsets. We argue that $V$ is spherically-local. 
    To that end, we show that for any pair of disjoint closed dyadic cubes $I$ and $J$ in $\bS^{d-1}$, we have that $\Lambda_J V \Lambda_I$ is finite-rank. We can decompose $\ZZ^d$ into four disjoint subsets
    \eq{\textstyle
        C_{I^c},\ C_I\cap X,\ C_I\cap Y,\ C_I\cap Z
    }
    where $X$ is the $\ZZ^d \setminus \cup_{k} R_k$; and $Y$ is the union of all $R_k$ such that $R_k\cap C_I\neq \varnothing$ and $R_k\cap C_J\neq \varnothing$; and $Z$ is the union of all $R_k$ such that $R_k\cap C_I=\varnothing$ or  $R_k\cap C_J=\varnothing$. With the decompositions, we write
    \eq{
        \Lambda_J V\Lambda_I = \Lambda_J V \Lambda_I (\Lambda_{I^c} + \Lambda_{C_I\cap X} + \Lambda_{C_I\cap Y} + \Lambda_{C_I\cap Z})
    }
    and study each term. It is clear that $\Lambda_JV\Lambda_I\Lambda_{I^c}=0$. Since $V$ is identity on $(\oplus_{k\in\NN} \im \Lambda_{R_k})^\perp $, it follows that $\Lambda_J V\Lambda_I\Lambda_{C_I\cap X} = \Lambda_J\Lambda_{C_I\cap X} = 0$. 
    Suppose $y\in C_I\cap Z$, then $y\in C_I$; and either $y\in C_J\cap R_k$ for some $k\in\NN$ such that $R_k\cap C_I=\varnothing$; or $y\in C_I\cap R_k$ for some $k\in\NN$ such that $R_k\cap C_J=\varnothing$. The former case is vacuous. In the latter case, we have $\Lambda_JV\Lambda_I\delta_y=\Lambda_JV\delta_y$; and since $V\delta_y\in \im \Lambda_{R_k}$ and $R_k\cap C_J=\varnothing$, it follows that $\Lambda_JV\delta_y=0$. Thus $\Lambda_J V\Lambda_{I}\Lambda_{C_I\cap U}=0$. 
    Finally, using \cref{lem:localized centers}, the set $C_I\cap Y$ is finite and hence $\Lambda_JV\Lambda_I=\Lambda_J V\Lambda_I \Lambda_{C_I\cap Y}$ is finite-rank. 

    The argument in the previous paragraph works for all operators $A\in\calB(\calH_d)$ such that $V\im\Lambda_{R_k}\subset \im\Lambda_{R_k}$ for all $k\in \NN$, and are identity on $(\oplus_{k\in\NN} \im \Lambda_{R_k})^\perp $. 
    Therefore, we can construct $V\sim_h \Id$ by deforming each invertible matrices of $V$ on each $\im \Lambda_{Y_k}$ to the identity matrices. So far, we have made the deformations $U\sim_h G\sim_h VG$ in $\calG(\calL_d)$.
    
    By construction $VG$ take the form
    \eq{
        VG =P + PVGP^\perp + P^\perp VG P^\perp = (P + P^\perp VG P^\perp)(\Id + PVGP^\perp)\,.
    }
    Since $VG\in\calL_d$ and $[P,L]=0$, it is clear that $P + P^\perp VG P^\perp$ and $\Id + PVGP^\perp$ are both spherically-local. Observe that $(\Id + t PVGP^\perp)$ is always in $\calG(\calL_d)$ for $t\in[0,1]$, with inverse $\Id - t PVGP^\perp$. Therefore, we have $VG\sim_h P + P^\perp VG P^\perp$ in $\calG(\calL_d)$. Let $S=P+P^\perp SP^\perp\in\calU(\calL_d)$ be the polar part of $P + P^\perp VG P^\perp$. It follows from \cite[Proposition 2.1.8]{rordam2000introduction} that we have $P + P^\perp VG P^\perp\sim_h S$ in $\calG(\calL_L)$.

    Since $P$ is spherically-proper, using \cref{prop:spherically-proper projections equivalent to identity}, we have $P\sim P^\perp$. Using \cref{lem:whitehead between P and P perp}, we have $S=P+P^\perp SP^\perp \sim_h W = PWP+P^\perp$ in $\calU(\calL_d)$ for some $W\in\calU(\calL_d)$.
    
    We have constructed $U\sim_h W=PWP+P^\perp$ in $\calG(\calL_d)$. Using \cite[Proposition 2.1.8]{rordam2000introduction}, we obtain the homotopy $U\sim_h W$ in $\calU(\calL_L)$ and concludes the proof.
\end{proof}

\begin{proof}[Proof of \cref{prop:compress unitary homotopy from matrix algebra back to original algebra}]
    Let $W_t\in\calU_{n+1}(\calL_d)$ be the homotopy implementing $U\oplus \Id_n\sim_h \Id_{n+1}$.
    
    Let $P_0:=P$. Since $P_0$ is spherically-proper, so is $P_0^\perp$. Using \cref{lem:split spherically-proper projections}, we may decompose $P_0^\perp$ into $n+1$ spherically-proper projections
    \eq{
        P_0^\perp = Q_0+ P_1 + \dots + P_{n-1} + P_n.
    }

    Using \cref{prop:spherically-proper projections equivalent to identity}, we have $Q_0\sim P_0^\perp$ and $P_k\sim \Id$ for all $k=1,\dots,n$. Let $V_k\in\calL_d$ be the spherically-local partial isometry such that $P_k=V_kV_k^\ast$ and $\Id = V_k^\ast V_k$ for $k=1,\dots,n$. Since $Q_0\sim P_0^\perp$, there exists a spherically-local partial isometry $V_0\in\calL_d$ such that $P_0^\perp = V_0^\ast V_0$ and $Q_0=V_0V_0^\ast$. Then 
    \eq{
    \begin{bmatrix}
        P_0+V_0 & V_1 & \dots & V_n
    \end{bmatrix}
    W_t
    \begin{bmatrix}
        (P_0+V_0)^\ast \\
        V_1^\ast \\
        \vdots \\
        V_n^\ast
    \end{bmatrix}
    }
    gives a homotopy in $\calU(\calL_d)$ that deforms $U$ to $\Id$. 
\end{proof}

\begin{lem}\label{lem:whitehead between P and P perp}
    Let $P$ and $Q$ be spherically-local projections such that $P\perp Q$ and $P\sim Q$. Suppose $U\in\calU(\calL_d)$ takes the form $PUP+P^\perp$, then there exists $W\in\calU(\calL_d)$ of the form $W=QWQ+Q^\perp$ such that $U\sim_h W$ in $\calU(\calL_d)$.
\end{lem}
\begin{proof}
    Let $V:\im P\to\im Q$ be the partial isomtery such that $P=V^\ast V$ and $Q=VV^\ast$. Let $R=\Id-P-Q$. In the decomposition of $\calH_d = \im P \oplus \im Q \oplus \im R$, consider the rotation
    \eql{\label{eq:rotation between P perp Q}
    R_t:= \begin{bmatrix}
        (\cos t)P & -(\sin t)V^\ast & 0 \\
        (\sin t) V & (\cos t) Q & 0 \\
        0 & 0 & R
    \end{bmatrix}.
    }
    Then the homotopy
    \eq{
    R_t
    \begin{bmatrix}
        PUP & 0 & 0 \\
        0 & Q & 0 \\
        0 & 0 & R
    \end{bmatrix}
    R_t^\ast
    }
    deforms $PUP+P^\perp$ to $QVUV^\ast Q + Q^\perp=:W$ in $\calU(\calL_d)$ for $t\in [0,\pi/2]$. 

\end{proof}

\begin{lem}\label{lem:non bulk projection path lifts to bulk path}
    Let $P$ and $Q$ be spherically-local projections. If $P$ is bulk-non-trivial and $P\sim_h Q$ in $\calP(\calL_d)$, then $Q$ is also bulk-non-trivial, and hence the homotopy itself lies in $\calP^\NT(\calL_d)$.
\end{lem}
\begin{proof}
    Here, we will mainly use the alternative characterization of bulk non-triviality described by \cref{lem:alternative bulk non triviality}.

    Let $P,Q\in\calP(\calL_d)$ and suppose $P$ is bulk-non-trivial and $P\sim_h Q$. It follows that $P\sim_u Q$, i.e., there exists $U\in\calU(\calL_d)$ such that $P=U^\ast QU$. Let $a\in C(\bS^{d-1})$ be such that $Q\rho(a)Q\in\calK(\calH_d)$. Then
    \eq{
    P\rho(a)P = U^\ast QU\rho(a)U^\ast QU = U^\ast Q[U,\rho(a)]U^\ast QU + U^\ast Q\rho(A)Q U\in\calK(\calH_d)
    }
    where we use $[U,\rho(a)]\in\calK(\calH_d)$. Since $P\in\calP^\NT(\calL_d)$, it follows that $a=0$. Now $P\sim_h Q$ implies $P^\perp\sim_h Q^\perp$. Follow the same argument as before, we have that $Q^\perp\rho(a)Q^\perp$ is not compact unless $a=0$. Thus $Q\in\calP^\NT(\calL_d)$.
\end{proof}

\begin{lem}[splitting a spherically-proper set]\label{lem:split spherically-proper projections}
Let $F\subset\ZZ^d$ be a spherically-proper sets. Then there exists spherically-proper disjoint subsets $F_1$ and $F_2$ of $F$ that partition $F=F_1\cup F_2$.
\end{lem}
\begin{proof}
    Let $F\subset\ZZ^d$ be spherically-proper. Let $\Set{I_n}$ be a countable basis for $\bS^{d-1}$. It suffices to show spherical properness with respect to the collection of (overlapping) cones $\Set{C_{I_n}}$. Let $A_n:= F\cap C_{I_n}$. The goal is to distribute points in $A_n$ to two disjoint sets $F_1$ and $F_2$ where each $A_n$ contributes infinitely many points to both sides.
    To that end, consider $\Set{(n_i,k_i)}_{i=1}^\infty$ an enumeration of $\NN\times \NN$. For $i=1$, pick $x_1\in A_{n_1}$ for $F_1$, and pick $y_1\in A_{n_1}\setminus\Set{x_1}$ for $F_2$. For $i=2$, pick $x_2\in A_{n_2}\setminus\Set{x_1,y_1}$ for $F_1$ and $y_2\in A_{n_2}\setminus \Set{x_1,y_1,x_2}$ for $F_2$, and so on. Since $A_n$ contains infinitely many points by spherical properness of $F$, at each step $i$, the set $A_{n_i}$ is nonempty after excluding finitely many points, and hence the algorithm is well-defined. Since we exclude previously chosen points, the sets $F_1$ and $F_2$ are disjoint. Moreover, for each $n\in\NN$, points in $A_n$ are chosen infinitely many times due to the enumeration of $\NN\times \NN$, i.e., the set $\Set{(n_i,k_i) | i\in \NN,\ n_i=n}$ is an infinite set. Therefore, $|F_1\cap A_n|$ and $|F_2\cap A_n|$ are both infinite for all $n\in\NN$. Finally, we toss all the remaining points in $F$ that are not selected in any steps to $F_1$ to make $F_1$ and $F_2$ a partition of $F$.

    Since $|F^c\cap C_I|=\infty$ for all non-empty open subsets $I\subset\bS^{d-1}$ and $F_1,F_2$ are both subsets of $F$, it follows that $|F_i^c\cap C_I|=\infty$ as well for $i\in\Set{1,2}$.
\end{proof}

\begin{lem}[countable decoupling via inclusion–exclusion]\label{lem:turn off elements}
Let $\calH$ be a separable Hilbert space. Let $A\in\calB(\calH)$ and $\ve>0$. Let $\{P_k\}_{k\in\NN},\{Q_k\}_{k\in\NN}$ be projections in $\calB(\calH)$ such that: the projections in $\{P_k\}_{k\in\NN}$ (resp. in $\{Q_k\}_{k\in\NN}$) are pairwise commutative; the operator $P_kAQ_k$ is compact for each $k\in\NN$; and 
\eql{\label{eq:norm bound of Pk A Qk}
    \norm{P_k A Q_k}\leq \frac{\ve}{2^{2k-1}},~ \forall k\in\NN .
}
Then there exists $B\in\calB(\calH)$ such that $\|A-B\|\leq \ve$ and $A-B$ is compact, and $P_kBQ_k=0$ for all $k\in\NN$.
\end{lem}
\begin{proof}
We would have liked to define $B$ as \eq{
    A - \sum_{i} P_i A Q_i\,.
} However, this formula may fail to represent the operator we want since the range of the projections $P_i,P_j$ or $Q_i,Q_j$ may overlap, which would mean we over-delete elements. To remedy this problem, inspired by the inclusion–exclusion formula, we define
\eql{
	S_n&:=\sum_{i=1}^n P_i AQ_i - \sum_{1\leq i<j\leq n}P_iP_j A Q_iQ_j \notag \\
	&\quad + \sum_{1\leq i<j<k\leq n}P_iP_jP_k A Q_iQ_jQ_k - \dots + (-1)^{n-1}P_1\dots P_n A Q_1\dots Q_n \notag \\
	&= \sum_{i=1} (-1)^{i+1}\left(\sum_{1\leq k_1<\dots<k_i\leq n} P_{k_1}\dots P_{k_i} A Q_{k_1}\dots Q_{k_i}\right) \label{eq:inclusion exclusion}
} which corrects all the over-counting. More precisely, let $\ve_k:=\|P_k AQ_k\|$. Then
\eql{\label{eq:norm bound on S_n}
	\|S_n\|\leq \sum_{k=1}^n 2^{k-1}\ve_k\,.
}
For example, we have $\|S_1\|=\|P_1 AQ_1\|=\ve_1$, and 
\eq{
\|S_2\|=\|P_1AQ_1+P_2 AQ_2-P_1P_2AQ_1Q_2\| \leq \ve_1+2\ve_2
}
where we used $\|P_1P_2AQ_1Q_2\|\leq \|P_2AQ_2\|=\ve_2$. Fix $l$, we count the number of terms $P_{k_1}\dots P_{k_l}AQ_{k_1}\dots Q_{k_l}$ in $S_n$ with $k_1<\dots <k_l$ having $k_l=m$. Then use the fact that
\eq{
	\|P_{k_1}\dots P_{k_l}AQ_{k_1}\dots Q_{k_l}\|\leq \|P_{k_l}AQ_{k_l}\|=\ve_m \,.
}
There are $2^{m-1}$ number of terms of that form. Let $S=\lim_{n\to\infty} S_n$. We need to show that the limit exists. To that end, for $m>n$, consider $S_m-S_n$. Using the formula \cref{eq:inclusion exclusion} and idea leading to upper bound \cref{eq:norm bound on S_n}, all the terms $P_{k_1}\dots P_{k_l}AQ_{k_1}\dots Q_{k_l}$ in $S_m-S_n$ will have $k_l\geq n+1$. Using \cref{eq:norm bound of Pk A Qk}, we have
\eq{
	\|S_m-S_n\| &\leq 2^n\ve_{n+1}+2^{n+1}\ve_{n+2}+\dots + 2^{m-1}\ve_m \\
	&\leq \frac{\ve}{2^{{n+1}}} +\dots + \frac{\ve}{2^m} \\
	&\leq \ve \sum_{k=n+1}^\infty \frac{1}{2^k}
}
which converges to $0$ as $n\to\infty$ (independent of $m$). 

We have $\|S\|\leq \ve$. Indeed, from \cref{eq:norm bound on S_n} and \cref{eq:norm bound of Pk A Qk}, we have
\eq{
	\|S\|\leq \sum_{k=1}^\infty 2^{k-1}\frac{1}{2^{k-1}}\frac{\ve}{2^k} \leq \ve \,.
}
Define
\eq{
    B = A-S \,.
}
Let us now show that
\eq{
	P_k B Q_k = 0,\ \forall k\in\NN
}
to prove that $B$ is indeed unaffected by the interactions originally present in $P_kAQ_k$. We show that $P_kS_nQ_k=P_kAQ_k$ for all $n\geq k$ by induction. From \cref{eq:inclusion exclusion}, we have the recursion relation
\eql{
	S_{n+1} = S_n + P_{n+1}AQ_{n+1} - P_{n+1}S_nQ_{n+1} \label{eq:recursion relation}
}
where $P_{n+1}AQ_{n+1}$ is from the first sum in \cref{eq:inclusion exclusion}, and $P_{n+1}S_nQ_{n+1}$ is from the rest of sums. Let $k\geq 1$ be arbitrary. It holds that $P_1S_1Q_1=P_1AQ_1$. Take $n=k$ in \cref{eq:recursion relation} and consider
\eq{
	P_{k+1} S_{k+1} Q_{k+1} &= P_{k+1} S_kQ_{k+1} + P_{k+1}P_{k+1}AQ_{k+1}Q_{k+1} - P_{k+1}P_{k+1}S_kQ_{k+1}Q_{k+1} \\
	&= P_{k+1}AQ_{k+1} \,.
}
Thus $P_kS_kQ_k=P_kAQ_k$ holds for all $k\geq 1$. Suppose $P_kS_{n}Q_k=P_kAQ_k$ holds. Using \cref{eq:inclusion exclusion}, we have
\eq{
	P_kS_{n+1}Q_k &= P_kS_nQ_k + P_kP_{n+1}AQ_{n+1}Q_k - P_kP_{n+1}S_nQ_{n+1}Q_k \\
	&=P_kS_nQ_k + P_{n+1}P_kAQ_kQ_{n+1} - P_{n+1}P_kS_nQ_kQ_{n+1} \\
	&= P_kAQ_k
}
where in the last equality we used the induction assumption.

Since $P_kAQ_k$ is compact, it follows that $S_n$ in \cref{eq:inclusion exclusion} is compact, and that its norm limit $S$ is compact. Therefore $A-B=S$ is compact.

\end{proof}

\section{The real symmetry classes}\label{sec:Real symmetries}
In this section, we treat the lower eight rows of the \nameref{table:Kitaev}, namely, those symmetry classes which involve an anti-$\CC$-linear symmetry operator (either time-reversal or particle-hole). 

We consider that space of spherically-local, self-adjoint unitary operators $\calSU(\calL_{d,N})$ that satisfies certain symmetry constraint. Here $\calL_{d,N}$ is the \Cstar-algebra of spherically-local operators acting on $\calH_d\otimes \CC^N\equiv \calH_{d,N}$; see \cref{def:spherically-local}.

\begin{table}
	\begin{center}
		\begin{tabular}{c|c|c|c}
			Class ($\Sigma$) & Structures & Algebraic Properties & Constraints on Systems \\\hline
			A & --- & --- & --- \\
			AIII & $\Pi$ & --- & $\{H, \Pi\} = 0$ \rule[-0.9ex]{0pt}{0pt}\\\hline
			AI & $\Theta$ & $\Theta^2 = +\Id_N$ & $[H, \Theta] = 0$ \rule{0pt}{2.6ex}\\ 
			BDI & $\Pi, \Theta$ & $\Theta^2 = +\Id_N, \, [\Theta, \Pi] = 0$ & $\{H, \Pi\} = 0 , \, [H, \Theta] = 0$ \\
			D & $\Xi$ & $\Xi^2 = +\Id_N$ & $\{H, \Xi\} = 0$ \\
			DIII & $\Pi, \Theta$ & $\Theta^2 = -\Id_N, \, \{\Theta, \Pi\} = 0$ & $\{H, \Pi\} = 0 , \, [H, \Theta] = 0$ \\
			AII & $\Theta$ & $\Theta^2 = -\Id_N$ & $[H, \Theta] = 0$ \\
			CII & $\Pi, \Theta$ & $\Theta^2 = -\Id_N, \, [\Theta, \Pi] = 0$ & $\{H, \Pi\} = 0 , \, [H, \Theta] = 0$ \\
			C & $\Xi$ & $\Xi^2 = -\Id_N$ & $\{H, \Xi\} = 0$ \\
			CI & $\Pi, \Theta$ & $\Theta^2 = +\Id_N, \, \{\Theta, \Pi\} = 0$ & $\{H, \Pi\} = 0 , \, [H, \Theta] = 0$ \\
		\end{tabular}
        \caption{The algebraic properties of symmetry operators and the constraint on the systems. In the structures, the symmetry operators act only on the internal degrees of freedom $\CC^N$. In the constraints above, the operators act on the full tensor product space as $\calC\otimes\Theta,\, \calC\otimes\Xi,\, \Id\otimes\Pi$.}
        \label{table:AZ constraint}
	\end{center}
\end{table}

\begin{defn}[The Altland--Zirnbauer symmetry classes]\label{defn:AZ symmetry classes explicit}
    Let $\Sigma$ be a symmetry class in \cref{table:AZ constraint}. Let the internal degrees of freedom $N\in\NN$ by arbitrary for non-chiral symmetric classes, and let $N=2W$ be even for chiral symmetric classes. We fix the chiral symmetry operator to take the form
    \eql{\label{eq:chiral symmetry operator}
    \Pi=\begin{bmatrix}
        \Id & 0 \\ 0 & -\Id
    \end{bmatrix}
    }
    with respect to the decomposition $\calH_d\otimes \CC^{2W} = \calH_d\otimes \CC^W \oplus \calH_d\otimes \CC^W$ of positive and negative chiral subspaces.
    We fix the symmetry operators, specified in the structure column, such that 
    \eql{\label{eq:symmetry operators hyper spherically local}
    [\Theta,\widehat{X}_j\otimes \Id_N] = 0, \quad [\Xi,\widehat{X}_j\otimes \Id_N] = 0,\quad j=1,\dots,d,
    }
    that is, the symmetry operators are hyper-spherically-local (as opposed to strictly acting within $\CC^{N}$).
    The chosen symmetry operators satisfy the algebraic properties specified in \cref{table:AZ constraint}. We define $\calAZ_{d,N}^{\Sigma}$ to be the space of operators $H\in\calSU(\calL_{d,N})$ such that $H$ satisfies the constraints in \cref{table:AZ constraint}. 
    We define $\calAZ_{d,N}^{\Sigma,\NT}$ as the space of operators $H\in\calAZ_{d,N}^\Sigma$ such that the projections $(H+\Id)/2$ is bulk-non-trivial as in \cref{defn:bulk non triviality}.
\end{defn}

The assumptions made in \cref{defn:AZ symmetry classes explicit} are non-vacuous, and each space $\calAZ_{d,N}^{\Sigma}$ is non-empty, as we will verify later.

As shown in the next result, the symmetry spaces $\calAZ_{d,N}^{\Sigma}$ are well-defined irrespective of particular choices of symmetry operators, and hence we do not include them in the notation. Denote the three Pauli spin matrices to be
\eq{
\sigma_x = \begin{bmatrix}
    0 & 1 \\ 1 & 0
\end{bmatrix},\quad \sigma_y=\begin{bmatrix}
    0 & -\ii \\ \ii & 0
\end{bmatrix},\quad \sigma_z=\begin{bmatrix}
    1 & 0 \\ 0 & -1
\end{bmatrix}.
}
Let $\calC$ be the complex conjugation operator which is the usual complex conjugation on $\CC^N$, and on $\calH_d$ is defined as $\calC (\sum\alpha \delta_x) = \sum \bar{\alpha} \delta_x$.

\begin{lem}\label{lem:symmtery spaces irrespective of symmetry operators}
    For each class $\Sigma$, the spaces $\calAZ_{d,N}^{\Sigma}$ defined via different symmetry operators are unitarily equivalent with a unitary that is hyper-spherically-local. In particular, we have
    \begin{enumerate}
    \item Let $\Sigma\in\Set{\mathrm{AI},\mathrm{D}}$ and $S\in\Set{\Theta,\Xi}$ be the respective symmetry operators. There exists a hyper-spherically-local unitary $U$ on $\calH_{d,N}$ such that $U^\ast SU=\calC$.
    \item Let $\Sigma\in\Set{\mathrm{AII},\mathrm{C}}$ and $S\in\Set{\Theta,\Xi}$ be the respective symmetry operators. 
    Let $M\in 2\NN$. Then there exists a hyper-spherically-local unitary $U:\calH_{d,M}\to\calH_{d,N}$ such that $U^\ast SU=-i\sigma_y\calC$.
    \item For class $\Sigma\in \Set{\mathrm{BDI},\mathrm{DIII},\mathrm{CI}}$, there exists a hyper-spherically-local $U$ on $\calH_{d,N}$ that commutes with $\Pi$ in \cref{eq:chiral symmetry operator} and $U^\ast\Theta U$ takes the form of $\calC$, $-i\sigma_y\calC$, $\sigma_x\calC$, respectively. 
    \item For class $\mathrm{CII}$, let $M\in 4\NN$. There exists a hyper-spherically-local unitary $U:\calH_{d,M}\to\calH_{d,N}$ that commute with $\Pi$ and 
    \eq{
    U^\ast\Theta U= \begin{bmatrix}
        -i\sigma_y\calC & 0 \\
        0 & -i\sigma_y\calC
    \end{bmatrix}.
    }

\end{enumerate}
\end{lem}
\begin{proof}
Let us prove the case for class $\mathrm{DIII}$ and leave the rest for the reader. We claim that there exists a hyper-spherically-local unitary operator $U$ on $\calH_d\otimes \CC^N$ that commutes with $\Pi$ and satisfies $U^\ast\Theta U=-i\sigma_y\calC$.

Let $\calV_{\pm} = \calH_d\otimes \CC^W$ be the positive and negative chiral spaces. Since $\Theta$ anti-commutes with $\Pi$, it follows that $\Theta$ maps anti-unitarily $\calV_\pm$ to $\calV_\mp$. Let $\Lambda^\pm_I$ be as in \cref{eq:lambda I on sphere} by on the positive and negative chiral subspaces.
For each $z\in \bS^{d-1}_\ZZ$, pick an orthonormal basis $\{\vf^+_{j,z}\}_{j\in\NN}$ for $\im\Lambda_{\{z\}}^+$. Since $\calV_+=\oplus_{z\in\bS^{d-1}_\ZZ} \im\Lambda_{\{z\}}^+$, the set $\{\vf_{j,z}\}_{j\in\NN,z\in\bS^{d-1}_\ZZ}$ forms an orthonormal basis for $\calV_+$. Define 
\eq{
\vf^-_{j,z}:=\Theta \vf^+_{j,z}.
}
We have $\Theta \vf^-_{j,z} = -\vf^+_{j,z}$.
Since $\Theta$ is hyper-spherically-local, by \cref{lem:hyper spherically local reduced by rays}, it follows that $\vf^-_{j,z}\in \im\Lambda_{\{z\}}^-$. 
Let $\{\eta_{j,z}^{\pm}\}_{j\in\NN}$ be an enumeration of the standard basis (those of the form $\delta_x\otimes e_i$) in $\im\Lambda_{\{z\}}^\pm$.
Let $U$ be the unitary operator that sends the standard basis $\eta_{j,z}^\pm$ in $\im\Lambda_{\{z\}}^\pm$ to $\vf^{\pm}_{j,z}$ for the respective chirality and points $z$ in $\bS^{d-1}_\ZZ$. By construction $U$ is spherically-hyper-local.
We have $U\Pi=\Pi U$. Indeed, on each invariant subspaces $\im\Lambda_{\{z\}}$ we have
\eq{
U\Pi \left(\sum_j\alpha_j\eta_{j,z}^+ + \sum_j\beta_j\eta_{j,z}^-\right) &= U\left(\sum_j\alpha_j\eta_{j,z}^+ - \sum_j\beta_j\eta_{j,z}^-\right) \\
&= \left(\sum_j\alpha_j\vf_{j,z}^+ - \sum_j\beta_j\vf_{j,z}^-\right) \\
&= \Pi U \left(\sum_j\alpha_j\eta_{j,z}^+ + \sum_j\beta_j\eta_{j,z}^-\right).
}
Let us verify $U^\ast\Theta U=-i\sigma_y\calC$. We have
\eq{
\Theta U \left(\sum_j\alpha_j\eta_{j,z}^+ + \sum_j\beta_j\eta_{j,z}^-\right) & = \Theta \left(\sum_j\alpha_j\vf_{j,z}^+ + \sum_j\beta_j\vf_{j,z}^-\right) \\
&= \left(\sum_j\bar{\alpha}_j\vf_{j,z}^- - \sum_j\bar{\beta}_j\vf_{j,z}^+\right) \\
&= U\left(\sum_j-\bar{\beta}_j\eta_{j,z}^+ + \sum_j\bar{\alpha}_j\eta_{j,z}^-\right) \\
&= U (-i\sigma_y\calC) \left(\sum_j\alpha_j\eta_{j,z}^+ + \sum_j\beta_j\eta_{j,z}^-\right) .
}
\end{proof}

\begin{lem}\label{lem:hyper spherically local reduced by rays}
    $A\in\calB(\calH_d\otimes \CC^N)$ is hyper-spherically-local iff $[A,\Lambda_I]=0$ for all measurable subsets $I\subset \bS^{d-1}$. In particular, it suffices to consider $[A,\Lambda_{\{z\}}]=0$ for all $z\in\bS^{d-1}_\ZZ$.
\end{lem}
\begin{proof}
    This follows from \cref{eq:spectral measure on Xd} and \cref{eq:lambda is spectral measure}. Indeed, if $A$ is hyper-spherically-local, then $[A,\rho(f)]=0$ for all $f\in C(\bS^{d-1})$ where $\rho:C(\bS^{d-1})\to \calX_{d,N}$ is the isomorphism \cref{eq:isomorphism continuous functions on sphere to algebra generated by Xis} and $\calX_{d,N}$ is the \Cstar-algebra generated by $\widehat{X}_1\otimes \Id_N,\dots,\widehat{X}_d\otimes \Id_N$. In particular, we have $[A,\rho(\lambda\mapsto\lambda)]=[A,\int\lambda dE(\lambda)]=0$ where we use \cref{eq:spectral measure on Xd}. It follows from \cite[Section 41, Theorem 2]{halmos2017introduction} that $[A,E(I)]=0$ for all measurable subsets $I\subset \bS^{d-1}$. Using \cref{eq:lambda is spectral measure}, it follows that $[A,\Lambda_I]=0$. The converse follows from \cite[Section 37, Theorem 4]{halmos2017introduction}.
\end{proof}


\subsection{Real spherically-local algebra}\label{subsec:real spherically local algebra}
We make a detour to studying the real spherically-local algebra.  For operators $A\in\calB(\calH_d\otimes \CC^N)$, we define $\overline{A}:=\calC A\calC$. The real spherically-local algebra $\calL_{d}^\RR$ is the real \Cstar-algebra defined by
\eq{
\calL_d^\RR := \Set{A\in\calL_{d} | \overline{A}=A}
}

The real symmetry classes are related to real Clifford algebra. The real Clifford algebra $\Cli_{p,q}$ is the graded real-$\ast$-algebra generated by $p$ self-adjoint generators which square to 1, and $q$ anti-self-adjoint generators which square to $-1$ and all generators anti-commute pairwise. The grading is defined by declaring the generators to be odd \cite{AtiyahBottShapiro1964}. We will denote $\operatorname{st}$ the natural grading automorphism on the Clifford algebras.

We unified the description of various symmetry classes in \cref{defn:AZ symmetry classes explicit} into one formula in the context of real shperically-local algebras using Clifford algebra and the space of odd, self-adjoint unitaries:

\begin{defn}\label{defn:odd SAU}
    Let $\calA$ be a $\ZZ_2$-graded, unital \Cstar-algebra with grading $\gamma$. We denote the space of odd, self-adjoint, unitary element as
    \eq{
    \calSU_o(\calA) := \Set{A\in\calSU(\calA) | \gamma(A)=-A}
    }
    where  $\calSU(\calA)$ denotes the self-adjoint unitary elements in $\calA$. The space $\calSU_o(\calA)$ is equipped with the topology induced by the \Cstar-norm. See \cite[Definition 2.1]{roe2004paschke} and \cite[Definition 2.1]{daele1988kI}.
\end{defn}

\begin{table}[ht]
  \centering
  \begin{tabular}{c|c|c|c | c}
    AZ   & $n$ & $p$ & $q$ & $\Cli_{p,q}$ \\ 
    \hline
    AI   & 0 & 1 & 0 & $\RR\oplus \RR$\\ 
    BDI  & 1 & 1 & 1 & $M_2(\RR)$\\ 
    D    & 2 & 0 & 1 & $\CC$\\ 
    DIII & 3 & 0 & 2 & $\HH$\\ 
    AII  & 4 & 0 & 3 & $\HH\oplus \HH$\\ 
    CII  & 5 & 0 & 4 & $M_2(\HH)$\\ 
    C    & 6 & 0 & 5 & $M_4(\CC)$\\ 
    CI   & 7 & 0 & 6 & $M_8(\RR)$\\ 
  \end{tabular}
  \caption{Smallest $p,q$ for each symmetry class that satisfies $n+p-q=1$.}
  \label{tab:p q pair in clifford for symmetry}
\end{table}

\begin{prop}\label{prop:clifford description of AZ classes shperically local}
    We have the topological homeomorphism:
    \eq{
    \calSU_o(\calL_{d}^\RR\otimes \Cli_{p,q}) \cong \calAZ_{d,N}^{\Sigma}
    }
    where $p,q\geq 0$ satisfy $n+p-q=1$ for $n=0,1,2,\dots,7$ corresponding to the enumeration of rows in real classes $\Sigma$ of \cref{tab:p q pair in clifford for symmetry} from top to bottom. 
    The algebra $\calL_{d}^\RR\otimes \Cli_{p,q}$ is graded by $\operatorname{id}\otimes \gamma$ where $\gamma$ is the natural $\ZZ_2$-grading on $\Cli_{p,q}$. In particular, we have
    \eq{
\calSU_o(\calL^\RR_{d}\otimes \Cli_{1,0}) &\cong \calAZ_{d,N}^\azAI \cong \calP(\calL_d^\RR)\cong \calSU(\calL_d^\RR)\\
\calSU_o(\calL^\RR_{d}\otimes \Cli_{1,1}) &\cong \calAZ_{d,N}^\azBDI \cong \calU(\calL_d^\RR)\\
\calSU_o(\calL^\RR_{d}\otimes \Cli_{0,1}) &\cong \calAZ_{d,N}^\azD \cong \Set{P\in\calP(\calL_d) | \overline{P} = P^\perp} \cong \Set{U\in\calSU(\calL_d) | \overline{U}=-U^\ast} \\
\calSU_o(\calL^\RR_{d}\otimes \Cli_{0,2}) &\cong \calAZ_{d,N}^\azDIII\cong \Set{U\in\calU(\calL_d) | \overline{U}=-U^\ast} \\
\calSU_o(\calL^\RR_{d}\otimes \Cli_{0,3}) &\cong \calAZ_{d,N}^\azAII \cong \calP(\calL_d^\RR\otimes \HH) \cong \calSU(\calL_d^\RR\otimes \HH)\\
\calSU_o(\calL^\RR_{d}\otimes \Cli_{0,4}) &\cong \calAZ_{d,N}^\azCII \cong \calU(\calL_d^\RR\otimes \HH) \\
\calSU_o(\calL^\RR_{d}\otimes \Cli_{0,5}) &\cong \calAZ_{d,N}^\azC \cong \Set{P\in\calP(\calL_d) | \calJ P\calJ = -P^\perp} \cong \Set{U\in\calSU(\calL_d) | \calJ U\calJ = U^\ast} \\
\calSU_o(\calL^\RR_{d}\otimes \Cli_{0,6}) &\cong \calAZ_{d,N}^\azCI \cong \Set{U\in\calU(\calL_d) | \calJ U\calJ = U^\ast} 
}
where $\calJ$ is any spherically-hyper-local anti-unitary operator on $\calH_d$ such that $\calJ^2=-\Id$.

\end{prop}

For \cref{prop:clifford description of AZ classes shperically local}, it suffices to show for the homeomorphism for $p,q$ in \cref{tab:p q pair in clifford for symmetry}. Indeed, denote $\operatorname{st}$ the natural grading automorphism on the Clifford algebras, using \cite[Corollary 3.6]{kellendonk2017c}, we have the graded isomorphism $(M_2(\RR)\otimes \Cli_{p,q},\operatorname{id}\otimes \operatorname{st})\cong (\Cli_{p+1,q+1},\operatorname{st})$. Then
\eq{
(\calL_{d}^\RR\otimes \Cli_{p+1,q+1},\operatorname{id}\otimes \operatorname{st}) \cong (\calL_{d}^\RR\otimes M_2(\RR)\otimes \Cli_{p,q},\operatorname{id}\otimes\operatorname{id}\otimes \operatorname{st}) \cong (\calL_{d,2}^\RR\otimes \Cli_{p,q},\operatorname{id}\otimes\operatorname{st})
}
where $\calL_{d,2}^\RR$ is the space of $A\in\calL_{d,2}$ such that $\overline{A}=A$. Using re-dimerization \cref{prop:re-dimerization}, we have the unitary equivalence $\calL_{d,2}^\RR \cong \calL_{d}^\RR$, which leads to the graded isomorphism $(\calL_{d,2}^\RR,\operatorname{id})\cong (\calL_{d}^\RR,\operatorname{id})$. Indeed, this is because the re-dimerization unitary $R:\calH\to\calH_2$ is real, i.e., it satisfies $\overline{R}=R$. Therefore
\eq{
\calSU_o(\calL_d^\RR\otimes \Cli_{p+1,q+1})\cong \calSU_o(\calL_{d,2}^\RR\otimes \Cli_{p,q}) \cong  \calSU_o(\calL_d^\RR\otimes \Cli_{p,q}).
}

\begin{proof}[Proof of \cref{prop:clifford description of AZ classes shperically local}]

The result follows from using the canonical form \cref{lem:symmtery spaces irrespective of symmetry operators} and the representations of real Clifford algebra in \cref{sec:odd structures in clifford algebra}, and the re-dimerization \cref{prop:re-dimerization}. We show for the case of DIII and leave the rest for the reader. Let the internal degrees of freedom $N=2W$ be even, let $\Pi$ be the \cref{eq:chiral symmetry operator}, and let $\Theta$ be any anti-linear operator such that it is hyper-spherically-local, $\Theta^2=-\Id$ and $[\Theta,\Pi]=0$. 
We are interested in the space
\eq{
\calAZ_{d,N}^{\azDIII} = \Set{H\in\calSU(\calL_{d,N}) | \{H,\Pi\}=0,\, [H,\Theta]=0}.
}
Using \cref{lem:symmtery spaces irrespective of symmetry operators}, we can choose $\Theta$ to be 
\eq{
\calC\otimes \begin{bmatrix}
    0 & -\calC_W \\ \calC_W & 0
\end{bmatrix} =: -i\sigma_y\calC
}
with respect to the chiral grading. Using the choice of $\Theta= -i\sigma_y \calC$ and re-dimerization, we have the homeomorphism
\eql{\label{eq:canonical space DIII}
\calAZ_{d,N}^\azDIII \cong \Set{U\in\calU(\calL_d) | U^\ast=-\overline{U}}=:\calU^\azDIII_d .
}
Indeed, if $H$ is self-adjoint, unitary and anti-commutes with $\Pi$, then $H$ is of the form 
\eq{
H=\begin{bmatrix}
    0 & U^\ast \\ U & 0
\end{bmatrix}
}
for $U\in\calU(\calL_{d,W})$. Using $[H,\Theta]=0$, it follows that $U^\ast = -\overline{U}$. Consider the re-dimerization unitary $R:\calH_d\to \calH_{d,W}$ in \cref{prop:re-dimerization}. Since it maps between standard bases, we have $\overline{R}=R$. If $U\in\calU(\calL_{d,N})$ satisfies $U^\ast=-\overline{U}$, then $R^\ast U R\in\calU(\calL_d)$ and
\eq{
(R^\ast U R)^\ast = R^\ast U^\ast R = R^\ast (-\overline{U})R = -\overline{R^\ast U R}.
}
where we use $\overline{R}=R$ in the last equality.
This establishes \cref{eq:canonical space DIII}. The homeomorphism \cref{eq:canonical space DIII} also implies that we have the homeomorphism $\calAZ_{d,N}^\azDIII\cong \calAZ_{d,M}^\azDIII$ for all $M,N\in 2\NN$. On the other hand, using the representations of Clifford algebra in \cref{sec:odd structures in clifford algebra}, we have
    \begin{equation}
    \label{eq:diii SUo}
    \calSU_o(\calL_{d}^\RR\otimes \Cli_{0,2}) = \Set{\begin{bmatrix}
        0 & B \\ -\overline{B} & 0
    \end{bmatrix} \in \calSU(\calL_{d}^\RR\otimes M_2(\CC))}
    \end{equation}
which is homeomorphic to $\Set{B\in \calU(\calL_{d}) | B^\ast=-\overline{B}} = \calU_d^\azDIII$, exactly \cref{eq:canonical space DIII}.

\end{proof}



\subsection{Classification}

In this subsection we complete the classification for the eight real Altland--Zirnbauer classes at the level of path-connected components.
Using the reformulation from \cref{subsec:real spherically local algebra}, each real symmetry space can be identified with a space of odd self-adjoint unitaries in a graded real algebra of the form
$\calL_d^\RR\otimes \Cli_{p,q}$ (\cref{prop:clifford description of AZ classes shperically local}), so the remaining task is to show that the corresponding strong invariants are not merely $K$-theoretic, but in fact \emph{complete} in the sense of $\pi_0$.
Equivalently, we prove that the natural map from path-components of the bulk-non-trivial symmetry spaces to the relevant van Daele $K$-groups is a bijection.
This establishes the real rows of the Kitaev table in the $\pi_0$-sense, and, together with the complex classes handled in \cref{sec:The classification of bulk-non-trivial spherically-local projections and unitaries}, completes the proof of \cref{thm:Kitaev table agrees with path-connected components of non-trivial insulators}.

\begin{thm}\label{thm:classification for real classes}
    For any $\Sigma$ in the eight real Altland--Zirnbauer symmetry classes, the set of path-connected components of $\calAZ_{d,N}^{\Sigma,\NT}$ agrees with the relevant entry within the \nameref{table:Kitaev}.
\end{thm}

The proof of \cref{thm:classification for real classes} relies on lifting $K$-theory classes defined on $\calSU_o$ \cref{defn:odd SAU} to the set of path-connected components of the symmetry spaces in \cref{prop:clifford description of AZ classes shperically local}. The $K$-theory builds on top of the $\calSU_o$ leads to the van Daele $K$-theory, which we now briefly describe. 

Let $\calA$ be a $\ZZ_2$-graded, unital, real \Cstar-algebra with $\ZZ_2$-grading $\gamma:\calA\to\calA$. Let $M_n(\calA)$ be the graded matrix algebra of $\calA$ where the $\ZZ_2$-grading on $M_n(\calA)$ is obtained by applying $\gamma$ elementwise. 
Consider the space $\calSU_o(M_n(\calA))$ and the space
$\bigsqcup_{n=1}^\infty \calSU_o(M_n(\calA))$.
There is a natural binary operation $\oplus$ on $\bigsqcup_{n=1}^\infty \calSU_o(M_n(\calA))$ by 
\eq{
    a\oplus b = \begin{bmatrix}
        a & 0 \\ 0 & b
    \end{bmatrix} \in\calSU_o(M_{n+m}(\calA))
}
for $a\in\calSU_o(M_n(\calA))$ and $b\in\calSU_o(M_m(\calA))$.
Suppose $\calSU_o(\calA)$ is nonempty, choose an element $e\in \calSU_o(\calA)$. Let $e_n\in\calSU_o(M_n(\calA))$ be the direct sum of $n$-copies of $e$.
Define a relation $\sim_e$ on $\bigsqcup_{n=1}^\infty \calSU_o(M_n(\calA))$ as follows. For $a\in\calSU_o(M_n(\calA))$ and $b\in\calSU_o(M_m(\calA))$, write $a\sim_e b$ if there exists $j,k\geq 0$ such that $a\oplus e_j \sim_h b\oplus e_k$. 
We define
\eql{\label{eq:Ke group}
\rmDK_e(\calA) := \calF_\infty(\calA) / \sim_e. 
}
Let $[a]_e$ denote the equivalence class containing $a$ in $\bigsqcup_{n=1}^\infty \calSU_o(M_n(\calA))$. Define a binary operation $+$ on $\rmDK_e(\calA)$ by $[a]_e+[b]_e = [a\oplus b]_e$. One can show that $\rmDK_e(\calA)$ equipped with $+$ is a commutative semigroup with a neutral element $0:=[e]_e$. 
Van Daele showed in \cite[Proposition 2.12]{daele1988kI} that if there exists $e\in\calSU_o(\calA)$ such that $e\sim_h -e$ in $\calSU_o(\calA)$, then $\rmDK_e(\calA)$ is a group. 
If such an element $e$ does not exist, van Daele showed that one may augment the algebra $\calA$ to $M_4(\calA)$ with certain $\ZZ_2$-grading and construct $e\in M_4(\calA)$ with the particular property, and define $\rmDK(\calA):=\rmDK_e(M_4(\calA))$. See \cite[Section 3]{daele1988kI}. Moreover, $\rmDK(\calA)$ is isomorphic to $\rmDK_{e}(\calA)$ if $\calA$ has an element $e\in\calSU_o(\calA)$ such that $e\sim_h -e$ in $\calSU_o(\calA)$. For higher van Daele $K$-group, we define
\eq{
\rmDK_n(\calA):= \rmDK(\calA\widehat{\otimes}\Cli_{p,q})
}
where $n+p-q=1$. See also \cite{roe2004paschke} and \cite[Definition 5.5]{kellendonk2017c}

\begin{prop}\label{prop:computation for van daele K group of real spherically local algebra}
    The group $\rmDK_n(\calL_{d}^\RR)$ is isomorphic to the values in \Cref{table:Kitaev}, where $n=0,1,\dots,7$ corresponds to $\Sigma=\mathrm{AI},\mathrm{BDI},\dots,\mathrm{CI}$.
\end{prop}
\begin{proof}
    We have
    \eq{
    \rmDK_n(\calL_d^\RR) \cong \rmDK_n(\calD(C(\bS^{d-1},\RR))) \cong \rmDK_n(\calD(\widetilde{C_0(\RR^{d-1}))}) \cong \rmKKO_{n-1}(C_0(\RR^{d-1}),\RR)
    }
    where the first isomorphism uses \cref{eq:dual characterization of local algebra} and the second one identifies the unitization of $C_0(\RR^{d-1},\RR)$ as $C(\bS^{d-1},\RR)$, and the third isomorphism uses \cite[Proposition 4.2]{roe2004paschke} that identified van Daele’s $K$-theory of the dual of an algebra with Kasparov's $K$-homology in the real case. The rest of the calculation is standard
    \eq{
    \rmKKO_{n-1}(C_0(\RR^{d-1}),\RR) &\cong \rmKKO_{n-1-(d-1)}(\RR,\RR)\\ &\cong \rmKO_{n-d}(\RR)\cong \begin{cases}
        \ZZ & n-d\equiv 0,4 \Mod 8  \\
        \ZZ_2 & n-d\equiv 1,2\Mod 8 \\
        0 & n-d\equiv 3,5,6,7 \Mod 8
    \end{cases}
    }
    where the first isomorphism uses \cite[Theorem 2.5.2]{schroder1993k} and the second isomorphism uses \cite[Theorem 2.3.8]{schroder1993k} that relates $KK$-group to the $K$-group of real \Cstar-algebra, and the last isomorphism can be found on \cite[p.~23]{schroder1993k}.
\end{proof}


At this stage one can run the same $\pi_0$-lifting argument separately in each of the eight real symmetry classes.
The overall strategy is the same: starting from the van Daele description of the relevant $K$-group, one shows that the defining stabilization relation can be \emph{compressed} to an actual homotopy in the original symmetry space by a pinning procedure (deforming a given symmetry-constrained operator so that, outside a spherically-proper region, it agrees with a fixed reference element).
The only changes from one class to another are (i) the precise symmetry relation dictated by the corresponding Clifford generators in \cref{prop:clifford description of AZ classes shperically local}, and (ii) the choice of a convenient basepoint element in $\calSU_o$ (in some classes the identity is admissible, while in others one uses a local dimerization as a canonical substitute, cf. \cref{defn:diii dimer operator}).
These variations introduce additional algebraic bookkeeping and some technical nuisance, but no new conceptual input beyond the argument presented below.

Before addressing the genuinely new symmetry constraints in the remaining real rows, it is worth noting that \cref{sec:The classification of bulk-non-trivial spherically-local projections and unitaries} already settles, in essence, half of the real Altland--Zirnbauer classes. Indeed, \cref{prop:clifford description of AZ classes shperically local} identifies the class AI and BDI spaces as the \emph{real} versions of our basic objects (homeomorphic to $\mathcal P(\mathcal L_d^{\mathbb R})$ and $\mathcal U(\mathcal L_d^{\mathbb R})$), while class AII and CII are the corresponding \emph{quaternionic} analogues (homeomorphic to $\mathcal P(\mathcal L_d^{\mathbb R}\otimes\mathbb H)$ and $\mathcal U(\mathcal L_d^{\mathbb R}\otimes\mathbb H)$).
Consequently, for these four rows there is no new conceptual work to do: the pinning-stabilization-compression mechanism from \cref{sec:The classification of bulk-non-trivial spherically-local projections and unitaries} carries over after imposing the appropriate real/quaternionic condition.

On the other hand, the Altland--Zirnbauer table is arranged into adjacent non-chiral/chiral pairs
$
(\mathrm{AI},\mathrm{BDI}), (\mathrm{D},\mathrm{DIII}), (\mathrm{AII},\mathrm{CII}), (\mathrm{C},\mathrm{CI})
$
so that, within each pair, understanding the chiral (unitary) formulation essentially gives the non-chiral (projection) formulation.

For readability, we therefore present the proof in full detail only for class DIII, being one of the non-obvious cases (as opposed to the chiral or non-chiral real or quaternionic cases).
Carrying out the same steps with the appropriate substitutions in the remaining real classes yields \cref{thm:classification for real classes} in complete generality.
In particular, combining \cref{thm:classification for real classes} with the complex-class results of \cref{sec:The classification of bulk-non-trivial spherically-local projections and unitaries} proves \cref{thm:Kitaev table agrees with path-connected components of non-trivial insulators}.


Similar to the proof in the complex case, we will focus on the space of invertible operators rather than the unitaries. To that end, we define
\eq{
\calG_d^{\azDIII} := \Set{A\in\calG(\calL_d) | A^{-1}=-\overline{A}} \supset \calU_d^{\azDIII}.
}
It is clear that $\Id\notin \calU_d^\azDIII$.
Nonetheless, the space $\calU_d^{\azDIII}$ is nonempty. Indeed, we have
\begin{defn}[dimer operator]\label{defn:diii dimer operator}
Consider a nearest-neighbor partition $\{x_k\}_{k\in\NN}\cup\{y_k\}_{k\in\NN}=\ZZ^d$ of the lattice points and define $E\in\calB(\calH_d)$ to be
\eql{\label{eq:diii dimer operator nearest neighbor}
E \delta_{x_k} = \delta_{y_k},\quad E \delta_{y_k} = -\delta_{x_k}
}
and extend linearly to all vectors in $\calH_d$.
\end{defn}

We have
\eq{
E\overline{E} \left(\sum_k\alpha_k\delta_{x_k}+\beta_k\delta_{y_k}\right) = E\calC \left(\sum_{k}\bar{\alpha}_k \delta_{y_k} + \bar{\beta}_k(-\delta_{x_k})\right) = -\left(\sum_k\alpha_k\delta_{x_k}+\beta_k\delta_{y_k}\right)
}
and similarly $\overline{E}E=-\Id$. In particular, $E$ is real, i.e., $\overline{E}=E$. Therefore $E\in\calU_d^{\azDIII}$. We call $E$ the dimer operator, which will be a canonical operator replacing the identity operator. 



\begin{lem}\label{lem:diii congruence}
    If $A\in\calG^\azDIII_d$ and $V\in\calG(\calL_d)$, then $VA\overline{V}^{-1}\in \calG^\azDIII_d$.
\end{lem}
\begin{proof}
    We have
    \eq{
    (VA\overline{V}^{-1})^{-1} = \overline{V} A^{-1} V^{-1} = -\overline{V}\calC A \calC V^{-1} = -\calC VA\overline{V}^{-1}\calC.
    }
\end{proof}

\begin{lem}\label{lem:diii symmetrization}
    Let $A\in\calG(\calL_d)$ such that $\sigma(-\overline{A}A)\cap (-\infty,0]=\varnothing$. Then the symmetrization
    \eql{\label{eq:diii symmetrization}
    \Psi(A):= A(-\overline{A}A)^{-1/2}\in \calG^\azDIII_d 
    }
    is well-defined using the holomorphic square root.
    Furthermore, if $A\in\calG^\azDIII_d$, then $\Psi(A)=A$; if $V\in\calG(\calL_d)$, then $\Psi(VA\overline{V}^{-1})=V\Psi(A)\overline{V}^{-1}$.
\end{lem}
\begin{proof}
Let $B=-\overline{A}A$. Using holomorphic functional calculus, we have 
\eq{
\overline{B^{-1/2}}  = (\overline{B})^{-1/2} =
(-A\overline{A})^{-1/2} = (ABA^{-1})^{-1/2} = AB^{-1/2}A^{-1}.
}
Then
\eq{
-\overline{\Psi(A)}\Psi(A) = \overline{AB^{-1/2}}AB^{-1/2} = -\overline{A} (AB^{-1/2}A^{-1}) AB^{-1/2} = -\overline{A} A B^{-1} = \Id.
}
\end{proof}

We will mostly apply symmetrization of \cref{lem:diii symmetrization} on $A\in\calG(\calL_d)$ that is close to $\calG_d^\azDIII$. The assumption holds. Indeed, if $B\in\calG_d^\azDIII$, then $-\overline{B}B=-\Id$, and if $A$ is close to $B$, we can use \cite[Theorem 10.20]{rudin1991functional} to get that $\sigma(-\overline{A}A)$ lies in a neighborhood of $\sigma(-\overline{B}B)=\{1\}$, which avoids the branch cut $(-\infty,0]$ needed to define the holomorphic square root.

The relation $A^{-1}=-\overline{A}$ forces certain uniform non-collinearity of pairs $\{A\delta_x,\delta_x\}$. Here we give a local version.
\begin{lem}\label{lem:diii non-collinearity}
    Let $x\in\ZZ^d$ and suppose $A\in\calB(\calH_d)$ satisfies $-\overline{A}A\delta_x=\delta_x$. Then
    \eql{\label{eq:non-collinearity lower bound}
    \dist(A\delta_x,\CC\delta_x) = \|\Lambda_{\{x\}}^\perp A\delta_x\| \geq \frac{1}{\|A\|}.
    }
    In particular, $A\delta_x$ and $\delta_x$ are linearly independent.
\end{lem}
\begin{proof}
 Write $A\delta_x = \alpha\delta_x + w_x$ for $w_x\perp \delta_x$ and $\alpha=\braket{\delta_x,A\delta_x}$. We claim that 
\eq{
\dist(A\delta_x,\CC\delta_x) = \|w_x\|\geq \frac{1}{\|A\|}
}
for all $x\in \ZZ^d$. We have
\eq{
A(\calC A\delta_x) = A\calC (\alpha\delta_x + w_x) = \bar{\alpha}A\delta_x + A\calC w_x.
}
Then
\eq{
\braket{\delta_x,A(\calC A\delta_x)} = \braket{\delta_x,\bar{\alpha}A\delta_x} + \braket{\delta_x,A\calC w_x} = |\alpha|^2 + \braket{\delta_x,A\calC w_x}.
}
On the other hand, using $-A\overline{A}\delta_x=\delta_x$, we have
\eq{
A (\calC A\delta_x) = A\overline{A}\delta_x =-\delta_x. 
}
Thus $\braket{\delta_x,A (\calC A\delta_x)}=-1$ as well. This gives
\eq{
1 + |\alpha|^2 = |\braket{\delta_x,A\calC w_x}|\leq \|A\|\|w_x\|
}
and hence the bound \cref{eq:non-collinearity lower bound}.
\end{proof}

The following proposition is the symmetric version of \cref{lem:localized centers}. We perturb a given $A\in\calG_d^\azDIII$ to a $T\in\calG(\calL_d)$ that has certain cone-separation structure. The tricky point is that $T$ looses the symmetry condition. Nonetheless, we may recover it locally via rank-one perturbation. 

\begin{lem}\label{lem:diii localized centers}
Let $A\in\calG_d^\azDIII$ and $\ve>0$ be small enough (depending on $A$). 
Then there exists a spherically-proper sequence $\{x_k\}_{k\in\NN}$ of points on $\ZZ^d$, and a $T\in\calL_d$ with $\|A-T\|\leq C_{A} \ve$ and $A-T\in\calK(\calH_d)$ such that 
\eql{\label{eq:diii local symmetry}
-T\overline{T} \delta_{x_k} = \delta_{x_k}
}
for all $k\in\NN$, and if we define
\eql{\label{eq:diii interaction range of B}
    R_k := \supp(T \delta_{x_k}) \cup \{x_k\} \cup \{y_k\} \subset \ZZ^d
 }
where $y_k$ is the nearest-neighbor point to $x_k$, then $|R_k|<\infty$ for all $k$; the sets $R_k$ are pairwise disjoint; and satisfies the cone-separation property: for any closed disjoint dyadic cubes $I$ and $J$, there are finitely many $k$ such that $R_k\cap C_I\neq \varnothing$ and $R_k\cap C_J\neq \varnothing$.
\end{lem}
\begin{proof}
    Let $A\in\calG^\azDIII_d$ and let $\ve >0$ be small to be determined. We use \cref{lem:localized centers} to perturb $A$ by at most $\ve$ to a spherically-local $G$ with $A-G\in\calK(\calH_d)$  and produce a spherically-proper sequence of points $\{x_k\}_{k\in\NN}$ with disjoint finite subsets $R_k:=\supp(G\delta_{x_k})\cup \{x_k\}$ satisfying the cone-separation property.
    In particular, we may enlarge the island $R_k$
    \eq{
    R_k = \supp(G\delta_{x_k}) \cup \{x_k\} \cup \{y_k\}
    }
    by adding a nearest-neighbor point $y_k$ to $x_k$, while keeping spherical properness of $\cup_k\{x_k,y_k\}$ and the properties $R_k$ satisfies. 

In general, $G\notin \calG_d^\azDIII$ and $-G\overline{G}=\Id$ does not hold. Nonetheless, we can recover it locally along the sequence $\{x_k\}_{k\in\NN}$, while keeping the separation structure of $R_k$ and the invertibility and locality of operator. Let $x\in\ZZ^d$.
Consider the defect operator $D:=G\overline{G}+\Id$. Since $A-G$ is compact and $A\overline{A}+\Id=0$, it follows that $D$ is compact.
Consider the defect vector $r_x:= D\delta_x $. 
Since $\delta_x$ converges to zero weakly and $D$ is compact, the sequence $r_x$ converges in norm to zero as $\|x\|\to \infty$. Consider the defect vector along the spherically-proper sequence $\{x_k\}_{k\in\NN}$. Since $r_{x_k}\to 0$, pick a subsequence, still spherically-proper, such that
\eq{
\|r_{x_k}\| \leq \ve 2^{-k}.
}
We can do this without losing spherical properness by a diagonal selection argument over a countable basis $\{I_n\}$ of open sets on $\bS^{d-1}$: enumerate pairs $(n,k)\in \NN^2$, and at stage $(n,k)$ choose a point in the cone $C_{I_n}$ far enough so that $\|r_{x_k}\| \leq \ve 2^{-k}$, and also not previously used. Re-index to a single sequence. The inherited $R_{x_k}$ remain disjoint and retain cone-separation because we only pass to a subsequence. So, after relabeling, assume
\eq{
\sum_{k=1}^\infty \|r_{x_k}\|^2 \leq \frac{\ve^2}{3} .
}
For $x\in\ZZ^d$, let $\eta_x :=\Lambda^\perp_{\{x\}} \overline{G}\delta_x$. Using \cref{lem:diii non-collinearity}, we have 
\eq{
\|\eta_x\|=\dist(\overline{G}\delta_x,\CC\delta_x) = \dist(G\delta_x,\CC\delta_x)\geq \dist(A\delta_x,\CC\delta_x)-\ve \geq \frac{1}{2\|A\|}.
}
where we use $\|A-G\|\leq \ve$ and choose $\ve$ to be smaller than $1/(2\|A\|)$.
Define the rank-one operator 
\eq{
K_x = -\frac{1}{\|\eta_x\|^2} r_x\otimes \eta_x^\ast.
}
Then $K_x\delta_x =0$ since $\eta_x\perp \delta_x$, and
\eq{
K_x \overline{G}\delta_x = -r_x\left\langle\frac{\eta_x}{\|\eta_x\|^2},\overline{G}\delta_x\right\rangle = -r_x\left\langle\frac{\eta_x}{\|\eta_x\|^2},\Lambda_{\{x\}} \overline{G}\delta_x + \eta_x\right\rangle = -r_x.
}
Now set $K:=\sum_{k=1}^\infty K_{x_k}$ along the spherically-proper sequence $\{x_k\}_{k\in\NN}$. The series converges in Hilbert-Schmidt norm. Indeed, we have
\eq{
\norm{\sum_{k=m}^n K_{x_k}}_{\mathrm{HS}} = \sum_{j,k=m}^n \braket{r_{x_j},r_{x_k}} \left\langle\frac{\eta_{x_k}}{\|\eta_{x_k}\|^2},\frac{\eta_{x_j}}{\|\eta_{x_j}\|^2}\right\rangle = \sum_{k=m}^n \frac{\|r_{x_k}\|^2}{\|\eta_{x_k}\|^2}\leq \sum_{k=1}^\infty \frac{\|r_{x_k}\|^2}{\|\eta_{x_k}\|^2} \leq \frac{4}{3}\|A\|^2\ve^2
}
where in the first equality, we have crucially used the fact that $\{\eta_{x_k}\}_{k\in\NN}$ are pairwise orthogonal. Indeed, since $\supp (G\delta_{x_k})\subset R_k $ and hence $\supp (\eta_{x_k})\subset R_k$, this follows from the disjointness of $\{R_k\}_{k\in\NN}$. Then
\eq{
\|K\|\leq \|K\|_\mathrm{HS} \leq \frac{2 \|A\|}{\sqrt{3}}\ve.
}

Let $T:=G+K$. Then $A-T=(A-G)-K$ is compact, and $\|A-T\|\leq C_A\ve$ where $C_A=1+2\|A\|/\sqrt{3}$ if we choose $\ve\leq 1/(2\|A\|)$. We now verify \cref{eq:diii local symmetry}:
\eq{
-T\overline{T}\delta_{x_k} = -T\calC (G+K)\delta_{x_k} = -T\overline{G}\delta_{x_k} = -(G+K)\overline{G}\delta_{x_k} = -G\overline{G}\delta_{x_k} + r_{x_k} = \delta_{x_k}.
}

\end{proof}

Recall that we say a projection $P$ reduces a bounded operator $A$ if $A$ is invariant under $\im P$ and $\im P^\perp$.

\begin{prop}\label{prop:diii reduced by proper projection}
    If $A\in\calG^\azDIII_d$, then there exists a spherically-proper projection $P$ reducing the dimer operator $E$, and a homotopy
    \eq{
    A\sim_h W=PWP+P^\perp EP^\perp
    }
    in $\calG_d^\azDIII$.
\end{prop}
\begin{proof}
    Let $A\in\calG^\azDIII_d$ and let $\ve >0$ be small to be determined. We use \cref{lem:diii localized centers} to perturb $A$ by $\ve$ to a $T\in\calG(\calL_d)$ and produce a spherically-proper sequence of nearest-neighbor dimers $\{x_k,y_k\}_{k\in\NN}$ with the relation \cref{eq:diii local symmetry} and a sequence of disjoint finite subsets $R_k:=\supp(T\delta_{x_k})\cup \{x_k\}\cup \{y_k\}$ satisfying the cone-separation property. Let $P$ be $\Lambda_{\cup_k\{x_k,y_k\}}$.

    We construct a pinning operator $V\in\calG(\calL_d)$ such that $VT\overline{V}^{-1}$ agrees with the dimer operator $E$ on $\im P$.
    On the finite-dimensional space $\im\Lambda_{R_k}$, define $V_k$ to be any invertible operator that acts as
    \eq{
    V_k\delta_{x_k}= \delta_{x_k} ,\quad   V_k (T\delta_{x_k})= \delta_{y_k}. 
    }
    Define $V:= \oplus_k V_k \oplus \Id_{(\oplus_k\im \Lambda_{R_k})^\perp} \in\calG(\calL_d)$. We now verify that $V T\overline{V}^{-1}$ agrees with $E$ on $\im P$. Indeed, we have
    \eq{
    VT\overline{V}^{-1} \delta_{x_k} = VT\calC V^{-1}\delta_{x_k} = VT\delta_{x_k} = \delta_{y_k}
    }
    and
    \eq{
    VT\overline{V}^{-1} \delta_{y_k}  =VT\calC V^{-1}\delta_{y_k} = VT\calC T\delta_{x_k} = VT\overline{T}\delta_{x_k}=-V\delta_{x_k} = -\delta_{x_k}
    }
    where in the second to last equality we use \cref{eq:diii local symmetry}. Let $V$ be identity on $(\oplus \im \Lambda_{R_k})^\perp$. 
    Using argument exactly the same as in \cref{prop:local unitary reduced by geo proper proj}, the operator $V$ is spherically-local and we can construct a homotopy $V\sim_h\Id$ in $\calG(\calL_d)$.

    Since $\|A-T\|\leq\ve$ can be made arbitrarily small, and $A\in\calG_d^\azDIII$, it follows that $\sigma(-\overline{T}T)$ lies in a neighborhood of $1$ and avoids $(-\infty,0]$. Therefore, we can perform symmetrization \cref{eq:diii symmetrization} on the straight-line homotopy $t\mapsto (1-t)A+tT$ and get $A\sim_h \Psi(T)$ in $\calG_d^\azDIII$.
    Using the previously constructed homotopy $V\sim_h \Id$ in $\calG_d^\azDIII$ and \cref{lem:diii congruence}, we get $\Psi(T)\sim_h V\Psi(T)\overline{V}^{-1}=\Psi(VT\overline{V}^{-1})$ in $\calG_d^\azDIII$, where the last equality follows from \cref{lem:diii symmetrization}.
    Let $S:=\Psi(VT\overline{V}^{-1})$. We argue that $S$ agrees with $E$ on $\im P$. Indeed, let $Z:=VT\overline{V}^{-1}$. Since $Z$ agrees with $E$ on $\im P$, then $-\overline{Z}Z=\Id$ on $\im P$, and so is $(-\overline{Z}Z)^{-1/2}$. Thus $S=\Psi(Z)=Z(-\overline{Z}Z)^{-1/2}$ agrees with $E$ on $\im P$.
    
    With respect to the decomposition $\im P\oplus \im P^\perp$, we may then write
    \eq{
    S = \begin{bmatrix}
        E & X \\ 0 & B
    \end{bmatrix}
    }
    where $E,X,B$, viewed on $\calH_d$, are $PSP,PSP^\perp,P^\perp SP^\perp$. Using $S\overline{S}=-\Id$, we have the relations $E\overline{X}+X\overline{B}=0$ and $B\overline{B}=-\Id_{\im P^\perp}$. Let
    \eq{
    W_t:= \begin{bmatrix}
        \Id & t X\overline{B} \\0 & \Id
    \end{bmatrix}.
    }
    Then $W_t\in\calG(\calL_d)$ and we have 
    \eq{
    W_tS\overline{W}_t^{-1} = \begin{bmatrix}
        E & -tE\overline{X}B + X + tX\overline{B}B \\ 0 & B
    \end{bmatrix}.
    }
    Setting $t=1/2$, the off-diagonal part vanishes
    \eq{
    -\frac{1}{2}E\overline{X}B + X + \frac{1}{2}X\overline{B}B = X + X\overline{B}B = 0.
    }
    Thus we obtain $S\sim_h PEP+P^\perp S P^\perp$. At least point we may conclude the proposition by swapping $P$ and $P^\perp$. Alternatively, to keep the chosen sequence $\{x_k\}_{k\in\NN}$ as support, we may use the rotation $R_t$ in \cref{eq:rotation between P perp Q} to get $PEP+P^\perp S P^\perp\sim_h PSP+P^\perp EP^\perp$ in $\calG_d^\azDIII$. Indeed, the rotation $R_t$ is real $R_t=\overline{R}_t$, and hence the congruence homotopy stays within $\calG_d^\azDIII$ by \cref{lem:diii congruence}.
\end{proof}

\begin{lem}\label{lem:diii partial isometry for dimers}
    Let $E$ be the nearest-neighbor dimer operator in \cref{eq:diii dimer operator nearest neighbor}. Let $P,Q$ be spherically-proper projections that reduce $E$. Then there exists a spherically-local, real partial isometry $V$ that implements $P\sim Q$ and commutes with $E$. The conclusion also holds when we take $Q$ to be $\Id$.
\end{lem}
\begin{proof}
    Let $\{x_k,y_k\}_{k\in\NN}$ and $\{\tilde{x}_k,\tilde{y}_k\}_{k\in\NN}$ be the nearest-neighbor pairs for $P=\Lambda_{\cup_k \{x_k,y_k\}}$ and $Q=\Lambda_{\cup_k\{\tilde{x}_k,\tilde{y}_k\}}$. Analogously to \cref{prop:spherically-proper projections equivalent to identity}, we can pick a bijection $f:\NN\to\NN$ between the sequence of dimerization, such that if we define
    \eq{
    V\delta_{x_k} = \delta_{\tilde{x}_{f(k)}},\quad V\delta_{y_k}=\delta_{\tilde{y}_{f(k)}}
    }
    then $V$ is spherically-local. One readily verifies that $VE=EV$.
\end{proof}

Let us denote $\calG_{d,n}^\azDIII:=\Set{A\in\calG_n(\calL_d) | A^{-1}= -\overline{A}}$. Let $E_n\in\calG_{d,n}^\azDIII$ be direct sum of $n$-copies of $E$ in \cref{eq:diii dimer operator nearest neighbor}. 

\begin{prop}\label{prop:diii compress unitary homotopy from matrix algebra back to original algebra}
    Let $A,B\in\calG_d^\azDIII$. Suppose they take the form $A=PAP+P^\perp EP^\perp$ and $B=PBP+P^\perp EP^\perp$ for some spherically-proper $P$ that reduces $E$. If $A\oplus E_n\sim_h B\oplus E_n$ in $\calG_{d,n+1}^\azDIII$ for some $n$, then $A\sim_h B$ in $\calG_d^\azDIII$.
\end{prop}
\begin{proof}    
    Let $P_0:=P$. Since $P_0$ is spherically-proper, so is $P_0^\perp$. Using \cref{lem:split spherically-proper projections}, we may decompose $P_0^\perp$ into $n+1$ spherically-proper projections
    \eq{
        P_0^\perp = Q_0+ P_1 + \dots + P_{n-1} + P_n.
    }
    Using \cref{lem:diii partial isometry for dimers}, we have $Q_0\sim P_0^\perp$ and $P_k\sim \Id$ for all $k=1,\dots,n$, implemented by real, spherically-local partial isometries $V_0,V_1,\dots,V_n$ such that $P_0^\perp = V_0^\ast V_0$, $Q_0=V_0V_0^\ast$, and  $P_k=V_kV_k^\ast$ and $\Id = V_k^\ast V_k$ for $k=1,\dots,n$. Moreover, $V_kE=EV_k$ for $k=0,\dots,n$. Consider
    \eq{
    V:= [P_0+V_0,V_1,\dots,V_n]\in M_{1,n+1}(\calL_d^\RR).
    }
    Direct computation shows that $V^\ast V=\Id_{n+1}$ and $VV^\ast = \Id$ and $VE_{n+1}=EV$.

    Let $W_t$ be the path for $A\oplus E_{n}\sim_h B\oplus E_n$ in $\calG_{d,n+1}^\azDIII$. Then $VW_t\overline{V}^{-1} = VW_tV^\ast$ gives the path $A\sim_h B$ in $\calG_d^\azDIII$.
\end{proof}

\begin{proof}[Proof of \cref{thm:classification for real classes} for class $\mathrm{DIII}$]
    Using \cref{prop:clifford description of AZ classes shperically local} and \cref{prop:computation for van daele K group of real spherically local algebra}, it suffices to show that the map
    \eql{\label{eq:diii pi0 to K group map}
    \pi_0(\calSU_o(\calL_d^\RR\otimes \Cli_{0,2})) \to \rmDK_e(\calL_d^\RR\otimes \Cli_{0,2})
    }
    which sends path-connected component to its van Daele $K$-group, is a bijection.  
    With specific representation of $\Cli_{0,2}$, the set $\calSU_o(\calL_d^\RR\otimes \Cli_{0,2})$ takes the form in \cref{eq:diii SUo}.
    Here, we choose $e$ to be the element
    \eq{
    e:=\begin{bmatrix}
        0 & E \\ -E & 0
    \end{bmatrix}
    }
    where $E$ is the nearest-neighbor dimer operator. For $\rmDK_e(\calL_d^\RR\otimes \Cli_{0,2})$ to make sense, we need to show that $e\sim_h-e$ in $\calSU_o(\calL_d^\RR\otimes \Cli_{0,2})$. To that end, for each dimer $(x,y)$, with respect to $\CC\delta_x\oplus \CC\delta_y$, we consider the rotation
    \eq{
    \begin{bmatrix}
        \cos t & -\sin t \\ \sin t & \cos t
    \end{bmatrix}
    \begin{bmatrix}
        0 & -1 \\ 1 & 0
    \end{bmatrix}
    \begin{bmatrix}
        \cos t & -\sin t \\ \sin t & \cos t
    \end{bmatrix}
    }
    for $t\in [0,\pi/2]$. Apply the rotation to each dimer, the homotopy stays spherically-local. We get $E\sim_h -E$ in $\calU_d^\azDIII$ and hence, using \cref{eq:diii SUo}, we have $e\sim_h -e$.  
    
    We show that if $U,V\in\calSU_o(\calL_d^\RR\otimes \Cli_{0,2})$ has the same $[U]_e=[V]_e$ van Daele $K$-group, then $U\sim_h V$ in $\calSU_o(\calL_d^\RR\otimes \Cli_{0,2})$. Write
    $
    U=\begin{bmatrix}
        0 & X \\ -\overline{X} & 0
    \end{bmatrix}$ and $ V = \begin{bmatrix}
        0 & Y \\ -\overline{Y} & 0
    \end{bmatrix}
    $
    for $X,Y\in\calU_d^\azDIII$. Using \cref{prop:diii reduced by proper projection}, there exist spherically-proper projections $P,Q$ reducing $E$ such that $X\sim_h A=PAP+P^\perp EP^\perp$ and $Y\sim_h \widetilde{B}=Q\widetilde{B}Q+Q^\perp EQ^\perp$ in $\calG_d^\azDIII$.
    In fact, we may choose $Q$ to be orthogonal to $P$. Indeed, after constructing $P$, we may construct $Q$ within $P^\perp$, where the construction of spherically-proper sequence of points traced back to \cref{lem:contained in annulus}. Using \cref{eq:rotation between P perp Q}, we then have $Y\sim_h Q\widetilde{B}Q+Q^\perp EQ^\perp \sim_h B=PBP+P^\perp EP^\perp$ in $\calG_d^\azDIII$, supported on the same $P$ as in $A$. Using \cref{lem:diii lift invertible to unitary} and \cref{eq:diii polar map}, there exists $S,T\in\calU_d^\azDIII$ such that $A\sim_h S=PSP+P^\perp EP^\perp$ and $B\sim_h T=PTP+P^\perp EP^\perp$ in $\calG_d^\azDIII$.
    Using \cref{lem:diii lift invertible to unitary} again, we get $X\sim_h S$ and $Y\sim_h T$ in $\calU_d^\azDIII$.
    
    Since $[U]_e=[V]_e$, it follows that
    \eq{
    \begin{bmatrix}
        0 & S \\ -\overline{S} & 0
    \end{bmatrix} \oplus \begin{bmatrix}
        0 & E \\ -E & 0
    \end{bmatrix}_n \sim_h \begin{bmatrix}
        0 & T \\ -\overline{T} & 0
    \end{bmatrix} \oplus \begin{bmatrix}
        0 & E \\ -E & 0
    \end{bmatrix}_n
    }
    in $\calSU_o(M_{n+1}(\calL_d^\RR\otimes \Cli_{0,2}))$. Equivalently, under a conjugation by some elementary matrix, we have
    \eq{
    S\oplus E_n \sim_h T\oplus E_n
    }
    in $\calU_{d,n+1}^\azDIII:= \Set{U\in\calU_{n+1}(\calL_d) | U^\ast = -\overline{U}}$. Now we use \cref{prop:diii compress unitary homotopy from matrix algebra back to original algebra} and get $S\sim_h T$ in $\calG_d^\azDIII$, and then we use \cref{lem:diii lift invertible to unitary} to get $S\sim_h T$ in $\calU_d^\azDIII$. This provides the desired homotopy $U\sim_h V$ in $\calSU_o(\calL_d^\RR\otimes \Cli_{0,2})$.
    
    We show that the map \cref{eq:diii pi0 to K group map} is surjective. Let $\xi\in\rmDK_e(\calL_d^\RR\otimes \Cli_{0,2})$. Then by construction there exists $U\in\calSU_o(M_n(\calL_d^\RR\otimes \Cli_{0,2}))$ for some $n$ such that $\xi = [U]_e$. The goal is to find $W\in \calSU_o(\calL_d^\RR\otimes \Cli_{0,2})$ such that $[W]_e=[U]_e$.
    Write $U=\begin{bmatrix} 0 & Z \\ -\overline{Z} & 0 \end{bmatrix}$ for $Z\in \calU_{d,n}^\azDIII$.
    Let $\ZZ^d=\bigcup_{k\in\NN}\{x_k,y_k\}$ be an enumeration of dimers. Let $\NN=I_1\cup I_2\cup \dots\cup I_n$ be a partition of index set in to $n$ disjoint infinite subsets such that that for each $j$, the associated set of lattice sites $F_j:=\bigcup_{k\in I_j}\{x_k,y_k\}$ is spherically-proper. This is possible because every cone contains infinitely many dimers; distribute dimers among $n$ bins so that each bin receives infinitely many dimers in every cone. Let $P_j:=\Lambda_{F_j}$. Using \cref{lem:diii partial isometry for dimers}, there are $V_1,\dots,V_n\in\calL_d^\RR$ such that $V_j^\ast V_j=\Id$ and $V_jV_j^\ast = P_j$ and $V_jE=EV_j$.
    Let $V:=(V_1,\dots,V_n)\in M_{1,n}(\calL_d^\RR)$. Then $VV^\ast = \Id$ and $V^\ast V=\Id_n$ and $VE_n=EV$.
    Consider $VZV^\ast \in\calU_d^\azDIII$. We argue that
    \eq{
    Z\oplus E \sim_h E_n\oplus VZV^\ast
    }
    in $\calU_{d,n+1}^\azDIII$. Indeed, this is achieved by the rotation
    \eq{
    \begin{bmatrix}
        \cos t \Id_n & -\sin t V^\ast \\ \sin t V & \cos t \Id
    \end{bmatrix}
    \begin{bmatrix}
        Z & 0 \\ 0 & E
    \end{bmatrix}
    \begin{bmatrix}
        \cos t \Id_n & \sin t V^\ast \\ -\sin t V & \cos t \Id        
    \end{bmatrix}
    }
    for $t\in[0,\pi/2]$. Let $W=\begin{bmatrix}0 & VZV^\ast \\ - \overline{VZV^\ast} & 0 \end{bmatrix} \in\calSU_o(\calL_d^\RR\otimes \Cli_{0,2})$. Then we have $U\oplus e\sim_h e_n\oplus W$ in $\calSU_o(M_{n+1}(\calL_d^\RR\otimes \Cli_{0,2}))$ and hence $[U]_e=[W]_e$.
\end{proof}

\begin{lem}\label{lem:diii lift invertible to unitary}
    If $U,V\in\calU_d^\azDIII$ and $U\sim_h V$ in $\calG_d^\azDIII$, then they are also homotopic in $\calU_d^\azDIII$. If $G\in\calG_d^\azDIII$, then there exists $U\in\calU_d^\azDIII$ such that $G\sim_h U$ in $\calG_d^\azDIII$.
\end{lem}
\begin{proof}
    The proof is based on the polar decomposition. We show that
    \eql{\label{eq:diii polar map}
    \calG_d^\azDIII \ni A\mapsto A|A|^{-1} \in \calU_d^\azDIII
    }
    maps into $\calU_d^\azDIII$. Let $A\in\calG_d^\azDIII$ and consider its polar decomposition $A=U|A|$. We first show the identity $|A^{-1}|=U|A|^{-1}U^\ast$ which is true regardless of symmetry condition. Indeed, we have
    \eq{
    |A^{-1}|=(AA^\ast)^{-1/2} = (U|A||A|U^\ast)^{-1/2} = U|A|^{-1}U^\ast.
    }
    So the polar decomposition of $A^{-1}$ is
    \eq{
    U^\ast |A^{-1}|=U^\ast( U|A|^{-1}U^\ast) = |A|^{-1}U^\ast = A^{-1}.
    }
    Now
    \eq{
    \calC A|A|^{-1}\calC = \calC A \calC \calC|A|^{-1}\calC = -A^{-1}|A^{-1}|^{-1} = -U^\ast = -(A|A|^{-1})^{-1}
    }
    which shows that $A|A|^{-1}\in\calU_d^\azDIII$.
    The conclusion of the lemma follows from continuity of the map \cref{eq:diii polar map}.

    To prove the other assertion, if $G\in\calG_d^\azDIII$, we consider the path $t\mapsto G|G|^{-t}$ for $t\in[0,1]$. 
\end{proof}

\begin{rem}
    We expect \cref{lem:diii lift invertible to unitary} to be true, as it should follow from the identification \cref{eq:diii SUo} and \cite[Proposition 2.5]{daele1988kI}.
\end{rem}

\appendix
\section{Odd structures in Clifford algebra}\label{sec:odd structures in clifford algebra}

To analyze the odd structure, we first consider $\Cli_{1,0} \cong \RR \oplus \RR$ with elementwise operation. Its $\ZZ_2$ grading is given by $\Cli_{1,0}^0 \cong \Set{(a,a) | a\in \RR}$ and $\Cli_{1,0}^1 \cong \Set{(a,-a) | a\in \RR}$, realized by mapping the generator $E_1 \mapsto (1,-1)$ while the identity maps to $(1,1)$. 

For $\Cli_{1,1} \cong M_2(\RR)$, the grading distinguishes between diagonal and off-diagonal matrices: the even subalgebra $\Cli_{1,1}^0$ consists of diagonal matrices $\text{diag}(a,d)$, while the odd subspace $\Cli_{1,1}^1$ contains off-diagonal matrices.
We verify this via generators 
\eq{
E_1 \mapsto \begin{bmatrix} 0 & 1 \\ 1 & 0 \end{bmatrix},\quad E_2 \mapsto \begin{bmatrix} 0 & -1 \\ 1 & 0 \end{bmatrix}.
}

For $\Cli_{0,1} \cong \CC$, the grading is given by real and imaginary components: $\Cli_{0,1}^0 \cong \RR$ and $\Cli_{0,1}^1 \cong i\RR$, achieved by sending $E_1 \mapsto i$. 

The algebra $\Cli_{0,2} \cong \HH$ can be represented by complex matrices 
\eq{
\Cli_{0,2} \cong \Set{\begin{bmatrix}a & b \\ -\overline{b} & \overline{a}\end{bmatrix}| a,b\in\CC},\quad \Cli_{0,2}^0 \cong \Set{\begin{bmatrix} a & 0 \\ 0 & \bar{a} \end{bmatrix} | a\in\CC},\quad \Cli_{0,2}^1 \cong \Set{\begin{bmatrix} 0 & b \\ -\bar{b} & 0 \end{bmatrix} | b\in\CC}
}
This structure is supported by generators 
\eq{
E_1 \mapsto \begin{bmatrix} 0 & 1 \\ -1 & 0 \end{bmatrix},\quad E_2 \mapsto \begin{bmatrix} 0 & i \\ i & 0 \end{bmatrix}
}
and the grading automorphism is given by
\eq{
\begin{bmatrix}a & b \\ -\overline{b} & \overline{a}\end{bmatrix}\mapsto \begin{bmatrix}a & -b \\ \overline{b} & \overline{a}\end{bmatrix}.
}

We have $\Cli_{0,3} \cong \HH \oplus \HH$ with elementwise operation and with grading $\Cli_{0,3}^0 \cong \Set{(a,a) | a\in \HH}$ and $\Cli_{0,3}^1 \cong \Set{(a,-a) | a\in \HH}$, established by generators $E_1 \mapsto (i,-i)$, $E_2 \mapsto (j,-j)$, and $E_3 \mapsto (k,-k)$.

We have $\Cli_{0,4} \cong M_2(\HH)$ which mirrors the matrix structure of $\Cli_{1,1}$, where $\Cli_{0,4}^0$ comprises diagonal quaternion matrices and $\Cli_{0,4}^1$ comprises off-diagonal ones. 
This is verified by mapping generators to $2\times 2$ matrices with quaternion entries: 
\eq{
E_1 \mapsto \begin{bmatrix} 0 & 1 \\ -1 & 0 \end{bmatrix},\quad E_2 \mapsto \begin{bmatrix} 0 & i \\ i & 0 \end{bmatrix},\quad E_3 \mapsto \begin{bmatrix} 0 & j \\ j & 0 \end{bmatrix},\quad E_4 \mapsto \begin{bmatrix} 0 & k \\ k & 0 \end{bmatrix}
}
and the grading automorphism is inner and given by reversing the sign of off diagonal elements 
\eq{
\operatorname{Ad}_{\sigma_3}:A\mapsto \sigma_3 A \sigma_3^\ast
}
where $A\in M_2(\HH)$ and $\sigma_3$ is the Pauli matrix in the $z$ direction.

Now consider $\Cli_{0,5}$. We claim that $\Cli_{0,5}\cong M_4(\CC)$ and the $\ZZ_2$ grading is given by
    \eql{\label{eq:Cli 0 5 grading}
     \Cli^0_{0,5}\cong \Set{\begin{bmatrix}
        A & B \\ -\overline{B} & \overline{A}
    \end{bmatrix} | A\in M_2(\CC)},\quad
     \Cli^1_{0,5}\cong \Set{\begin{bmatrix}
        A & B \\ \overline{B} & -\overline{A}
    \end{bmatrix} | B\in M_2(\CC)}.
    }
We have $\Cli_{0,5}\cong \Cli_{0,4}\hat{\otimes}\Cli_{0,1}\cong \Cli_{0,4}\otimes \Cli_{0,1}$ where the first isomorphism uses \cite[Theorem 3.10]{karoubi2009k} and the second isomorphism uses \cite[Proposition 14.5.1]{blackadar1998k} with the fact that the grading automorphism on $\Cli_{0,4}$ is inner. Then $\Cli_{0,4}\otimes \Cli_{0,1}\cong M_2(\HH)\otimes \CC \cong M_2(\RR)\otimes \HH\otimes \CC$ where the grading automorphism is $\operatorname{Ad}_{\sigma_3}\otimes \operatorname{id}\otimes \calC$. We have $\HH\otimes \CC\cong M_2(\CC)$ where the isomorphism is given by
\eq{
1\otimes z_0 + i\otimes z_1+j\otimes z_2 + k\otimes z_3 \mapsto \begin{bmatrix}
    z_0 + iz_1 & z_2 + iz_3 \\
    -z_2 + iz_3 & z_0-iz_1
\end{bmatrix}
}
where $z_i\in\CC$. The automorphism $\operatorname{id}\otimes \calC$ applied to $\HH\otimes \CC\cong M_2(\CC)$ gives
\eq{
\begin{bmatrix}
    z_0 + iz_1 & z_2 + iz_3 \\
    -z_2 + iz_3 & z_0-iz_1
\end{bmatrix} = \begin{bmatrix}
    a & b \\ c& d
\end{bmatrix}\xmapsto{\operatorname{id}\otimes \calC} \begin{bmatrix}
    \bar{d} & -\bar{c} \\ -\bar{b} & \bar{a}
\end{bmatrix}
}
since $\calC$ puts complex conjugation on each $z_i$. The automorphism $\operatorname{Ad}_{\sigma_3}\otimes \operatorname{id}\otimes \calC$ applied to $\Cli_{0,5}\cong M_2(\RR)\otimes \HH\otimes \CC \cong M_2(\RR)\otimes M_2(\CC)\cong M_4(\CC)$ then gives
\eq{
\operatorname{Ad}_{\sigma_3}\otimes \operatorname{id}\otimes \calC\begin{bmatrix}
    a_1 & a_2 & b_1 & b_2 \\
    a_3 & a_4 & b_3 & b_4 \\
    c_1 & c_2 & d_1 & d_2 \\
    c_3 & c_4 & d_3 & d_4
\end{bmatrix} &= \operatorname{Ad}_{\sigma_3}\otimes \operatorname{id}\otimes \operatorname{id}
\begin{bmatrix}
    \bar{a}_4 & -\bar{a}_3 & \bar{b}_4 & -\bar{b}_3 \\
    -\bar{a}_2 & \bar{a}_1 & -\bar{b}_2 & \bar{b}_1 \\
    \bar{c}_4 & -\bar{c}_3 & \bar{d}_4 & -\bar{d}_3 \\
    -\bar{c}_2 & \bar{c}_1 & -\bar{d}_2 & \bar{d}_1
\end{bmatrix}\\ &= 
\begin{bmatrix}
    \bar{a}_4 & -\bar{a}_3 & -\bar{b}_4 & \bar{b}_3 \\
    -\bar{a}_2 & \bar{a}_1 & \bar{b}_2 & -\bar{b}_1 \\
    -\bar{c}_4 & \bar{c}_3 & \bar{d}_4 & -\bar{d}_3 \\
    \bar{c}_2 & -\bar{c}_1 & -\bar{d}_2 & \bar{d}_1
\end{bmatrix}.
}
Thus, the even and odd subspaces are
\eq{
\Cli_{0,5}^0&=\Set{\begin{bmatrix}
    a_1 & a_2 & b_1 & b_2 \\
    -\bar{a}_2 & \bar{a}_1 & \bar{b}_2 & -\bar{b}_1 \\
    c_1 & c_2 & d_1 & d_2 \\
    \bar{c}_2 & -\bar{c}_1 & -\bar{d}_2 & \bar{d}_1
\end{bmatrix} | a_i,b_i,c_i,d_i\in\CC} \\
\Cli_{0,5}^1&=\Set{\begin{bmatrix}
    a_1 & a_2 & b_1 & b_2 \\
    \bar{a}_2 & -\bar{a}_1 & -\bar{b}_2 & \bar{b}_1 \\
    c_1 & c_2 & d_1 & d_2 \\
    -\bar{c}_2 & \bar{c}_1 & \bar{d}_2 & -\bar{d}_1
\end{bmatrix} | a_i,b_i,c_i,d_i\in\CC}
}
After applying the basis change $\begin{bmatrix}
         1 & 0 & 0 & 0 \\
         0 & 0 & 0 & 1 \\
         0 & 1 & 0 & 0 \\
         0 & 0 & 1 & 0
    \end{bmatrix}$, which swaps second and third basis, and then second and forth basis, we obtain
\eq{
\Cli_{0,5}^0&=\Set{\begin{bmatrix}
    a_1 & b_2 & a_2 & b_1 \\
    \bar{c}_2 & \bar{d}_1 & -\bar{c}_1 & -\bar{d}_2 \\
    -\bar{a}_2 & -\bar{b}_1 & \bar{a}_1 & \bar{b}_2 \\
    c_1 & d_2 & c_2 & d_1
\end{bmatrix} | a_i,b_i,c_i,d_i\in\CC} \\
\Cli_{0,5}^1&=\Set{\begin{bmatrix}
    a_1 & b_2 & a_2 & b_1 \\
    -\bar{c}_2 & -\bar{d}_1 & \bar{c}_1 & \bar{d}_2 \\
    \bar{a}_2 & \bar{b}_1 & -\bar{a}_1 & -\bar{b}_2 \\
    c_1 & d_2 & c_2 & d_1
\end{bmatrix} | a_i,b_i,c_i,d_i\in\CC}
}
which are exactly \cref{eq:Cli 0 5 grading}.

Consider the algebra $\Cli_{0,6}$. We claim that
\eq{
\Cli_{0,6}\cong \Set{\begin{bmatrix}
    A & B \\ \overline{B} & \overline{A}
\end{bmatrix} | A,B\in M_4(\CC)}
}
and the $\ZZ_2$ grading with respect to the isomorphism is given by
\eql{\label{eq:Cli 0 6 grading}
 \Cli^0_{0,6}\cong \Set{\begin{bmatrix}
    A & 0 \\ 0 & \overline{A}
\end{bmatrix} | A\in M_4(\CC)},\quad
 \Cli^1_{0,6}\cong \Set{\begin{bmatrix}
    0 & B \\ \overline{B} & 0
\end{bmatrix} | B\in M_4(\CC)}.
}
We have $\Cli_{0,6}\cong \Cli_{0,4}\hat{\otimes}\Cli_{0,2}\cong \Cli_{0,4}\otimes \Cli_{0,2}$ where the first isomorphism uses \cite[Theorem 3.10]{karoubi2009k} and the second isomorphism uses \cite[Proposition 14.5.1]{blackadar1998k} with the fact that the grading automorphism on $\Cli_{0,4}$ is inner. 
In particular, the natural grading automorphism on  $\Cli_{0,6}\cong \Cli_{0,4}\otimes \Cli_{0,2}$ is then given by $\mu\otimes \tau$ where $\mu,\tau$ are the automorphisms on $\Cli_{0,4}$ and $\Cli_{0,2}$, respectively, the explicit forms provided in the previous paragraph.
Recall
\eq{
\Cli_{0,2}\cong \HH \cong \Set{\begin{bmatrix}a & b \\ -\overline{b} & \overline{a}\end{bmatrix}| a,b\in\CC}
}
and $\Cli_{0,4}\cong M_2(\HH)\cong M_2(\RR)\otimes \HH.$ Now
\eq{
\HH\otimes \HH &\cong \Set{\begin{bmatrix}
    ac & bc & ad & bd \\
    -\bar{b}c & \bar{a}c & -\bar{b}d & \bar{a}d \\
    -a\bar{d} & -b\bar{d} & a\bar{c} & b\bar{c} \\
    \bar{b}\bar{d} & -\bar{a}\bar{d} & -\bar{b}\bar{c} & \bar{a}\bar{c}
\end{bmatrix} : a,b,c,d\in\CC}\\ &\cong \Set{\begin{bmatrix}
    ac & bc & bd & -ad \\
    -\bar{b}c & \bar{a}c & \bar{a}d & \bar{b}d \\
    \bar{b}\bar{d} & -\bar{a}\bar{d} & \bar{a}\bar{c} & \bar{b}\bar{c} \\
    a\bar{d} & b\bar{d} & -b\bar{c} & a\bar{c}
\end{bmatrix} : a,b,c,d\in\CC}
}
where we applied the basis change $\begin{bmatrix}
    1 & 0 & 0 & 0 \\
    0 & 1 & 0 & 0 \\
    0 & 0 & 0 & 1 \\
    0 & 0 & -1 & 0
\end{bmatrix}$. We thus obtain
\eq{
\HH\otimes \HH \cong \Set{\begin{bmatrix}
    A & B \\ \overline{B} & \overline{A}
\end{bmatrix} | A,B\in M_2(\CC)}.
}
The automorphism on $\HH\otimes \HH$ is $\operatorname{id}\otimes \tau$, where $\tau$ reverses the sign of number $d$ above; or equivalently
\eq{
\begin{bmatrix}
    A & B \\ \overline{B} & \overline{A}
\end{bmatrix} \mapsto \begin{bmatrix}
    A & -B \\ -\overline{B} & \overline{A}
\end{bmatrix} .
}
Now
\eq{
\Cli_{0,6}\cong M_2(\RR) \otimes \Set{\begin{bmatrix}
    A & B \\ \overline{B} & \overline{A}
\end{bmatrix} | A,B\in M_2(\CC)} \cong \Set{\begin{bmatrix}
    A_1 & B_1 & A_2 & B_2 \\
    \overline{B}_1 & \overline{A}_1 & \overline{B}_2 & \overline{A}_2 \\
    A_3 & B_3 & A_4 & B_4 \\
    \overline{B}_3 & \overline{A}_3 & \overline{B}_4 & \overline{A}_4
\end{bmatrix} | A_i,B_i\in M_2(\CC)}.
}
After applying the basis change $\begin{bmatrix}
     I & 0 & 0 & 0 \\
     0 & 0 & 0 & I \\
     0 & I & 0 & 0 \\
     0 & 0 & I & 0
\end{bmatrix}$, which swaps second and third basis, and then second and forth basis, we obtain
\eq{
\Cli_{0,6}\cong 
\Set{\begin{bmatrix}
    A_1 & B_2 & B_1 & A_2 \\
    \overline{B}_3 & \overline{A}_4 & \overline{A}_3 & \overline{B}_4 \\
    \overline{B}_1 & \overline{A}_2 & \overline{A}_1 & \overline{B}_2 \\
    A_3 & B_4 & B_3 & A_4
\end{bmatrix} | A_i,B_i\in M_2(\CC)} \cong \Set{\begin{bmatrix}
    U & V \\ \overline{V} & \overline{U}
\end{bmatrix} | U,V\in M_4(\CC)}.
}
The grading automorphism $\Cli_{0,6}\cong M_2(\RR)\otimes \HH\otimes \HH$ is given by $\mu\otimes \operatorname{id}\otimes{\tau}$ where $\mu$ reverses the sign of off-diagonal entries in $M_2(\RR)$, or the $A_2,B_2,A_3,B_3$ entries above. Therefore
\eq{
\mu\otimes \operatorname{id}\otimes{\tau} \begin{bmatrix}
    A_1 & B_2 & B_1 & A_2 \\
    \overline{B}_3 & \overline{A}_4 & \overline{A}_3 & \overline{B}_4 \\
    \overline{B}_1 & \overline{A}_2 & \overline{A}_1 & \overline{B}_2 \\
    A_3 & B_4 & B_3 & A_4
\end{bmatrix} &=  \operatorname{id}\otimes \operatorname{id}\otimes \tau 
\begin{bmatrix}
    A_1 & -B_2 & B_1 & -A_2 \\
    -\overline{B}_3 & \overline{A}_4 & -\overline{A}_3 & \overline{B}_4 \\
    \overline{B}_1 & -\overline{A}_2 & \overline{A}_1 & -\overline{B}_2 \\
    -A_3 & B_4 & -B_3 & A_4
\end{bmatrix}\\ &= 
\begin{bmatrix}
    A_1 & B_2 & -B_1 & -A_2 \\
    \overline{B}_3 & \overline{A}_4 & -\overline{A}_3 & -\overline{B}_4 \\
    -\overline{B}_1 & -\overline{A}_2 & \overline{A}_1 & \overline{B}_2 \\
    -A_3 & -B_4 & B_3 & A_4
\end{bmatrix},
}
that is
\eq{
\begin{bmatrix}
    U & V \\ \overline{V} & \overline{U}
\end{bmatrix}\mapsto \begin{bmatrix}
    U & -V \\ -\overline{V} & \overline{U}
\end{bmatrix}.
}
It is then clear that the even and odd subspaces are given by \cref{eq:Cli 0 6 grading}.

\begingroup
\let\itshape\upshape
\printbibliography
\endgroup
\end{document}